\documentclass[a4paper,11pt,twoside]{ThesisStyle}

\usepackage{amsmath,amsthm,amssymb}
\usepackage{bbm}
\usepackage[T1]{fontenc}
\usepackage[utf8]{inputenc}
\usepackage[english]{babel}
\usepackage{complexity}
\usepackage{tikz}
\usepackage{xspace}
\usepackage{pdfpages}

\usepackage[left=1.1in,right=1.1in,top=1.1in,bottom=1.1in,includefoot,includehead,headheight=13.6pt]{geometry}

\usepackage{silence}

\usepackage{aecompl}
\usepackage{url}

\usepackage[printonlyused,withpage]{acronym}

\graphicspath{{.}{images/}}

\usepackage{color}
\definecolor{linkcol}{rgb}{0,0,0.4}
\definecolor{citecol}{rgb}{0.5,0,0}

\usepackage[nottoc, notlof, notlot]{tocbibind}
\usepackage{minitoc}
\usepackage{rotating}                    
\usepackage{fancyhdr}                    

\let\headruleORIG\headrule
\renewcommand{\headrule}{\color{black} \headruleORIG}

\usepackage{colortbl}
\arrayrulecolor{black}

\fancypagestyle{plain}{
  \fancyhead{}
  \fancyfoot{}
  
}

\usepackage[ruled,vlined, linesnumbered]{algorithm2e}

\makeatletter

\def\cleardoublepage{\clearpage\if@twoside \ifodd\c@page\else%
  \hbox{}%
  \thispagestyle{empty}
  \newpage%
  \if@twocolumn\hbox{}\newpage\fi\fi\fi}

\makeatother

{%

\hrulefill
\vspace*{0.5cm}%
\end{minipage}
}

\let\minitocORIG\minitoc
\renewcommand{\minitoc}{\minitocORIG \vspace{1.5em}}

\usepackage{subfigure}
\usepackage{multirow}
{ \begin{list}%
	{$\bullet$}%
	{\setlength{\labelwidth}{25pt}%
	 \setlength{\leftmargin}{30pt}%
	 \setlength{\itemsep}{\parsep}}}%
{ \end{list} }

\newtheorem{definition}{Definition}
\renewcommand{\epsilon}{\varepsilon}

\ifpdf
  \usepackage[pagebackref,hyperindex=true]{hyperref}
\else
  \usepackage[dvipdfm,pagebackref,hyperindex=true]{hyperref}
\fi

\renewcommand*{\backref}[1]{}
\renewcommand*{\backrefalt}[4]{%
\ifcase #1 %
(Not cited.)%
\or
(Cited on page~#2.)%
\else
(Cited on pages~#2.)%
\fi}

\hypersetup
{
bookmarksopen=true,
pdftitle="Manuscript title",
pdfauthor="Your name",
pdfsubject="Manuscript topic in a few words", 
pdfmenubar=true, 
pdfhighlight=/O, 
colorlinks=true, 
pdfpagemode=UseNone, 
pdfpagelayout=SinglePage, 
pdffitwindow=true, 
linkcolor=linkcol, 
citecolor=citecol, 
urlcolor=linkcol 
}

\newcommand{\enum}[1]{\textsc{Enum}\smash{\cdot}#1}
\newcommand{\test}[1]{\textsc{Check}\smash{\cdot}#1}

\newcommand{\ext}[1]{\textsc{ExtSol}\smash{\cdot}#1}
\newcommand{\jsol}[1]{\textsc{RankedSol}\smash{\cdot}#1}
\newcommand{\Output}{\mathtt{Output}}
\newcommand{\OutputP}{\mathrm{OutputP}}
\newcommand{\EnumP}{\mathrm{EnumP}}
\newcommand{\IncP}{\mathrm{IncP}}
\newcommand{\DelayP}{\mathrm{DelayP}}
\newcommand{\NextP}{\mathrm{NextP}}
\newcommand{\SDelayP}{\mathrm{SDelayP}}
\newcommand{\QueryP}{\mathrm{QueryP}}
\newcommand{\strongpdelay}{{strong polynomial delay}}
\newcommand{\kSAT}[1]{#1\text{-}\SAT}

\newcommand{\ETH}{\mathsf{ETH}}
\newcommand{\SETH}{\mathsf{SETH}}
\newcommand{\UIncP}{\mathrm{UsualIncP}}
\newcommand{\enumDNF}{$\enum{\textsc{DNF}}$}

\newcommand{\cF}{\ensuremath{\mathcal{F}}}

\newcommand{\ccS}{\ensuremath{\mathcal{S}}}

\newcommand{\Next} {\ensuremath{\text{next}}\xspace}
\newcommand{\mot}[1] {\textbf{#1}\xspace}
\newcommand{\alphabet} {\ensuremath{\mathcal{A}}\xspace}

\newcommand{\setofmotifs} {\ensuremath{\mathcal{M}}\xspace}

\newcommand{\Cl}{\ensuremath{Cl}}

\newcommand{\EnumClo}{\textsc{EnumClosure}}

\newcommand{\Mem}{\textsc{Membership}}

\newcommand{\PIT}{\textsc{PIT}}

\tikzstyle{label}=[draw=black,fill=white,text=black,circle]%
\tikzstyle{labelsat}=[draw=black,fill=lightgray,text=black,circle]%
\tikzstyle{centre}=[draw=black,fill=gray,text=white,circle]
\tikzstyle{noeud}=[draw=black,fill=white,circle]

\newtheorem{theorem}{Theorem}[chapter]
\newtheorem{proposition}{Proposition}[chapter]
\newtheorem{lemma}{Lemma}[chapter]
\newtheorem{corollary}{Corollary}[chapter]
\newtheorem{openproblem}{\bf Open problem}[chapter]

\theoremstyle{definition}
\newtheorem{conjecture}{Conjecture}[section]
\newtheorem{example}{Example}

\begin{document}

\includepdf[pages=-]{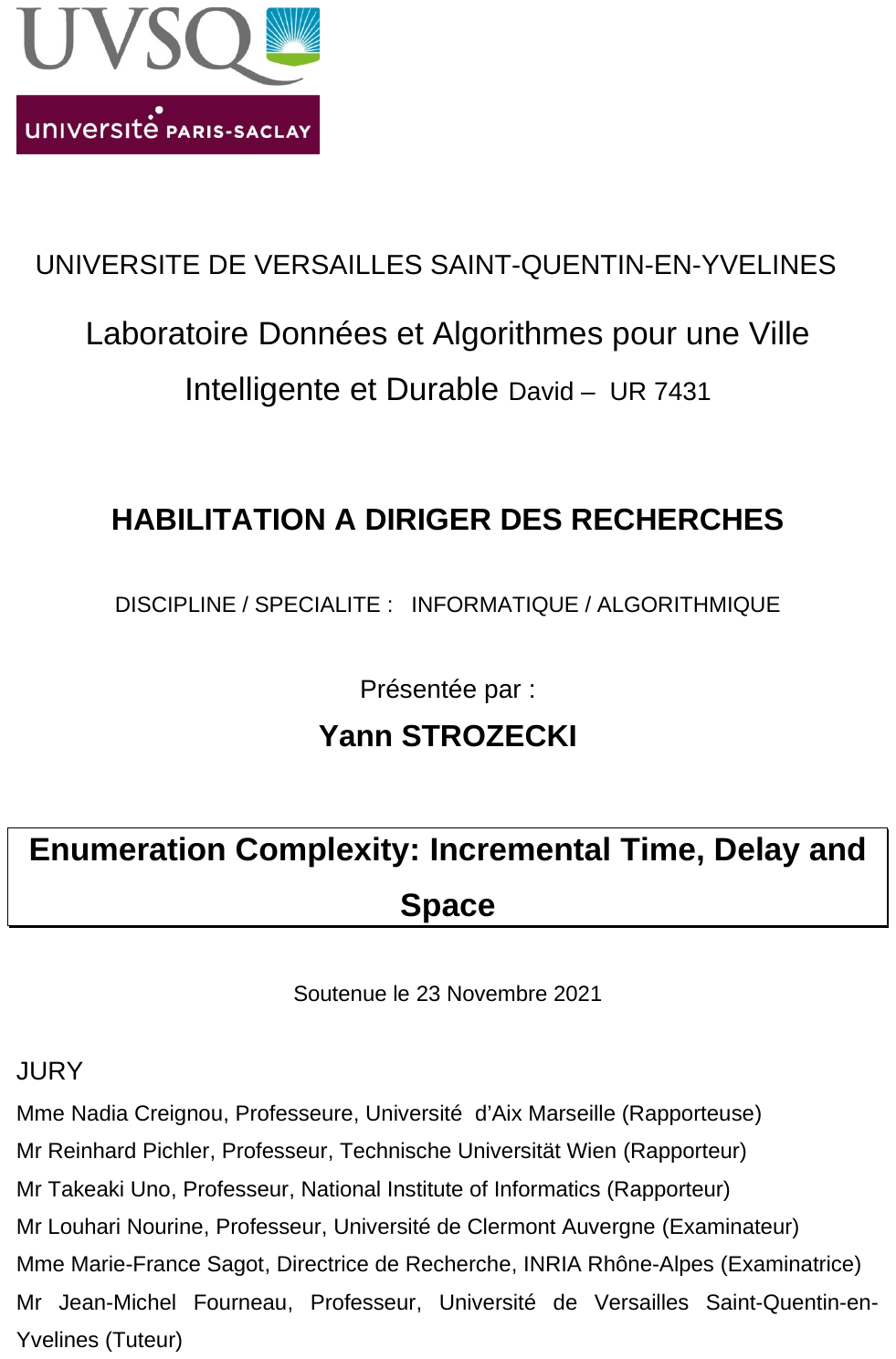}

\pagenumbering{roman}

\setcounter{page}{0}
\cleardoublepage

\section*{Acknowledgments}

Merci à mes enfants de m'avoir permis de ne pas trop me précipiter pour écrire mon habilitation.
Merci au confinement pour m'avoir permis de prendre encore plus de recul et moins de RER. Merci à ma compagne de m'avoir épousé pendant que je rédigeais ce manuscrit, j'espère pouvoir encore continuer longtemps à équilibrer ma vie familiale et professionelle ainsi. Merci à ma mère, mon frère et tout ma famille d'être toujours présents et de m'aider pour le pot :) 

Je suis très reconnaissant à tous les collègues qui m'ont proposé des problèmes, invité à des séminaires et des conférences et m'ont posé de nombreuses questions: c'est grâce à vous j'ai continué à travailler sur l'énumération. 
Je suis également reconnaissant à tous mes coauteurs, avec qui j'ai eu beaucoup de plaisir à travailler sur les résultats présentés dans ce manuscrit.
Je suis prêt pour des nouvelles questions Arnaud et Florent, il va falloir trouver une garde pour vos enfants ! 
Thanks to Olivier Commowick for distributing a nice thesis template (at \url{http://olivier.commowick.org/thesis_template.php}).

I heartily thank the reviewers of this manuscript, Nadia Creignou, Reinhard Pichler and Takeaki Uno to have accepted to take my manuscript as a holiday reading and for their excellent and quick work. I also thank the whole jury, to have found the time to be present for the defense. I want to especially thank Jean-Michel, my tutor, that I have so frequently bothered with administrative questions about this HDR.

\dominitoc	
\tableofcontents

\mainmatter

\chapter{Introduction}

\section{Foreword}
 
During the past ten years, I have worked in several areas related to algorithms and complexity.
I was first interested in algebraic complexity and computer algebra during a post-doc in Toronto, working with Bruno Grenet, Pascal Koiran and Natacha Portier.
Then, I began to study algorithmic game theory with David Auger, Pierre Coucheney and our PhD student Xavier Badin de Montjoye,
working on algorithms to find optimal strategies for simple stochastic games.
More recently, I have studied periodic scheduling problems with my PhD student Maël Guiraud, both from a theoretical and practical point of view, to optimize the latency of telecommunication networks.

In this habilitation thesis, I focus on a fourth topic I have studied since the beginning of my PhD: the structural complexity of enumeration, in particular the relationships between different notions of tractability. This thesis tries to be a survey on the complexity of enumeration, with an emphasis on my contributions to the area. A good part of this thesis is taken from my recent survey~\cite{strozecki2019enumeration} and some of my articles on the complexity of enumeration~\cite{capelli2019incremental,DBLP:journals/dmtcs/MaryS19,capelli2020enumerating}.
I found inspiration and a broad view on enumeration in several theses~\cite{Bagan09,phdstrozecki,brault2013pertinence,mary2013enumeration,marino2015enumeration,vigny2018query}. I was also motivated to write a survey on enumeration complexity and this thesis by the ever increasing community of researchers interested in enumeration
coming from many different backgrounds: graph algorithms, parametrized complexity, exact exponential algorithms, logic, databases, enumerative combinatorics, applied algorithms in bioinformatics, cheminformatics, networks \dots 
In the last ten years, the topic has attracted much attention and these different communities began to share their ideas and problems, as exemplified by the creation of WEPA, the International Workshop on Enumeration Problems and Applications, two recent Dagstuhl worskhops \emph{``Algorithmic Enumeration: Output-sensitive, Input-Sensitive, Parameterized, Approximative''} and \emph{``Enumeration in Data Management''}, and the creation of wikipedia pages for enumeration complexity and algorithms. 
The results I present have been obtained with my co-authors, in particular Florent Capelli, Arnaud Mary, Arnaud Durand, Sandrine Vial and Franck Quessette but they also originate from helpful discussions with my colleagues Nadia Creignou, Frédéric Olive, Étienne Grandjean, Stefan Mengel, Mamadou Kanté, Lhouari Nourine, Alexandre Vigny, Florent Madelaine and many others.

\section{Introduction}

Modern enumeration algorithms date back to the $70$'s with graph algorithms~\cite{tiernan1970efficient}, while fundamental complexity notions for enumeration have been proposed 30 years ago by Johnson, Yannakakis and Papadimitriou~\cite{johnson1988generating}. However, much older problems can be reinterpreted as enumeration: the baguenaudier game~\cite{lucas1882recreations} from the $19$th century can be seen as the problem of enumerating integers in Gray code order. There are even thousand years old examples of methods to list simple combinatorial structures, the subsets or the partitions of a finite set, as reported by Ruskey~\cite{ruskey2003combinatorial} in his book on combinatorial generation.
 Algorithms to list all integers, tuples, permutations, combinations, partitions, set partitions, trees of a given size are also called combinatorial algorithms and are particular enumeration problems. Combinatorial algorithms are the subject of a whole volume of the Art of Computer Programming~\cite{knuth2011art} and Knuth even confessed that these are his favorite algorithms; while I agree I would extend that appreciation to all enumeration algorithms.

Hundreds of different enumeration problems have now been studied, see the actively maintained list \href{http://www-ikn.ist.hokudai.ac.jp/~wasa/enumeration_complexity.html}{Enumeration of Enumeration Algorithms and Its Complexity}~\cite{wasa2016enumeration}. Some enumeration problems have important practical applications: A database query is the enumeration of assignments of a formula, data mining relies on generating all frequent objects in a large database (frequent itemsets~\cite{agrawal1994fast} and frequent subgraphs~\cite{jiang2013survey}), finding the minimal transversals of a hypergraph has applications in many fields~\cite{hagen2008algorithmic} such as biology, machine learning, cryptography \dots 

A classical approach to enumeration is to see it as a variation on decision or search problems where one tries to get more information. As an example, consider the matchings of a graph, we may want to solve the following tasks.
\begin{itemize}
 \item (decision problem) Decide whether there is a matching. 
 \item (search problem) Produce a (maximal) matching. 
 \item (optimization problem) Produce a matching of largest cardinality. 
 \item (counting problem) Count all matchings.
 \item (enumeration problem) List all (maximal) matchings.
\end{itemize}

Usually, to analyze the complexity of a problem, we relate the time to produce the output to the size of the input. The originality of enumeration problems is that the output is usually very large with regard to the input. Hence, interpreting the \emph{total time} to produce the whole set of solutions as a function of the input size is not informative, since it is almost always exponential. 
To make enumeration complexity relevant, we need to consider other parameters of the problem out of the instance size:  the size of the output (or its cardinal since it is a set of solutions) and the size of a single solution in the output. We must also study finer complexity measures than total time, since it does not allow to differentiate most enumeration problems. 

The simplest way to improve on the analysis of enumeration algorithms is to evaluate how the \emph{total time} to compute all solutions relates to \emph{the size of the input and of the output}. Algorithms whose complexity is given in this fashion are often called \emph{output sensitive}, by contrast to \emph{input sensitive algorithms}~\cite{fomin2010exact}.  Note that output sensitivity is relevant even when the number of objects to produce is small. In computational geometry, it allows to give better complexity bounds, for instance on the convex hull problem~\cite{chan1996optimal}. When the \emph{total time} is polynomial in the size of the input and the output, the algorithm is said to be \emph{output polynomial} (or sometimes in total polynomial time). Output polynomial is a good measure of tractability when \emph{all} elements of a set must be generated, for instance to count the number of solutions or to compute some statistics on the set of solutions. 

Often, output polynomial is not restrictive enough given the outstanding number of solutions and we must ask for a total time \emph{linear} in the number of solutions.
In that case, the relevant complexity measure is the total time divided by the number of solutions called {\em amortized time} or {\em average delay}. Many enumeration algorithms generating combinatorial objects are in constant amortized time or CAT, such as the generation of unrooted trees of a given size~\cite{wright1986constant}, linear extensions of a partial order~\cite{pruesse1994generating} or integers given in Gray code order~\cite{knuth1997art}. Uno also proposed in~\cite{uno2015constant} a general method to obtain constant amortized time algorithms, which can be applied, for instance, to find the matchings or the spanning trees of a graph.

Enumeration algorithms are also often used to compute an optimal solution by generating all admissible solutions. For instance, finding maximum common subgraphs up to isomorphism, a very important problem in cheminformatics, is $\NP$-hard and is solved by listing all maximal cliques~\cite{ehrlich2011maximum}.
The notion of best solution is not always clear and enumeration is then used to build libraries of interesting objects to be analyzed by experts, as it is done in biology, chemistry or network analytics~\cite{andrade2016enumeration,barth2015efficient,bohmova2018computing}. In particular, when confronted to a multicriteria optimisation problem, a natural approach is to enumerate the Pareto's frontier~\cite{papadimitriou2000approximability,vassilvitskii2005efficiently,bazgan2015approximate}.
In all these applications, if the set of solutions is too large, we are interested in generating the largest possible subset of solutions. Hence, a good enumeration algorithm should guarantee that it will find as many solutions as possible in a predictable amount of time. In this case, \emph{polynomial incremental time} algorithms are more suitable: an algorithm is in  polynomial incremental time if the time needed to enumerate the first $k$ solutions is polynomial in $k$ and in the size of the input. Such algorithms naturally appear when the enumeration task is of the following form: given a set of elements and a polynomial time function acting on tuples of elements, produce the closure of the set by the function. One can generate such closure by iteratively applying the function until no new element is found. As the set grows, finding new elements becomes harder. For instance, the best algorithm to generate all circuits of a matroid uses a closure property of the circuits~\cite{khachiyan2005complexity} and is thus in polynomial incremental time. The fundamental problem of generating the minimal transversals of a hypergraph can also be solved in quasi-polynomial incremental time~\cite{fredman1996complexity} and some of its restrictions in polynomial incremental time~\cite{eiter2003new}.

However, when one wants to process a set in a streaming fashion such as the
answers of a database query, incremental polynomial time is not enough and we need a
good \emph{delay} between the output of two consecutive solutions, usually bounded by a
polynomial in the input size. We refer to such algorithms as \emph{polynomial
  delay} algorithms. Many problems admit such algorithms, e.g. enumeration of
the cycles of a graph~\cite{read1975bounds}, the satisfying assignments of
some tractable variants of $\SAT$~\cite{creignou1997generating} or the spanning trees and
connected induced subgraphs of a graph~\cite{avis1996reverse}. All polynomial
delay algorithms are based on few methods such as \emph{backtrack search} (also called flashlight search or binary partition) or \emph{reverse search}, see~\cite{mary2013enumeration} for a survey.

When the size of the input is much larger than the size of one solution, think of generating subsets of vertices of a hypergraph or a small query over a large database, polynomial delay is an unsatisfactory measure of efficiency. The good notion of tractability is \emph{\strongpdelay}, i.e. the delay is polynomial \emph{in the size of the last solution}. A folklore example is the enumeration of the paths in a DAG, which is in delay linear in the size of the last generated path. More complex problems can then be reduced to generating paths in a DAG, such as enumerating the minimal dominating sets in restricted classes of graphs~\cite{golovach2018output}.

Unlike classical complexity classes, none of the classes introduced in this thesis but $\EnumP$, the equivalent of $\NP$, have complete problems. Because of that, no notion of reduction for enumeration
seems better than the others, and for many proof of hardness, ad hoc reductions are used. To overcome the lack of completness result, several works restrict themselves to smaller families of enumeration problems in the hope of better classifying their complexity: assignments of SAT formula~\cite{creignou1997generating}, homomorphisms~\cite{bulatov2012enumerating}, subsets given by saturation operators~\cite{DBLP:journals/dmtcs/MaryS19}, FO queries over various structures~\cite{DBLP:journals/sigmod/Segoufin15}, maximal subgraphs~\cite{cohen2008generating,conte2019new,conte2019listing}\dots 

\paragraph{Organization}

In Chapter~\ref{chap:framework}, enumeration problems and the related computation model are defined, with an emphasis
on the consequences of several definitional choices. 
Then, in Chapter~\ref{chap:time}, we introduce complexity classes related to three time complexity measures, \emph{total time}, \emph{incremental time} and \emph{delay}. For each of these classes, we provide a separation theorem and a characterization when possible. 
In Chapter~\ref{chap:space}, we study several restrictions on space in enumeration. In particular, we show that incremental linear time is in fact equal to polynomial delay, even with polynomial space.
Chapter~\ref{chap:low} is devoted to the presentation of low complexity classes, with strong contraints on the delay, to capture tractability in different contexts. In particular, we present several methods to show that a problem has a strong polynomial delay, applied to the problem of generating the models of a DNF formula. 
We also investigate the restricted class of problems which admits a uniform random generator of solutions, as a way to avoid using space in addition of time,
and show that it can be related to polynomial delay. 
Instead of generating solutions in a random order or an unspecified one, it is often relevant to generate them in a specified order, to get more ``interesting'' solutions first. We review how an additional constraint on the order of generated solutions may increase the complexity in Chapter~\ref{chap:order}.
In Chapter~\ref{chap:limits}, we review the few lower bound results for enumeration problems. To obtain more lower bounds and classification results, we present several restricted frameworks from logic and algebra. In addition to lower bounds, we provide many algorithms for tractable classes of saturation problems, interpolation problems and generalized first order queries.
 In Chapter~\ref{chap:practical}, we briefly present an algorithm developped to enumerate a certain kind of planar map useful in cheminformatic, underlying the differences between the design of a theoretical algorithm and 
 the implementation of an algorithm used in practice. To deal with a huge solution space, even constant delay algorithm are not satisfactory,
  hence, in Chapter~\ref{chap:horizons}, we present several alternative approaches to the task of listing solutions exhaustively.

\chapter{Enumeration Framework}
\label{chap:framework}
\minitoc

\section{Enumeration Problem}

Let $\Sigma$ be a finite alphabet and $\Sigma^*$ be the set of finite words built on $\Sigma$.
We denote by $|x|$ the length of $x \in \Sigma^*$.
Let $A\subseteq \Sigma^{*}\times\Sigma^{*}$ be a binary predicate, we write $A(x)$ for the set of $y$ such that $A(x,y)$ holds. The enumeration problem $\enum{A}$ is the function which associates $A(x)$ to $x$. The element $x$ is often called the instance or the input, while an element of $A(x)$ is called a solution. We denote the cardinality of a set $S$ by $|S|$.

\begin{example}
Let $H = (X,E)$ be an hypergraph, with $X$ the set of vertices and $E$ the set of hyperedges.
A transversal or hitting set is a subset $T$ of $X$, such that all hyperedges $e \in E$ have a non-empty intersection
with $T$. In other words, $T$ is a cover of the hypergraph $H$. The problem of generating all transversals minimal for inclusion
is a fundamental enumeration problem, with many applications~\cite{hagen2008algorithmic}.

The binary predicate $\textsc{Min-Transversals}(H,T)$ is true if and only if $T$ is a minimal transversal of $H$. 
The problem of listing minimal transversals is denoted by $\enum{\textsc{Min-Transversals}}$. Hypergraph $H$ is an input of the problem,  
$\textsc{Min-Transversals}(H)$ denotes the set of all minimal transversals of $H$ and $T \in \textsc{Min-Transversals}(H)$ is called a solution of $\enum{\textsc{Min-Transversals}}$ (for input $H$).
\end{example}

In this thesis, we only consider predicates $A$ such that $|A(x)|$ is finite for all $x$.
This assumption could be lifted and the definitions on the complexity of enumeration adapted to the infinite case. This is not done here because infinite sets of solutions behave quite differently when studying the complexity of their enumeration. However, there are many natural infinite enumeration problems such as listing all primes or all words of a context-free language~\cite{florencio2015naive}.

We can further reduce the set of enumeration problems by adding constraints on their solutions. First, 
the size of each solution can be bounded by a function of the instance size. It is reasonable since in many problems the size of a solution is fixed, known beforehand and not too large, otherwise we would not even try to produce them.

\begin{definition}
 A binary predicate $A$ is \emph{polynomially balanced} if there is a polynomial $p$, such that, for all $y \in A(x)$, $|y| \leq p(|x|)$. 
\end{definition}

Let $\test{A}$ be the problem of deciding, given $x$ and $y$, whether $y \in A(x)$. 
In almost all practical enumeration problems, one can check efficiently whether a string is a solution. This can be captured by the constraint $\test{A} \in \P$
From a polynomially balanced predicate $A$ with $\test{A}$ in polynomial time, we can define 
an $\NP$ problem by asking whether $A(x)$ is empty or a $\#\P$ problem by asking for $|A(x)|$.
It is then natural to define a class of enumeration problems from these predicates, analogous to $\NP$.

\begin{definition}
 The class $\EnumP$ is the set of all problems $\enum{A}$ where $A$ is polynomially balanced and 
 $\test{A} \in \P$.
\end{definition}

\begin{example}
Any transversal $T$ of an hypergraph $H = (X,E)$ satisfies $T \subseteq X$, hence $|T| \leq |H|$ and $\textsc{Min-Transversals}$ is polynomially balanced.
Moreover, $\test{\textsc{Min-Transversals}}$ can be solved in polynomial time in $H$. First we test whether $T$ covers all edges of $E$ and then
we verify it is minimal for this property, that is $T \setminus \{v\}$ is not a transversal for all $v$ vertices of $H$.
 Hence, $\enum{\textsc{Min-Transversals}} \in \EnumP$. 
\end{example}

The problems in $\EnumP$ can be seen as the task of listing the solutions (or witnesses)
of $\NP$ problems. One good property of $\EnumP$ is to have some complete problems for the simple parsimonious reduction. However, there is no standard notion of reduction which makes $\EnumP$-complete all seemingly hard problems of $\EnumP$, and this difficulty is already apparent in the definition of the polynomial hierarchy for enumeration.
This hierarchy is built in~\cite{creignou2017complexity} by adding oracles in the polynomial hierarchy
to polynomial delay or incremental polynomial time machine. This gives a strict hierarchy of hard problems,
with some natural complete examples in its first levels. Note that, in contrast to counting complexity,
defining a hierarchy through the complexity of $\test{A}$ does not seem relevant, since it is easy to define a ''simple'' enumeration problem with a hard $\test{A}$ problem by adding many trivial solutions to the problem.

Finally, note that in the definition of $\EnumP$ nothing specific about enumeration is taken into account. 
We are able to define it before even specifying the computation model or complexity measures specific to enumeration, as it only relies on the complexity of deciding the auxiliary problem $\test{A}$.

\section{Model of Computation} 

The model of computation is the random access machine (RAM) with comparison, addition, subtraction and multiplication as its basic arithmetic operations and an operation $\Output(i,j)$ which outputs the concatenation of the values of registers $R_i, R_{i+1}, \dots, R_j$. RAM machines have been introduced by Cook and Reckhow to better model the storage of existing machines~\cite{DBLP:journals/jcss/CookR73,aho1974design}; for variants designed for enumeration see~\cite{Bagan09,phdstrozecki}. All instructions are in constant time except the arithmetic instructions which are in time logarithmic in the sum of the integers they are called on. 

A RAM machine solves $\enum{A}$ if, on every input $x \in \Sigma^{*}$, it produces a sequence $(y_{1}, \dots, y_{s})$ such that $ A(x) = \left\lbrace y_{1}, \dots, y_{s} \right\rbrace $ and for all $i\neq j,\, y_{i} \neq y_{j}$, that is \emph{no solution must be repeated}! We may assume that all registers are initialized to zero. The space used by the machine at some point of its computation is the sum of the length of the integers up to the last registers it has accessed. We define by $T_M(x,i)$ the time taken by the machine $M$ on input $x$ up to point when the $i$th $\Output$ instruction is executed. Usually we drop the subscript $M$ and write $T(x,i)$ when the machine is clear from the context. 
The \emph{delay} of a RAM machine which outputs the sequence $(y_{1}, \dots, y_{s})$ is the maximum over all $i\leq s$ of the time the machine uses between the generation of $y_i$ and $y_{i+1}$, that is $\max_{1 \leq i < s}{T(x,i+1)-T(x,i)}$. In some works, preprocessing and postprocessing times are considered separately from the delay. It is extremely rare to need more time for deciding if the enumeration is finished than to output a solution, hence we consider in this thesis that \emph{there is no postprocessing}: enumeration stops as soon as the last solution is output. To make low complexity classes interesting, it is important to allow preprocessing, that is $T_M(x,1)$, to be larger than the delay and we will mention it when appropriate.
The \emph{average delay} or \emph{amortized time} of a RAM machine is the average time needed by the machine to produce a solution, that is $T(x,s)/s$. The delay is an upper bound to the average delay. 

\paragraph{Why a RAM instead of a Turing Machine?}

  While the RAM machine better maps to real computers and is thus better 
  to measure precisely the complexity of algorithms, it can be simulated with cubic slowdown
  by a Turing Machine~\cite{DBLP:journals/jcss/CookR73,papadimitriou2003computational}. A polynomial time Church's Thesis states that all realistic computational models are equivalent up to polynomial slowdown\footnote{the thesis is not true for quantum computers, but we still do not know whether they are physically implementable}, which makes the computation model irrelevant for classical complexity. However, we can isolate a sequence of operations \emph{during the execution} of a RAM which takes exponentially more time on a Turing Machine which simulates it. The RAM has the power of \emph{indirection}: it can access any address in constant time, while the Turing machine should traverse all its tape to read some cell. This allows to use dictionary data structures essential in enumeration such as AVL trees or tries~\cite{cormen2009introduction} which gives linear time access to elements inside an exponential set.

\paragraph{Why such a constant time OUTPUT instruction?}

The choice of the OUTPUT instruction and of its complexity is only relevant for 
algorithms with a sublinear delay in the size of the solution output, in particular a constant delay.
In all definitions of RAM machine for enumeration~\cite{Bagan09,phdstrozecki,mary2013enumeration,brault2013pertinence} the OUTPUT instruction can be issued in constant time. This is required to make constant delay interesting by capturing problems like Gray code enumeration or query answering, otherwise only a constant number of constant size solutions can be generated in constant delay. This is similar to the definition of logarithmic space, where machine can write an output in a special tape which is not taken into account in the space used~\cite{papadimitriou2003computational}. 
Constant time output is meaningful, when only the deltas between solutions are output rather than solutions themselves. It is also relevant, if we just do a constant time operation on each solution such as counting them or evaluating some measure which depends only on the constant amount of changes between two consecutive solutions.

The originality of the RAM model proposed in this thesis is to allow outputing solutions at different positions of the memory, to really take advantage of indirection. In previous models the size of the solution was implicit, here we make it explicit and also maintain its position in memory. This choice does not seem a big stretch from reality since in a computer, the memory zone in which solutions are stored is not always the same. In fact, a programmer does not even control (but rather the system and the cpu) where information is physically stored in memory. Our model enables us to list, with constant delay, consecutive solutions which may differ by an unbounded number of elements, which is not possible in the traditional model with fixed registers for the output. 

\paragraph{Why this cost model ?}

The cost model used to take into account the different operations
of the RAM has no impact on complexity classes defined by a polynomial bound, so the choice is mostly arbitrary. 
We choose to count addition and multiplication as a linear number of operations
in the size of their arguments. We could choose to count the addition as a unit time operation, 
but not the multiplication otherwise we can generate doubly exponential numbers in linear time.
Note that we allow unbounded integers in registers to deal with large data structures.

For small complexity classes, such as linear delay or constant delay, the choice of the cost model becomes extremely relevant. In such situations, sometimes implicitly, the \emph{uniform cost model} (see~\cite{DBLP:journals/jcss/CookR73,aho1974design}) is chosen: addition, multiplication and comparison are in constant time. However, if the input is of size $n$, the machine has $\log(n)$ word-size, i.e. integers in registers are bounded by $n$. This model is robust enough to define linear time computation~\cite{grandjean2002machine} and constant delay in enumeration~\cite{DBLP:journals/tocl/DurandG07,Bagan09}. It is similar to the word RAM model, a transdichotomous model~\cite{fredman1993surpassing}, used to give finer and more realistic bounds for data structures. 
In some enumeration algorithms, all solutions must be stored and there can be $2^n$ of them. 
Hence, restricting integers in registers to be of size $\log(n)$ or even $k\log(n)$ for a fixed $k$ is not sufficient to address all the memory needed to store the solutions. Hence, rather than bounding the register size, a good compromise is to take as cost of an instruction the logarithm of the sum of its arguments \emph{divided by $n$}. Alternatively, some constant time operations on unbounded integers such as comparison and incrementation can be allowed.

\section{Parsimonious Reduction}

In this section we introduce a first notion of reduction for enumeration problems. 
A reduction is a binary relation over enumeration problems (or equivalently over the relations which define the problems).
A reduction, denoted by $\prec$, must be \emph{transitive}, that is if $\enum{A} \prec \enum{B}$ and  $\enum{B} \prec \enum{C}$ then
$\enum{A} \prec \enum{C}$. 
A class of enumeration problems $\mathcal{C}$ is closed under the reduction $\prec$ if $\enum{B} \in \mathcal{C}$ and $\enum{A} \prec \enum{B}$ 
imply $\enum{A} \in \mathcal{C}$. This property guarantees that reductions can be used to provide algorithms and not just hardness results. Note that if a reduction allows for too much computation with regard to a class, the class is not closed under the reduction, for instance $\P$ is not closed under exponential time reductions.

The \emph{parsimonious reduction} for counting problems enforces equality of numbers of solutions, by requiring a bijection between sets of solutions in addition to the bijection between inputs. If this bijection is explicit and tractable, then such a reduction is adapted to enumeration complexity.

\begin{definition}[Parsimonious Reduction]
Let $\enum{A}$ and $\enum{B}$ be two enumeration problems. 
A parsimonious reduction from $\enum{A}$ to $\enum{B}$ is a pair of polynomial time computable functions $f, g$ such that for all $x$, $g(x,\cdot)$ restricted to $B(f(x))$ is a bijection with $A(x)$ and $g(x,\cdot)^{-1}$ is polynomial time computable.
\end{definition}

\begin{example}
Let $\textsc{Clique}(G,S)$ be the predicate which is true when $S$ is a clique in the graph $G$.
Let $\textsc{Independent-Set}(G,S)$ be the predicate which is true when $S$ is an independent set in the graph $G$.
Let $f$ be the function which maps the graph $G$ to its complement graph $H$, that is two vertices are adjacent in $G$ if and only if they are not adjacent in $H$. 
Let $g(G,S)$ be the function which maps $S$ a set of vertices in $G$ to the same set of vertices $S$ in $f(G)$.
Using these two functions computable in polynomial time, we have $\enum{\textsc{Clique}} \prec \enum{\textsc{Independent-Set}}$
(and also $\enum{\textsc{Independent-Set}}  \prec  \enum{\textsc{Clique}}$).
\end{example}

Let us denote by $g_x(y)$ the bijection beetween $B(f(x))$ and $A(x)$ in the definition of parsimonious reduction.
The condition that $g_x^{-1}$ is polynomial time computable is in the definition only to ensure that $\EnumP$ is stable under
parsimonious reduction as we now prove.
Recall that $\enum{B} \in \EnumP$ means that $B$ is polynomially balanced and  $\test{B}$ is in polynomial time.
To prove that $\EnumP$ is closed under parsimonious reductions, we must prove that if $\enum{A} \prec \enum{B}$ and $\enum{B} \in \EnumP$,
then $\enum{A} \in \EnumP$.  The predicate $A(x,y)$ holds if and only if $B(f(x),g_x^{-1}(y))$, hence it can be decided in 
polynomial time because we have assumed that $B$, $f$ and $g_x^{-1}$ are polynomial time computable. Second, since $f$ and $g$ are polynomial time computable, the elements of $A(x) = g(x,B(f(x)))$ are polynomial in the size of $x$. 

An $\EnumP$-complete problem is defined as a problem in $\EnumP$ to which any problem in $\EnumP$ reduces by parsimonious reduction.
 The problem $\enum{3SAT}$, the task of listing all solutions of a $3$-CNF formula is $\EnumP$-complete, since the reduction used in the proof that 3SAT is $\NP$-complete~\cite{cook1971complexity} is parsimonious.

 Let us consider the predicate $3SAT_{0}(\phi,x)$ which is true if and only if $x$ is a satisfying assignment of the $3$-CNF formula $\phi$ or $x$ is the all zero assignment. Then $3SAT_{0}(\phi)$ is never empty and therefore many problems of $\EnumP$ cannot be reduced to $\enum{3SAT_{0}}$ by parsimonious reduction. However, the problem intuitively feels like a complete problem and it can be made so by relaxing the reduction. Indeed, from an enumeration perspective,  $\enum{3SAT_{0}}$ is exactly as hard as $\enum{3SAT}$: given an algorithm to enumerate $\enum{3SAT_{0}}$, we can enumerate $\enum{3SAT}$ by discarding if necessary the all zero solution (or adding it when doing the inverse reduction). This reduction requires an additional polynomial time computation only. This shows that parsimonious reductions are not sufficient to classify the complexity of enumeration problems, and we present several other reductions in the next chapter.

\chapter{Time Complexity: Total Time, Incremental Time and Delay}
\label{chap:time}
\minitoc

\section{Output Polynomial Time}\label{sec:outputsens}

To measure the complexity of an enumeration problem, we consider the total time taken to compute all solutions. If the total time 
of an algorithm depends on the size of the input only, we speak of \emph{input sensitive} algorithm. The Bron-Kerbosch algorithm~\cite{bron1973algorithm} finds all maximal cliques
in a graph in time $O(2^{n/3})$, with $n$ the number of vertices of the graph. Since there can be as many as $2^{n/3}$ maximal cliques in a graph of $n/3$ disjoint triangles, Bron-Kerbosch algorithm
is optimal with regard to the input size. However, its complexity does not decrease with the number of solutions, or at least there is no proof of such result. Since the number of solutions is much smaller than $2^{n/3}$ for most instances, it is more precise and relevant to bound the total time as a function of the size of the input and of \emph{the output}. By contrast to input sensitive algorithms, such algorithms are called \emph{output sensitive}. Note that bounding the total time by a polynomial in the input size only is not relevant, since it forces the number of solutions to be polynomial. 
Then, to define efficient algorithms, it is natural to bound the total time by a polynomial in the number of solutions and in the size of the input. Algorithms with this complexity are said to be in \textbf{output polynomial time} or sometimes in polynomial total time. 

\begin{definition}[Output polynomial time]
 A problem $\enum{A}\in \EnumP$  is in $\OutputP$  if there is a polynomial $p(x,y)$ and a machine $M$ which solves $\enum{A}$ and such that for all $x$, $T(x,|A(x)|) < p(|x|,|A(x)|)$.
\end{definition}

For instance, if we represent a polynomial by its set of monomials, then classical algorithms for interpolating multivariate polynomials from their values are output polynomial~\cite{zippel1990interpolating,strozecki2010} as they produce the polynomial in time proportional to the number of its monomials, see Chapter~\ref{chap:limits} for more details. 

The classes $\EnumP$ and $\OutputP$ may be seen as analog of $\NP$ and $\P$ for enumeration.
It turns out that their separation is equivalent to the $\P = \NP$ question, since an algorithm 
in $\OutputP$ allows to decide whether there is at least one solution in polynomial time.

\begin{proposition}[Folklore, see~\cite{capelli2019incremental}]\label{prop:output}
 $\OutputP = \EnumP$ if and only if $\P = \NP$.
\end{proposition}

Output polynomial algorithms are useful when all solutions need to be generated, 
for a proof of optimality or to build an exhaustive library of objects. However,
if we want to generate only a few solutions, because the whole process is too long, 
an output polynomial algorithm does not offer any guarantee, since all solutions can be produced
at the end of the computation, see Figure~\ref{fig:outputpoly}.

\begin{figure}
\label{fig:outputpoly}
\begin{center}
\begingroup%
  \makeatletter%
  \providecommand\color[2][]{%
    \errmessage{(Inkscape) Color is used for the text in Inkscape, but the package 'color.sty' is not loaded}%
    \renewcommand\color[2][]{}%
  }%
  \providecommand\transparent[1]{%
    \errmessage{(Inkscape) Transparency is used (non-zero) for the text in Inkscape, but the package 'transparent.sty' is not loaded}%
    \renewcommand\transparent[1]{}%
  }%
  \providecommand\rotatebox[2]{#2}%
  \newcommand*\fsize{\dimexpr\f@size pt\relax}%
  \newcommand*\lineheight[1]{\fontsize{\fsize}{#1\fsize}\selectfont}%
  \ifx\svgwidth\undefined%
    \setlength{\unitlength}{465.41879248bp}%
    \ifx\svgscale\undefined%
      \relax%
    \else%
      \setlength{\unitlength}{\unitlength * \real{\svgscale}}%
    \fi%
  \else%
    \setlength{\unitlength}{\svgwidth}%
  \fi%
  \global\let\svgwidth\undefined%
  \global\let\svgscale\undefined%
  \makeatother%
  \begin{picture}(1,0.23835109)%
    \lineheight{1}%
    \setlength\tabcolsep{0pt}%
    \put(0,0){\includegraphics[width=\unitlength,page=1]{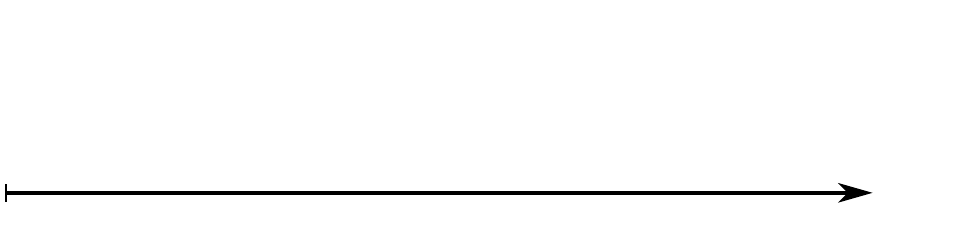}}%
    \put(0.9110806,0.0312695){\color[rgb]{0,0,0}\makebox(0,0)[lt]{\lineheight{0}\smash{\begin{tabular}[t]{l}$time$\end{tabular}}}}%
    \put(0,0){\includegraphics[width=\unitlength,page=2]{total.pdf}}%
    \put(0.23544903,0.21864857){\color[rgb]{0,0,0}\makebox(0,0)[lt]{\lineheight{0}\smash{\begin{tabular}[t]{l}$|A(x)|$ solutions in time $T(x,|A(x)|) \leq p(|x|,|A(x)|)$\end{tabular}}}}%
    \put(-0.00169958,0.00036508){\color[rgb]{0,0,0}\makebox(0,0)[lt]{\lineheight{1.25}\smash{\begin{tabular}[t]{l}0\end{tabular}}}}%
  \end{picture}%
\endgroup%

\end{center}
\caption{An enumeration in output polynomial time. A red bar represents the output of a solution at some point in time of the enumeration.}
\end{figure}

\section{Incremental Polynomial Time}\label{sec:inc}

When generating all solutions is too long, we want to be able to generate at least some.
To do so, we can fix an arbitrary parameter $k$ and ask to enumerate only $k$ solutions. Sometimes,
an order is fixed on the solutions and the problem is to find the \emph{$k$-best} ones. The $k$-shortest path and the $k$-minimum spanning tree problems are typical $k$-best problems with many applications as explained in the survey~\cite{eppstein2015k}. In the enumeration paradigm, we are slightly more general:
we do not require a parameter $k$ (either fixed or given as input), but want to measure and bound the time used by an algorithm to produce any given number of solutions. This dynamic version of the total time is called \emph{incremental time}. Given an enumeration problem $A$, we say that a machine $M$ solves $\enum{A}$ in incremental time $f(k)g(n)$ and preprocessing $h(n)$ if on every input $x$ of size $n$, we have:

\begin{itemize}
	\item $T(x,1) \leq h(n)$ (preprocessing)
	\item for all $k \leq |A(x)|$, $T(x,k) - T(x,1) \leq f(k)g(n)$ (incremental time)
\end{itemize}

Separating the preprocessing from the rest of the computation allows to set up expensive data structures or to simplify the input at the beginning of an enumeration.
In that way, classes with small incremental time contain more interesting problems and we better represent the time to enumerate many solutions, which otherwise may be artificially hidden by the preprocessing time. When the preprocessing is dominated by the incremental time, that is $h(n) \leq f(1)g(n)$, we omit it.

\begin{definition}[Incremental polynomial time]
Let $a \geq 1$. A problem $\enum{A} \in \EnumP$ is in $\IncP_a$ if there is a constant $b$, a polynomial $p$ and a machine $M$ which solves $\enum{A}$ on instances of size $n$ in incremental time $O(k^an^b)$ and preprocessing time $O(p(n))$. Moreover, we define $\IncP = \bigcup_{a \geq 1} \IncP_a$.
\end{definition}

Note that allowing or not arbitrary polynomial time preprocessing does not modify the class $\IncP$, since this preprocessing can always be interpreted as a polynomial time before outputing the first solution.
We require that $\enum{A} \in \EnumP$ so that $\enum{A} \in \IncP$, however there are a few examples of problems, whose solutions cannot be tested in polynomial time, but with an algorithm in incremental polynomial time.
A clause sequence is the set of clauses satisfied by some assignment of a CNF formula, hence deciding whether a formula admits the clause sequence with all clauses is $\NP$-complete. However, there is a clever 
incremental time algorithm to generate all clause sequences for $k$-CNFs~\cite{berczi2021generating}. It relies on finding enough simple to enumerate solutions, to have the time to divide the problem into many simpler subformulas with few sequences in common.

Let $A$ be a binary predicate, $\mathsf{AnotherSol_A}$ is the search problem defined as, given $x$ and a set $\mathcal{S}$, find $y \in A(x) \setminus \mathcal{S}$ or answer that $ A(x) \subseteq \mathcal{S}$ (see~\cite{phdstrozecki,creignou2017complexity}). The problems in $\IncP$ are the ones with a polynomial time search problem, as stated in the following proposition. 

\begin{proposition}[Folklore, see~\cite{capelli2019incremental}]
\label{prop:anothersolincp}
  Let $A$ be a predicate such that $\enum{A} \in \EnumP$. $\mathsf{AnotherSol_A}$ is in $\FP$ if and only if $\enum{A}$ is in $\IncP$.
\end{proposition}

Recall that the \textbf{delay} of producing the $i$th solution by some algorithm, is the time between the production of two consecutive solutions, that is, for $1 \leq k < |A(x)|$, $|T(x,k+1) - T(x,k)|$.
The class $\IncP$ is usually defined as the class of problems solvable by an algorithm with a delay polynomial in the number of already generated solutions and the size of the input.

\begin{definition}[Usual definition of incremental polynomial time]
Let $a \geq 0$. A problem $\enum{A} \in \EnumP$ is in $\UIncP_a$ if there is a constant $b$, a polynomial $p$ and a machine $M$ which solves  $\enum{A}$ on instances of size $n$, with preprocessing time $O(p(n))$ and such that $|T(x,k+1) - T(x,k)| \in O(k^an^b)$.
 \end{definition}
 
This alternative definition is motivated by saturation algorithms, which produce solutions by applying some polynomial time rules to enrich a set of solutions until saturation. There are many saturation algorithms, for instance computing the accessible vertices from a source in a graph by flooding, computing a finite group from its generators, enumerating the circuits of matroids~\cite{khachiyan2005complexity} or computing a closure by set operations~\cite{mary2016efficient}.
\begin{example}
Let $\textsc{Union}$ be the binary predicate such that $\textsc{Union}(x,y)$ is true if and only if $x$ is a set $\{s_1,\dots,s_m\}$ of subsets of $[n]$ and 
if $y \subseteq [n]$ can be obtained by union of elements in $x$, i.e. there is $I \subseteq [m]$, such that $y = \cup_{i \in I} s_i$. 

$\enum{\textsc{Union}}$ can be easily solved by a saturation algorithm, where the production rule is the union of two solutions and the set of solutions is initialized by
$x$. As an example, consider $x = \{\{1,2\},\{2,3\},\{1,4\},\{1,3,4\}\}$. The production rule successively adds $\{1,2\} \cup \{2,3\}$, $\{1,2\} \cup \{1,4\}$, $\{1,2\} \cup \{1,3,4\}$ then the algorithm stops since no union of the obtained solutions gives a new solution. This algorithm proves that $\enum{\textsc{Union}} \in \IncP_2$, since for each new solution output we need to check its union with all other solutions.
\end{example}

With our definition of incremental polynomial time, we better capture the fact that investing more time guarantees more solutions to be output, which is a bit more general at first sight than bounding the delay because the time between two solutions is not necessarily regular. It turns out that these classes are the same by \emph{amortization}, that is storing the output solutions in a queue and outputing them regularly from the queue.

\begin{proposition}[Folklore, see~\cite{capelli2019incremental}]\label{prop:UInca}
For all $a \geq 0$,  $\IncP_{a+1} = \UIncP_{a}$.
\end{proposition}

The amortization method plus a data structure to detect duplicates such as a trie, 
allow to deal with bounded repetitions of solutions and is often used to simplify the design of enumeration algorithms,
see for instance the well-named Cheater's Lemma~\cite{DBLP:conf/pods/CarmeliK19}.

Since $\IncP$ is characterized by the search problem $\mathsf{AnotherSol_A}$, it can be related to the class $\TFNP$ introduced in~\cite{megiddo1991total}. A problem in $\TFNP$ is a polynomially balanced polynomial time predicate $A$ such that for all $x$, $A(x)$ is not empty. An algorithm solving a problem $A$ of $\TFNP$ on input $x$ outputs one element of $A(x)$. The class $\TFNP$ can also be seen as the functional version of $\NP \cap \coNP$.
It turns out that the separation of $\IncP$ and $\OutputP$ is equivalent to the separation of $\TFNP$ and $\FP$.
The proof is based on modifying $\mathsf{AnotherSol_A}$ into a problem of $\TFNP$, so that, assuming $\TFNP = \FP$, a solution to this problem is found which can then be used to solve $\mathsf{AnotherSol_A}$. Conversely, a $\TFNP$ problem can be turned into an enumeration problem in $\OutputP$, by padding.

\begin{proposition}[\cite{capelli2019incremental}]\label{prop:sepincp}
 $\TFNP = \FP $ if and only if $\IncP = \OutputP$.
\end{proposition}

The class $\OutputP$ is conditionally different from $\EnumP$ as shown in Proposition~\ref{prop:output}
and it is also strictly larger than $\IncP$ by Proposition~\ref{prop:sepincp}. However, no natural problem is known to be inside $\OutputP$ but not in $\IncP$.

\begin{openproblem}
Find a natural problem in $\OutputP$ but not in $\IncP$. It should be a problem for which $\mathsf{AnotherSol}$ is $\TFNP$-hard because of Proposition~\ref{prop:sepincp}, but the decision version of $\mathsf{AnotherSol}$, which asks whether another solution exists, should be in $\P$.  
One candidate problem is the enumeration of minimal dominating sets in $K_t$-free graphs, which is in $\OutputP$~\cite{bonamy2020enumerating} but is not known to be in $\IncP$.
\end{openproblem}

We can get a finer separation of classes than in Proposition~\ref{prop:sepincp}: $(\IncP_a)_{a\in \mathbb{R}^+}$ form a strict hierarchy inside $\IncP$.
We need to assume some complexity hypothesis since $\P = \NP$ implies $\IncP = \IncP_1$. Because we need 
to distinguish between different polynomial complexities as in fine grained complexity~\cite{williams2018some}, we also rely on the Exponential Time Hypothesis ($\ETH$). It states that there exists $\epsilon > 0$ such that there is no algorithm for $\kSAT{3}$ in time $\tilde{O}(2^{\epsilon n})$ where $n$ is the number of variables of the formula and $\tilde{O}$ means that factors in $n^{O(1)}$ are hidden. 

\begin{theorem}[\cite{capelli2019incremental}]\label{th:hierarchy}
  If $\ETH$ holds, then $\IncP_a \subsetneq \IncP_b$ for all $a<b$.
\end{theorem}

The theorem is proved by considering a modified version of SAT with a carefully chosen number of trivial solutions and an exponential repetition of the real ones. Then, if $\IncP_a = \IncP_b$ we can find a solution of any SAT formula slightly faster than brute force by using the incremental algorithm, which in turn proves a larger collapse of the $\IncP_a$ classes. Applying this trick repeatedly proves the theorem. With a similar but simpler proof, the same strict hierarchy for $\OutputP$ can be obtained.

\begin{openproblem}
Prove the converse of Theorem~\ref{th:hierarchy} or weaken the $\ETH$ hypothesis. 
\end{openproblem}

In several applications, in particular when enumerating the result of a query over a database, it is asked to output the 
solutions regularly, so to process those solutions on the fly without storing them nor pausing the enumeration process.  
Incremental polynomial time guarantes only a limited form of regularity as illustrated in Figure~\ref{fig:incremental}, which does not seem
sufficient.

\begin{figure}
\label{fig:incremental}
\begin{center}
\begingroup%
  \makeatletter%
  \providecommand\color[2][]{%
    \errmessage{(Inkscape) Color is used for the text in Inkscape, but the package 'color.sty' is not loaded}%
    \renewcommand\color[2][]{}%
  }%
  \providecommand\transparent[1]{%
    \errmessage{(Inkscape) Transparency is used (non-zero) for the text in Inkscape, but the package 'transparent.sty' is not loaded}%
    \renewcommand\transparent[1]{}%
  }%
  \providecommand\rotatebox[2]{#2}%
  \newcommand*\fsize{\dimexpr\f@size pt\relax}%
  \newcommand*\lineheight[1]{\fontsize{\fsize}{#1\fsize}\selectfont}%
  \ifx\svgwidth\undefined%
    \setlength{\unitlength}{465.41879248bp}%
    \ifx\svgscale\undefined%
      \relax%
    \else%
      \setlength{\unitlength}{\unitlength * \real{\svgscale}}%
    \fi%
  \else%
    \setlength{\unitlength}{\svgwidth}%
  \fi%
  \global\let\svgwidth\undefined%
  \global\let\svgscale\undefined%
  \makeatother%
  \begin{picture}(1,0.33220006)%
    \lineheight{1}%
    \setlength\tabcolsep{0pt}%
    \put(0,0){\includegraphics[width=\unitlength,page=1]{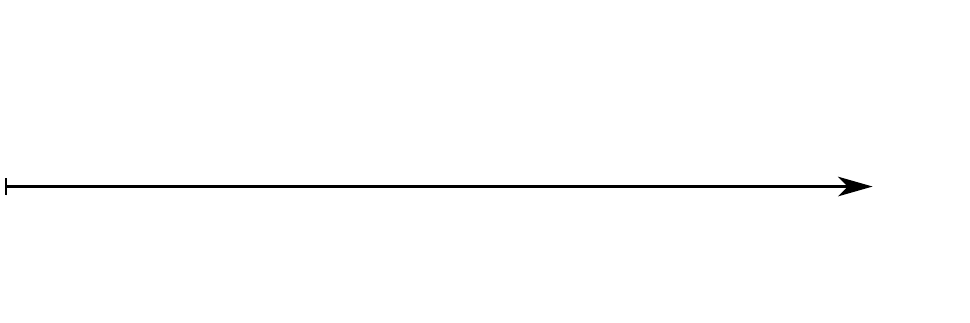}}%
    \put(0.91108059,0.13160861){\color[rgb]{0,0,0}\makebox(0,0)[lt]{\lineheight{0}\smash{\begin{tabular}[t]{l}$time$\end{tabular}}}}%
    \put(0,0){\includegraphics[width=\unitlength,page=2]{incremental.pdf}}%
    \put(0.00452229,0.31261085){\color[rgb]{0,0,0}\makebox(0,0)[lt]{\lineheight{0}\smash{\begin{tabular}[t]{l}$k$ solutions in time $O(k^an^b)$\end{tabular}}}}%
    \put(-0.00169958,0.1007042){\color[rgb]{0,0,0}\makebox(0,0)[lt]{\lineheight{1.25}\smash{\begin{tabular}[t]{l}0\end{tabular}}}}%
    \put(0,0){\includegraphics[width=\unitlength,page=3]{incremental.pdf}}%
    \put(0.36059798,0.00536311){\makebox(0,0)[lt]{\lineheight{1.25}\smash{\begin{tabular}[t]{l}delay in $O(k^{a-1}n^b)$\end{tabular}}}}%
  \end{picture}%
\endgroup%

\end{center}
\caption{An enumeration in incremental polynomial time. Delay guaranteed when Proposition~\ref{prop:UInca} is used.}
\end{figure}

\section{Polynomial Delay}\label{sec:delay}


If we want to have a regular enumeration process, then the delay between solutions should not change 
along the enumeration. The \textbf{delay of a machine}, on input $x$, is the maximum of the delay between two solutions it produces,
that is $\max_{1 \leq k < |A(x)|}|T(x,k+1) - T(x,k)|$. In this section, we bound the delay by a polynomial in the input size,
but we could require a stronger bound, a theme which is explored in Chapter~\ref{chap:low}. 

\begin{definition}[Polynomial delay]
A problem $\enum{A} \in \EnumP$ is in $\DelayP$ if it is solved by a machine $M$ with preprocessing and delay polynomial in its input size.
 \end{definition}
 
Observe that, by definition, $\DelayP = \UIncP_0$ and is thus equal to $\IncP_1$ and included in $\IncP$ by Proposition~\ref{prop:UInca}. Polynomial delay is the most usual notion of tractability in enumeration, both because it is a strong property (regularity and linear total time) and because it is relatively easy to obtain. Indeed, most methods used to design enumeration algorithms such as backtrack search with an easy extension problem~\cite{DBLP:journals/dmtcs/MaryS19}, or efficient traversal of a supergraph of solutions~\cite{LawlerLK80,avis1996reverse}, yield polynomial delay algorithms when they are applicable. Algorithms in $\DelayP$ have the same total time and the same incremental time as algorithms in $\IncP_1$, but they enumerate solutions more regularly as shown in Figure~\ref{fig:delay}. We further study in Chapter~\ref{chap:space} the relationship between $\DelayP$ and $\IncP_1$, when the space is polynomially bounded.

\begin{figure}
\label{fig:delay}
\begin{center}
\begingroup%
  \makeatletter%
  \providecommand\color[2][]{%
    \errmessage{(Inkscape) Color is used for the text in Inkscape, but the package 'color.sty' is not loaded}%
    \renewcommand\color[2][]{}%
  }%
  \providecommand\transparent[1]{%
    \errmessage{(Inkscape) Transparency is used (non-zero) for the text in Inkscape, but the package 'transparent.sty' is not loaded}%
    \renewcommand\transparent[1]{}%
  }%
  \providecommand\rotatebox[2]{#2}%
  \newcommand*\fsize{\dimexpr\f@size pt\relax}%
  \newcommand*\lineheight[1]{\fontsize{\fsize}{#1\fsize}\selectfont}%
  \ifx\svgwidth\undefined%
    \setlength{\unitlength}{465.41879248bp}%
    \ifx\svgscale\undefined%
      \relax%
    \else%
      \setlength{\unitlength}{\unitlength * \real{\svgscale}}%
    \fi%
  \else%
    \setlength{\unitlength}{\svgwidth}%
  \fi%
  \global\let\svgwidth\undefined%
  \global\let\svgscale\undefined%
  \makeatother%
  \begin{picture}(1,0.21013047)%
    \lineheight{1}%
    \setlength\tabcolsep{0pt}%
    \put(0,0){\includegraphics[width=\unitlength,page=1]{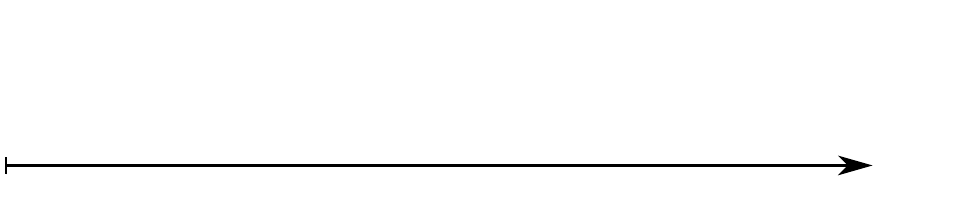}}%
    \put(0.91108056,0.03126949){\color[rgb]{0,0,0}\makebox(0,0)[lt]{\lineheight{0}\smash{\begin{tabular}[t]{l}$time$\end{tabular}}}}%
    \put(0,0){\includegraphics[width=\unitlength,page=2]{delay.pdf}}%
    \put(-0.00169958,0.00036508){\color[rgb]{0,0,0}\makebox(0,0)[lt]{\lineheight{1.25}\smash{\begin{tabular}[t]{l}0\end{tabular}}}}%
    \put(0,0){\includegraphics[width=\unitlength,page=3]{delay.pdf}}%
    \put(0.51643706,0.16673867){\makebox(0,0)[lt]{\lineheight{1.25}\smash{\begin{tabular}[t]{l}delay in $O(n^a)$\end{tabular}}}}%
    \put(0,0){\includegraphics[width=\unitlength,page=4]{delay.pdf}}%
    \put(0.03108613,0.19054126){\makebox(0,0)[lt]{\lineheight{1.25}\smash{\begin{tabular}[t]{l}preprocessing\\\end{tabular}}}}%
  \end{picture}%
\endgroup%

\end{center}
\caption{An enumeration in polynomial delay.}
\end{figure}

Contrarily to $\IncP$, there is no characterization of $\DelayP$ related to the complexity of some
search or decision problem, but there are some interesting implications. Let $\ext{A}$ be the problem 
of deciding, given $x$ and $y$, whether there is $y'$ such that $yy' \in A(x)$. 
Assuming the alphabet is $\{0,1\}$, then the set $A(x)$ can be divided into solutions beginning by $0$ and solutions beginning by $1$.
Then, $\ext{A}$ can be used to decide whether any of these subsets of solutions is empty. If $\ext{A} \in \P$, then this method of partitionning
the set of solutions, applied recursively on non-empty sets of solutions is an efficient version of backtrack search, often called \emph{binary partition} (see e.g.~\cite{DBLP:journals/dmtcs/MaryS19}). 
In this manuscript we call this method the \emph{flashlight search}, the flashlight stands for the use of $\ext{A}$ which ''illuminates'' a whole subtree of the search tree
to discard it if it contains no solution.

\begin{proposition}
 If $\ext{A} \in \P$, then $\enum{A} \in \DelayP$.
\end{proposition}

\begin{example}
As an illustration of the flashlight search, we use it to solve $\enum{\textsc{Union}}$.
Let us assume that a set over $n$ element is given by its characteristic vector. Let $y$ be a partial characteristic vector,
which can be represented by $A$ the set of positions of $1$ in $y$ (elements in the set represented by $y$) and $B$ the set of positions of $0$.
Extending $y$ such that $yy' \in \textsc{Union}(x)$ can be recast as finding a set $S \in \textsc{Union}(x)$ such that $A \subseteq S$ and $S \cap B = \emptyset$.
Let $\displaystyle{x^B = \cup_{s \in x,\, s \cap B = \emptyset} s}$, by construction it is the largest set of $\textsc{Union}(x)$ which has an empty intersection with $B$.
Hence, $\ext{\textsc{Union}}(x,y)$ is true if and only if $A \subseteq x^B$, and $\ext{\textsc{Union}}$ can thus be solved in time $O(mn)$, when
$x$ is a set of $m$ subsets of $[n]$.
The enumeration algorithm is a recursive algorithm, which does a depth first traversal of the implicit tree of partial solutions,
as represented in Figure~\ref{fig:binary}. The delay of the algorithm is the depth of the tree times the complexity of solving $\ext{\textsc{Union}}$,
that is $O(mn^2)$ but it can be improved to $O(mn)$ by amortizing the cost of solving $\ext{\textsc{Union}}$ on a branch, see~\cite{DBLP:journals/dmtcs/MaryS19}.
\end{example}

\begin{figure}
\begin{center}
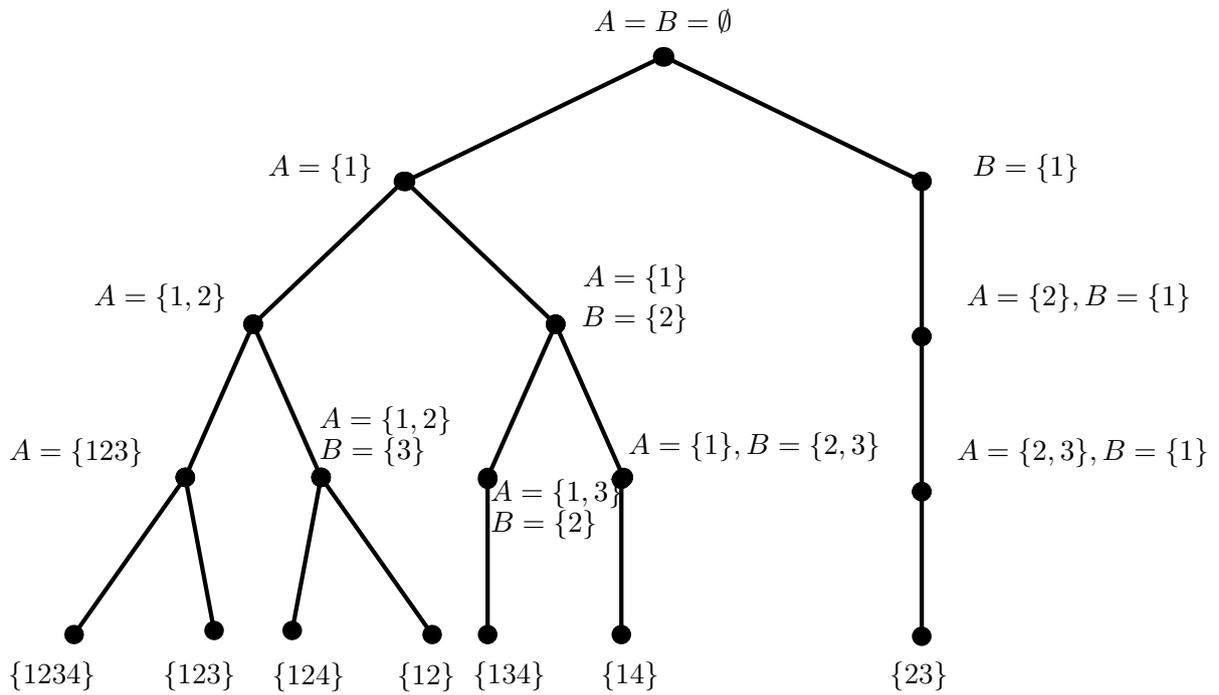
\end{center}
\caption{The tree of partial solutions of $\enum{\textsc{Union}}$, on the input $\{\{1,2\},\{2,3\},\{1,4\},\{1,3,4\}\}$. The branches containing no solution
are not explored by the algorithm nor represented in the figure.}
\label{fig:binary}
\end{figure}

The second method to solve an enumeration problem $\enum{A}$ in polynomial delay is to define 
a graph with set of vertices $A(x)$, often called the \emph{supergraph}. Assume that the supergraph satisfies the following properties:
\begin{itemize}
\item The supergraph is connected.
\item The neighborhood of a vertex can be produced in polynomial time.
\end{itemize}

Then, a depth first traversal of the graph allows to find all solutions with polynomial delay.  
Often, the supegraph is oriented, then to obtain a polynomial delay traversal, either it should be strongly connected
or we must be able to efficiently compute a set of vertices from which all other vertices can be reached. 
Several methods have been proposed to define strongly connected supergraphs, such as proximity search~\cite{conte2019new} or retaliation-free paths~\cite{cao2020enumerating}.

Contrarily to the flashlight method, the supergraph method requires as much space as there are solutions, since all elements of $A(x)$ must be marked (and thus kept in memory) to realize the traversal. To get rid of the space, one can try to define a spanning tree or a spanning forest over the supergraph. If the children and the parent of any vertex in the spanning tree can be found in polynomial time, then a traversal of the spanning tree is an enumeration algorithm and it does not need to store all solutions but only a single one.
This method, called the \emph{reverse search}, has been introduced by Avis and Fukuda~\cite{avis1996reverse} to solve problems like enumerating triangulations of points in the plane,
 spanning trees of a graph or topological orderings of an acyclic graph.

\begin{example}
The supergraph method can be used to solve $\enum{\textsc{Union}}$. 
Let $x = \{s_1, \dots, s_m\}$ an instance of $\enum{\textsc{Union}}$, with for all $i$, $s_i \subseteq [n]$.
We define a directed graph $G_x$ with for vertices, the elements of $\textsc{Union}(x)$ and $\emptyset$.
 There is an arc $(y_1,y_2)$ in $G_x$ if there exists $i$ such that $y_2 = y_1 \cup s_i$. Clearly, any element
in $\textsc{Union}(x)$ can be reached from $\emptyset$, hence traversing $G_x$ solves $\enum{\textsc{Union}}$
with polynomial delay. See Figure~\ref{fig:supergraph} for an example.

An implicit spaning tree of the supergraph can be defined as follows: $y_1$ is the parent of $y_2$
if it is the smallest lexicographic element such that $(y_1,y_2)$ is an arc of the supergraph. When traversing $G_x$, it is possible to skip all arcs which are
not in the parent relation, by computing the smallest set in $A(x)$ which can yield some $y$ by union with an element of $x$. 
This can be done in polynomial time, but with higher complexity than the binary partition method presented in the previous example.
\end{example}

\begin{figure}
\begin{center}
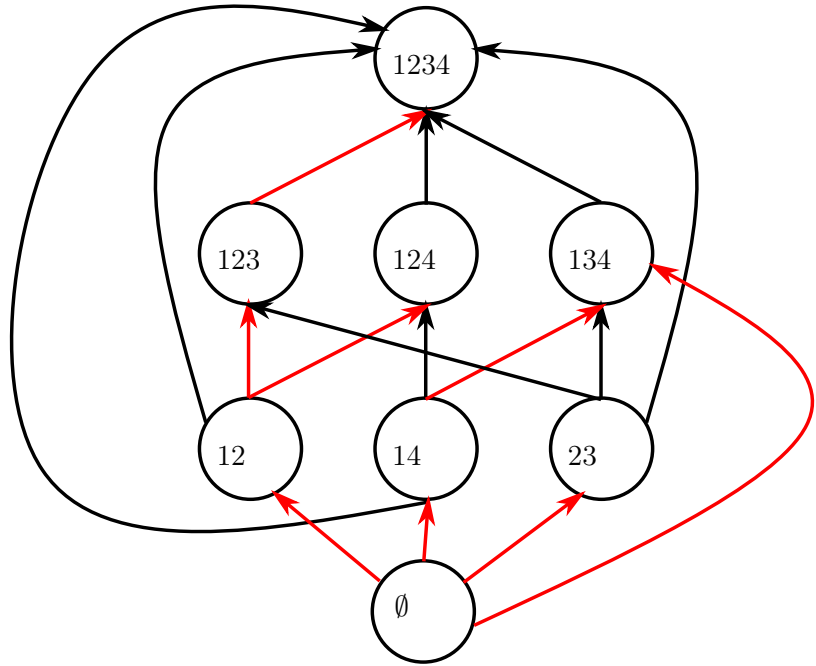
\end{center}
\caption{The supergraph of solutions of $\enum{\textsc{Union}}$, on the input $\{\{1,2\},\{2,3\},\{1,4\},\{1,3,4\}\}$. Arcs of the implicit spanning tree in red.}
\label{fig:supergraph}
\end{figure}

\section{Separation of Enumeration Classes}

In the previous sections, we have provided several separations, under various complexity hypotheses, 
which can be summarized by the following chain of inclusions: $$\DelayP = \IncP_1 \subsetneq \dots \subsetneq \IncP_{i} \subsetneq \dots \subset  \IncP  \subsetneq \OutputP \subsetneq \EnumP .$$

As explained in~\cite{capelli2019incremental}, these separations are unconditional if we remove the constraints that problems are in
$\EnumP$, i.e. that the solutions can be verified in polynomial time, by using the time hierarchy Theorem~\cite{hartmanis1965computational}.
  However, several properties do not hold in that case, such as the closure of $\DelayP$ under union presented in the next section or the equivalences between separation of enumeration complexity classes and $\P \neq \NP$ or  $\FP \neq \TFNP$.
  
Consider the $\textsc{SAT}$ problem, then $\ext{\textsc{SAT}}$ is equivalent to $\textsc{SAT}$, that is $\textsc{SAT}$ is autoreducible.
Hence, if $\P = \NP$, then $\ext{\textsc{SAT}} \in P$ and using the binary partition method, $\enum{\textsc{SAT}} \in \DelayP$, which yields the following proposition.

\begin{proposition}
$\P = \NP \Leftrightarrow \EnumP = \DelayP$.
\end{proposition}

The last proposition is somehow surprising, since it implies that separating $\EnumP$ from $\DelayP$ is as hard as separating it from $\OutputP$.

\section{Stability under Operations and Reductions}

Since an enumeration problem asks for a set of solutions, we can build new problems by acting on this set: the predicate
$A \cup B$ is defined as $(A \cup B)(x) = A(x) \cup B(x)$ and it defines the new problem $\enum{(A \cup B)}$. This definition can easily be adapted to any set operation such as the complement or the intersection.  We can also define the deletion or subtraction of $B$ from $A$ as $(A \setminus B)(x) = A(x) \setminus B(x)$, with the additional constraint that $B(x) \subseteq A(x)$. This is inspired by the subtractive reduction, defined for counting problems~\cite{durand2005subtractive}. If $B(x)$ is polynomial in $|x|$, then we say that $(A \setminus B)$ is a polynomial size deletion.

We say that a complexity class $\mathcal{C}$ is \emph{stable} under some operation, say the union, when  $\enum{A}\in \mathcal{C}$ and $\enum{B} \in \mathcal{C}$ implies $\enum{A \cup B} \in \mathcal{C}$. It turns out that $\DelayP$ is stable for several operations.

\begin{theorem}[\cite{Bagan09},\cite{phdstrozecki}]\label{th:stability}
The classes  $\OutputP$, $\IncP$, $\DelayP$ are stable under union, Cartesian product and polynomial size deletion.
\end{theorem}

The proof for cartesian product or polynomial size deletion are straightforward. The method for solving $A\cup B$ relies on a simple priority scheme: compute the next solution of $A$, if it is not a solution of $B$ then output it. If it is a solution of $B$, then output the next solution of $B$. This allows to deal with repetition of solutions and only 
multiply the delay by two, hence polynomial delay is ensured for a union of a polynomial number of problems in $\DelayP$. If the solutions of $A$ and $B$ are generated in the same order, then it is not necessary to test whether a solution of $A$ is also a solution of $B$, by dynamically merging the two ordered streams of solutions.

As a consequence, these operations can be used to decompose enumerations problems into simpler problems and thus 
may simplify the design of good algorithms. 

Several other operations do not leave any enumeration classes stable but $\EnumP$ because they allow 
to build an $\EnumP$-hard problem from simpler ones. For instance, deciding whether there is a model to the conjunction of a Horn Formula and an affine formula is $\NP$-complete~\cite{creignou2001complexity}, while enumerating the models of a Horn or an affine formula can be done with polynomial delay~\cite{creignou1997generating}.

\begin{theorem}[\cite{phdstrozecki}]
If $\P \neq \NP$ then $\DelayP$, $\IncP$, $\OutputP$ are not stable under deletion, intersection or complement. 
\end{theorem}

 Inspired by these operations, we would like to design the most general parsimonious reduction,
 so that $\DelayP$ is stable under its action. As explained in the previous chapter, to make $\enum{3SAT_{0}}$ complete, it is enough to generalize parsimonious reduction by allowing addition or deletion of a polynomial number of solutions themselves produced in polynomial time. $\DelayP$ is stable under such reductions because of Proposition~\ref{th:stability}. Several reductions for $\OutputP$, $\IncP$ and $\DelayP$ have been given by Mary in~\cite{mary2013enumeration}. We give here a polynomial delay reduction, which associates to each solution of some problem a subset of solutions of another problem such that all these subsets partition the desired set of solutions. 
 
 \begin{definition}[Polynomial delay reduction~\cite{mary2013enumeration}]\label{def:reduction}
  Let $\enum{A}, \enum{B} \in \EnumP$ and $f,g$ two polynomial time computable functions.
  The function $f$ maps instances of $\enum{A}$ to instances of $\enum{B}$, while $g$ maps a solution of $\enum{B}$ to a set of solutions of 
  $\enum{A}$.
  The pair $(f,g)$ is a polynomial delay reduction from $\enum{A}$ to $\enum{B}$ if there is a polynomial $p$ such that:
  \begin{itemize}
   \item $\bigcup_{s \in B(f(x))} g(s) = A(x)$ (all solutions of $A$ are obtained)
   \item $|\{ s \in B(f(x)) \mid g(s) = \emptyset \} | \leq p(|x|)$ (all but a polynomial number of solutions of $B$ give solutions of $A$)
   \item $\forall s_1,\,s_2 \in B(f(x)), \, g(s_1) \cap g(s_2) = \emptyset$ (no repetition of solutions)
  \end{itemize}
 \end{definition}
 
 Using this reduction, it has been proven that listing minimal transversals of a hypergraph
 is equivalent to listing the minimal dominating sets of a graph~\cite{kante2014enumeration}, a result which has motivated the study of enumeration of minimal dominating sets over various graph classes~\cite{kante2015polynomial,bonamy2019enumerating}.

 We can relax polynomial delay reductions by allowing a polynomial number of repetitions,
 such reductions are called e-reduction in~\cite{creignou2017complexity}.
 Instead of mapping a solution of $B$ to a polynomial number of solutions of $A$, we could even allow 
a polynomial delay algorithm to generate any number of solutions of $A$ from a solution of $B$.

To define a Turing reduction for enumeration, we give access to an oracle solving an enumeration problem using two additional instructions, one to begin an enumeration, the other to get the next solution. A reduction from 
$A$ to $B$ is a machine which solves $\enum{A}$ using an oracle to the problem $\enum{B}$. When the machine has a polynomial delay or an incremental polynomial time, the reductions are called D or I reductions~\cite{creignou2017complexity}.

No complete problems are known for the complexity classes we have introduced but $\EnumP$, whatever the reduction considered. In fact, artificial constructions using padding, as presented for 3SAT and parsimonious reduction, are often enough to build a hard but not complete problem. Moreover, when $\IncP_i$ is stable under a reduction, which is the case for all mentioned reductions but the Turing reductions,
there is no complete problem in $\IncP$ because it implies a collapse of the strict hierarchy given in Theorem~\ref{th:hierarchy}.

\begin{corollary}
Let $\prec$ be a reduction for enumeration problems such that, for all $a \geq 1$, 
$\IncP_a$ is stable under $\prec$. If $\ETH$ holds, there is no complete problem in $\IncP$ for $\prec$.
\end{corollary}

This lack of complete problems underlines the need to find other methods to prove hardness
of enumeration problems.

\chapter{Space Complexity}
\label{chap:space}
\minitoc

In the previous Chapter, only time restrictions have been considered,
but space plays a major role in practical enumeration algorithms and is often a bottleneck,
for instance when generating maximal cliques~\cite{DBLP:conf/icalp/ConteGMV16}.
We present several complexity classes, designed to capture the idea of bounded space and study their relations with their counterparts with unbounded space.

\section{Polynomial Space}

We could define an equivalent of $\PSPACE$ by considering all enumeration problems which can be solved
by a machine which uses a space polynomial in the input size. Such a class contains $\EnumP$, 
and all the polynomial hierarchy defined in~\cite{creignou2017complexity}.

As a tractability measure, it is more relevant to ask for algorithms with polynomial space and polynomial delay. In fact, several algorithmic methods have been designed, 
to realize the traversal of a supergraph of solutions with polynomial delay and \emph{polynomial space}~\cite{LawlerLK80,avis1996reverse,conte2019listing}.
We denote by $\DelayP^{poly}$ the class of problems solvable in polynomial delay and polynomial space, 
and by $\IncP^{poly}$ the class of problems solvable in incremental polynomial time and polynomial space.
When trying to prove lower bounds, it could be helpful to add the constraint
of polynomial space. For instance, we could try to prove that $\enum{\textsc{Min-Transversals}} \notin \DelayP^{poly}$, which should be easier than proving it is not in $\DelayP$.  

Several properties of $\IncP$ do not hold for $\IncP^{poly}$. The amortization method which allows to trade space for regularity as in Proposition~\ref{prop:UInca} does not work when requiring polynomial space. Hence, with polynomial space $\IncP_1^{poly}$ may not be equal to $\DelayP^{poly}$. We deal with this problem at the end of the chapter.
Moreover, it does not seem possible to characterize the class $\IncP^{poly}$ by a variant of the problem $\mathsf{AnotherSol}$, which makes it quite different from $\IncP$. 

A natural question is to understand whether the use of exponential space is useful for $\DelayP$ or $\IncP$. If we relax the constraint that the problems must be in $\EnumP$, then we can define a problem with an exponential number of trivial polynomial size solutions and one which is the solution to an $\EXP$-complete problem.  Such a problem is in $\DelayP$, since outputting the trivial solutions gives enough time to solve the $\EXP$-complete problem. However, if this problem is in $\DelayP^{poly}$, then the enumeration algorithm is also a polynomial space algorithm solving an $\EXP$-complete problem, which implies the unlikely collapse $\PSPACE = \EXP$.

\begin{openproblem}
Prove that $\DelayP \neq \DelayP^{poly}$ and $\IncP \neq \IncP^{poly}$ assuming some complexity hypothesis.
\end{openproblem}

\section{Memoryless Polynomial Delay}

 A notion of memoryless polynomial delay algorithm is proposed in~\cite{phdstrozecki}, under the name of strong polynomial delay, which is used in this thesis for another concept.
A memoryless algorithm can use only the last solution and the input to produce the next solution.

We say that a problem $\enum{A}$ is in $\NextP$ if for every instance $x$ there is a total order $<_x$ such that the following problems are in $\FP$:
\begin{enumerate}
 \item given $x$, output the first element of $A(x)$ for $<_x$
\item given $x$ and $y \in A(x)$ output the next element of $A(x)$  for  $<_x$ or  a special value if there is none 
\end{enumerate}

The enumeration of the minimal spanning trees or of the maximal matchings of a weighted graph 
are in $\NextP$ (see \cite{avis1996reverse,uno1997algorithms}). Many problems of queries over databases
are solved by implementing a Next function, which given a solution produces the next one (often in constant time), and a Start function which gives the first solution,  see for instance~\cite{DBLP:conf/pods/SchweikardtSV18}.
In fact, any problem which can be solved by traversal of an implicit tree of solutions, without storing global information, is in $\NextP$.  As a consequence, reverse search algorithms or flashlight algorithms are often memoryless. 

Several results proved for $\DelayP$ also hold for $\NextP$. If $\P=\NP$, then $\EnumP = \DelayP^{poly} = \NextP$, since $\enum{SAT}$ can be solved using a flashlight search.
Moreover, $\NextP$ is stable under cartesian product, disjoint union and polynomial size deletion, using the same methods as for $\DelayP$, but the method for proving stability under union is not memoryless.

\begin{openproblem}
Can we separate $\NextP$ from $\DelayP$ modulo some complexity hypotesis ? If so, can we separate $\NextP$ from $\DelayP^{poly}$, by taking advantage that a
$\DelayP^{poly}$ algorithm can store an information computed during the whole enumeration process, while a $\NextP$ algorithm uses a "local" information to find the next solution. 
\end{openproblem}

\section{Sources of Exponential Space}

There are three common sources of exponential space in enumeration.

The most common is that many algorithms generate each solution several times, a typical example is the algorithm used to generate
maximal independent sets in lexicographic order given in~\cite{johnson1988generating}. Any algorithm can be turned into an enumeration algorithm without repetition by storing all output solutions in a trie to detect and avoid repetitions. If the number of repetitions is bounded by some function $f(n)$, with $n$ the size of the instance, then the elimination of duplicates multiplies the incremental time of the algorithm with repetition by $f(n)$ only. However, this requires to store all output solutions, which may exponentially increase the space used. 

When an algorithm produces some solution $s$ at time $t$, to decide whether $s$ has already been output, we can run the enumeration from start to time $t$.  
This method requires no \emph{additional space}. If the original algorithm is in incremental time $k^ip(n)$, and the number of repetitions of a solution is bounded by $f(n)$,
then the algorithm without repetition is in incremental time $O(k^{i+1}p(n)f(n))$. Hence, a problem which is in $\IncP_i$ because of an algorithm using polynomial space and an exponential datastructure to detect duplicates in polynomial number, is also in $\IncP_{i+1}^{poly}$ using this method.
There is a tradeoff between incremental time and space using a hybrid of these two methods to avoid repetitions: store the last $l$ produced solutions and rerun the algorithm to test whether a solution is output only up to the point the first of the $l$ stored solutions is produced.

The second source of exponential space is the use of saturation algorithms: Since they compute new solutions from old ones, they require to store all produced solutions. There is no known generic method to get rid of the space in that case. However, it may be possible to design a different algorithm to solve the problem, see Chapter~\ref{chap:limits} where we show that a whole class of saturation problems can be solved with polynomial space and polynomial delay using a flashlight search.

\begin{openproblem}
The enumeration of the circuits of a binary matroid can be done in incremental polynomial time, using a saturation algorithm~\cite{khachiyan2005complexity}. Can we achieve the same complexity with a polynomial space only?
\end{openproblem} 

The third common source of exponential space is the amortization procedure of Proposition~\ref{prop:UInca}. It builds a polynomial delay algorithm from one in linear incremental time, by storing the produced solutions in a queue, which can be as large as the whole set of solutions. We show in the next section, that this source of exponential space can be eliminated.

\section{Amortization in Polynomial Space}

Let $M$ be a machine in linear incremental time, it produces $k$ solutions in time $kp(n)$. We say that $p(n)$ is the \emph{incremental delay} of $M$. 
The aim of this section is to show how to build a machine $M'$, which produces the same solutions as $M$, with a polynomial delay almost equal to the incremental
delay of $M$ and with a small space overhead.

\subsection{Gaps}

In~\cite{capelli2019incremental}, we show how to amortize an incremental linear time enumeration,
if it is regular enough. To capture the notion of regularity,
we define the notion of gap for a machine $M$:  $i$ is a \emph{$p$-gap} of $M$ (for $x$) if $T(x,i+1) - T(x,i) > p(|x|)$. 
If there is a polynomial $p$, such that $M$ has no $p$-gap, then $M$ has delay $p$. In fact, if $M$ has
only few $p$-gaps, it can be turned into a polynomial delay algorithm, see~\cite{capelli2019incremental}.

We can prove something more general, by requiring the existence of a large interval of solutions without $p$-gaps rather than bounding the number of gaps, as stated in Proposition~\ref{prop:amortization_regular}. It captures more cases, for instance an algorithm which outputs an exponential number of solutions at the beginning without gaps and then has a superpolynomial number of gaps. The idea is to compensate for the gaps by using the dense parts of the enumeration. It is done by running several copies of the incremental linear algorithm, and detecting intervals of time in the eumeration with many output solutions, so that it can be used to amortize intervals with fewer solutions. The proof is quite technical,
 see~\cite{capelli2019incremental}.

\begin{proposition}\label{prop:amortization_regular}
 Let $M$ be an incremental linear time machine solving $\enum{A}$ using polynomial space. 
 Assume there are two polynomials $p$ and $q$ such that for all $x$ of size $n$, and for all $k \leq |A(x)|$
 there exists  $a < b \leq k$ such that $b - a > \frac{k}{q(n)}$ and there are no $p$-gaps between the $a$th and the $b$th solution. Then $\enum{A} \in \DelayP^{poly}$.
\end{proposition}

\subsection{Geometric Amortization}

The method of the last section is not sufficient to prove $\IncP_1^{poly} = \DelayP^{poly}$. 
To do so, we introduce a new amortization method, which is also based on simulating the incremental linear algorithm at several points in time. It is simpler, more general and more efficient than the one designed to prove Proposition~\ref{prop:amortization_regular}.
 We give here the algorithm and a sketch of its analysis; for more details and related results, see~\cite{capelli2021amortization}. 

To simplify the presentation, we consider that the computation of the incremental linear machine $M$ on an input $x$ can be represented by a list $L_{M,x}$. 
If $M$ outputs a solution at time step $i$, then the $i$th element of $L_{M,x}$ contains this solution, otherwise it contains a special value denoting the absence of solution.
It is possible to go from an element of $L_{M,x}$ to the next one by simulating one computation step of $M$. To store a pointer to the $i$th element of $L_{M,x}$, we store the state of $M$ at time step $i$. 

We describe the complexity of an algorithm working on such a list as the number of reads of elements of $L_{M,x}$ for the time and the number of stored pointers for the space. In particular, for an algorithm outputting solutions of $L_{M,x}$, the delay of the algorithm is the maximal number of reads between the output of two distinct solutions or between the output of a solution and the end of the algorithm. We ignore the cost of maintaining counters, of switching between different contexts of computation etc. It can be shown that these costs amount only to a constant time overhead on a RAM machine~\cite{capelli2021amortization}.

 The pseudo code of the \emph{geometric amortization} method we propose is given in Algorithm~\ref{alg:enumit}. 
Let $M$ be the machine we amortize, it is in incremental delay $p(n)$ and outputs $\ell$ solutions on input $x$. We assume that $p(n)$ and $\ell$, or a bound on their values, are known.
 The idea is to maintain a list $p$ of $N := 1+\lceil \log(\ell)  \rceil$ pointers to elements of $L_{M,x}$. At the start of the algorithm, we initialize each $p[i]$ for $i<N$ to be a pointer to the first entry of $L_{M,x}$. We also initialize a table $c$ of $N$ counters, whose entries are initialized to $1$. Each time $p[i]$ is moved to the next entry of $L_{M,x}$, we increment $c[i]$ so that $p[i]$ always points toward $L[c[i]]$ (lines~\ref{line:moveptr}-\ref{line:incrc} in Algorithm~\ref{alg:enumit}).

Each $p[i]$ explores $L_{M,x}$. If at some point in the algorithm, $p[i]$ points toward a solution of $L_{M,x}$, and if $2^{i-1}p(n) < c[i] \leq 2^i p(n)$, then we output the solution (for $i=0$, we take the bounds $1 \leq c[0] \leq p(n)$). Thus, solutions will be output at most once since the index of the solution determines which pointer  discovers it.

We now describe the mechanism we use to move the different pointers. We move one pointer $p[j]$ for at most $2p(n)$ steps (corresponding to the loop starting at Line~\ref{line:loop2} in Algorithm~\ref{alg:enumit}). If we find a solution in the zone corresponding to $p[j]$, that is $[1+2^{j-1}p(n); 2^j p(n)]$, we output it and move back to pointer $p[N-1]$. If we move by more than $2p(n)$ steps without finding a solution in the right zone, we decrease $j$ by $1$ and continue the exploration. If $j$ is $0$ and no solution has been found in the zone of $p[0]$, that is $[1;p(n)]$, we stop the algorithm.

\SetKw{Break}{break}
\begin{algorithm}
 \KwData{A list $L$ containing $\ell$ solutions and an incremental delay $p(n)$}
 \Begin{
   $N \gets 1+\lceil \log(\ell) \rceil$\;
   Initialize a table $p$ of size $N$ containing $N$ pointers to $L[1]$\;
   Initialize a table $c$ of size $N$ containing $N$ integers set to $1$\;
   $j \gets N-1$\;
   \While{$j > 0$}{
     \For{$b \gets 2p(n)$ \textbf{to} $0$}{\label{line:loop2}
       $s \gets $ the value of the entry of $L$ pointed by $p[j]$\;
       Move $p[j]$ to the right\; \label{line:moveptr}
       $c[j] \gets c[j]+1$\; \label{line:incrc}
       \If{$s$ is a solution \textbf{\emph{and}} $1+2^{j-1}p(n) \leq c[j] \leq 2^jp(n)$}{
         Output solution $s$\;
         $j \gets N-1$\;
         \Break;
       }
     }
     \If{$b = 0$}{$j \gets j-1$}
    }
   }
   \caption{Pseudocode for the geometric amortization method}
  \label{alg:enumit}
\end{algorithm}

\begin{theorem}[\cite{capelli2021amortization}]\label{th:amortization}
Let $M$ be a machine in incremental delay $p(n)$, space $s(n)$. Algorithm~\ref{alg:enumit} on input $(L_{M,x},\ell,p(|x|))$ outputs all solutions of $L_{M,x}$ with a delay $O(\log(\ell)p(n))$, space $O(\log(\ell)(s(n)+\log(\ell)))$. 
\end{theorem}

\begin{proof}
(Rough sketch)
Algorithm~\ref{alg:enumit} uses $O(\log(\ell))$ pointers in the table $p$. Each of them can be stored with space in $O(S)$ where 
$S$ is the space used by the machine $M$.
It is easy to see by construction that the delay between two solutions is $O(\log(\ell)p(n))$. Indeed, the complexity of executing the loop starting at Line~\ref{line:loop2} is $O(p(n))$. Now, either a solution is found before $j$ reaches $0$, in which case, we have a delay of at most $O(Np(n)) = O(\log(\ell)p(n))$ since $j$ is decreasing and starts at $N-1$. Or $j$ reaches $0$ and no solution is found and this is the end of the algorithm, thus there is a delay of at most $O(\log(\ell)p(n))$ between the last output solution and the end of the algorithm. 

As an invariant of the algorithm, we prove that $c[i+1]$, that is the number of steps done by the $i+1$st pointers is larger than the number of solutions produced by the first $i$ pointers times $p(n)$. It can then be used to prove that Algorithm~\ref{alg:enumit} does not stop before outputing all solutions.
\end{proof}

Let us consider a slight variation of Algorithm~\ref{alg:enumit}, in which a pointer which has finished to explore its zone is removed from the list of pointers. In that case, on average, each entry of $L_{M,x}$ is read at most twice, which implies that the total time of the geometric amortization is the same as the original enumeration, up to a constant factor.

An implementation of the geometric amortization method called Coussinet is available on Florent Capelli's \href{http://florent.capelli.me/coussinet/coussinet.html}{website},
see Figure~\ref{fig:coussinet} for an illustration.

\begin{figure}
\includegraphics[scale = 0.2]{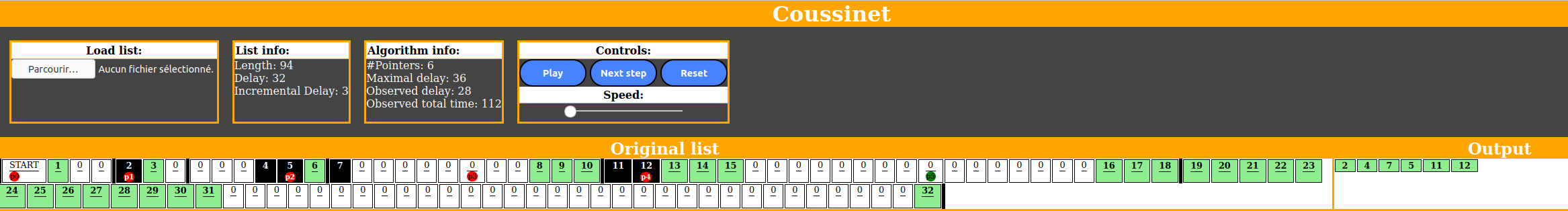}
\caption{The interface of Coussinet which takes an $\IncP_1$ list of outputs and does the amortization as in Theorem~\ref{th:amortization}.}
\label{fig:coussinet}
\end{figure}

As a corollary of Theorem~\ref{th:amortization}, we obtain a collapse of complexity classes, which was an open question in~\cite{capelli2019incremental}.

\begin{corollary}
$\IncP_1^{poly} = \DelayP^{poly}$.
\end{corollary}

As a consequence of this corollary, the right measure of tractability seems to be incremental linear time rather than polynomial delay. Indeed, it is not clear that polynomial delay is really necessary in applications and we have proved that if it is, we can obtain it from a linear incremental algorithm,
with a small space overhead and no change in the total time.
As an application of Theorem~\ref{th:amortization}, we amortize a flashlight search solving $\enum{\textsc{DNF}}$ to turn an average delay into a delay in Chapter~\ref{chap:low}.

\subsection{Amortizing without Assumptions}

To amortize a machine $M$ on an input $x$ thanks to Algorithm~\ref{alg:enumit}, we must know these values (or a bound on them):
\begin{itemize}
  \item the space used by the machine $M$ on $x$
  \item the number of solutions produced by $M$ on $x$
  \item the incremental delay of $M$ on $x$
\end{itemize}

We briefly explain how we can do without knowing the first two values, while 
the third one is more problematic (see~\cite{capelli2021amortization}).

The space used by $M$ is used implicitly in the algorithm to store the pointers to different 
parts of the list $L_{M,x}$. Indeed, a pointer represents the state of the machine at some point of its computation,
hence we need to store all used registers. If we have a bound on this space, we reserve enough space for each pointer at the beginning of the algorithm.
If we do not know how much space is used, we can use a dynamic array, with lazy initialization and expansion. The overhead is either logarithmic in time
or $n^{\epsilon}$ in space, for any $\epsilon$, depending on the expansion factor of the dynamic array.

The number of solutions in $L_{M,x}$ is an input of Algorithm~\ref{alg:enumit}, it is used to determine the number of pointers.
 Note that, in Algorithm~\ref{alg:enumit}, all pointers which are not yet in their zone move together, one after the other.
Hence, we could just move one of them to represent them all. In this variant, all pointers are in their zone, except the last one.
When the last one enters its zone, we can duplicate it to create a new pointer for the next zone. This algorithm mimicks Algorithm~\ref{alg:enumit},
but does not require the number of pointers beforehand. If the machine $M$ does not use more than $O(t)$
space after $t$ time steps, then the copy operation to obtain a new pointer is done lazily with a constant time overhead. 
If we have no guarantee on $M$, the time overhead is at most logarithmic.

Finally, we need the incremental delay of the machine $M$, to move each pointer in $L_{M,x}$ by twice this number of steps before going to the previous pointer.
In the classical amortization method using a queue, the incremental delay is also required, since it controls when an element is pulled out of the queue to be output. We prove that the classical amortization algorithm can be made to work, even when the incremental delay is unknown.
 The idea is to maintain a bound on the incremental delay of a machine $M$, while simulating it. We use a value larger than this bound, which depends on a fixed parameter 
$\epsilon$, to output elements out of the queue. This gives a polynomial delay algorithm, but the obtained delay is worse than the incremental delay of the original algorithm.

\begin{proposition}[\cite{capelli2021amortization}]
For any $\epsilon > 0$, and machine $M$ in incremental delay $O(p(n))$, there is a machine  $M_{\epsilon}$ which enumerates all outputs of $M$ in the same order, in delay $O(p(n)^{1 + \epsilon})$.
\end{proposition}

\begin{openproblem}
Can we use the method of dynamically evaluating the incremental delay in Algorithm~\ref{alg:enumit}?
Prove that there is a machine, which takes as input another machine in incremental linear time and polynomial space and enumerates the same solutions in
polynomial delay and polynomial space.
\end{openproblem}

When amortizing without knowing the incremental delay, it seems impossible to obtain a delay equal to the incremental delay.
To prove a lower bound, we assume that the amortization method can only work on $L_{M,x}$ and not on the machine $M$ and input $x$ directly. To further simplify, in the complexity
we only take into account a read in $L_M$ as an elementary operation, and we allow an unbounded memory. Since the memory is unbounded, every solution found while exploring $L_{M,x}$ can be stored in a queue to be later output, hence an amortization method only needs to go through $L_{M,x}$ once.
In these settings, an efficient amortization method explores the list $L_{M,x}$, stores the solutions in a queue and its only freedom is to determine when solutions from the queue should be output. This depends only on the number of elements of $L_{M,x}$ seen, the number of solutions encountered and the time since the last output. When such an amortization algorithm outputs solutions from the queue too fast, it is possible to build an adversary list and prove the following theorem.

\begin{proposition}[\cite{capelli2021amortization}]
There is no RAM machine, which given a list $L_{M,x}$ representing the computation of a machine $M$ on $x$ in incremental delay $p(|x|)$, enumerates all outputs of $M$ in delay $(s-2)p(|x|)$ where $s$ is the number of solutions generated by $M$.
\end{proposition}

\chapter{Below Polynomial Delay}
\label{chap:low}
\minitoc

\section{Strong Polynomial Delay}

When the input is huge with regard to a single solution, we would like a sublinear delay in the input size,
or even a delay which depends only on the size of a single generated solution. This situation naturally appears when the input is a hypergraph
and solutions are subsets of its vertices, or when solutions are of constant size, for instance when enumerating solutions of a first order query over a database.
In the following definition, a preprocessing \emph{polynomial in the size of the input} is allowed,
since the input should at least be read before outputting solutions.

\begin{definition}[Strong Polynomial delay]
A problem $\enum{A} \in \EnumP$ is in $\SDelayP$ if there exist constants $a\geq 0$ and $C>0$, and there is a machine $M$ which solves $\enum{A}$ for all inputs $x$ such that:
\begin{itemize}
  \item $T(x,1) \leq p(|x|)$, with $p$ a polynomial (polynomial preprocessing)
  \item for all $1 < k \leq |A(x)|$, $|T(x,k) - T(x,k-1)| < C|y_k|^a$, where $y_k$ is the $k$th solution output by $M$ (strong polynomial delay). 
\end{itemize} 
 \end{definition}

 Presently, very few problems have strong polynomial delay algorithms: generating $s-t$ paths in a DAG, generating assignments of existential $FO$ formulas with second order free variables~\cite{durand2011enumeration} or existential $MSO$ formulas over bounded tree width structures~\cite{DBLP:journals/dam/Courcelle09,AmarilliBJM17}.  
 When solving an enumeration problem with an infinite set of solutions~\cite{florencio2015naive}, any reference to the input size or the total time is irrelevant and strong polynomial delay is the proper notion of tractability. Moreover, a strong polynomial delay algorithm can be applied to large inputs given implicitly: this idea is used to give polynomial delay algorithms for generating dominating sets over several classes of graphs~\cite{blind2020locally}.

  Even an extremely simple problem like $\enum{DNF}$, the enumeration of models (or satisfying assignments) of a DNF formula, has no simple algorithm in strong polynomial delay. A formula in disjunctive normal form (DNF) is a disjunction of terms. Each term is a conjunction of litterals, that is variables or negation of variables. For instance $D = (X_1 \wedge \neg X_2) \vee (X_2 \wedge X_3)$ is a DNF formula with two terms over the variables $\{X_1,X_2,X_3\}$, and four models $\{(1, 0, 0),(1, 0, 1),(0, 1, 1), (1,1,1)\}$.
   The structure of the assignments of a DNF is extremely regular: it is the \emph{union} of the models of its terms. A term sets the values of some variables, and we can enumerate all possible values of the remaining variables in constant delay using Gray Code. The main difficulty to obtain an algorithm with strong polynomial delay for $\enum{DNF}$ is that \emph{the union is not disjoint} and the simple algorithm consisting in enumerating the solutions of each term  one after the other produces redundant solutions. Solution repetitions because of non disjoint union is a common problem in enumeration and this issue appears in its simplest form when solving $\enum{DNF}$. We hope that understanding finely the complexity of this problem and designing better algorithms to solve it will shed some light on other similar problems. We conjecture that $\enum{DNF}$ cannot be solved with strong polynomial delay.

\begin{conjecture}[DNF Enumeration Conjecture]
 $\text{\enumDNF} \notin \SDelayP$.
\end{conjecture}

We did even state a Strong DNF Enumeration Conjecture in~\cite{capelli2020enumerating}, similar in precision to $\ETH$,
stating that no algorithm solves the problem {\enumDNF} in delay $o(m)$,  where $m$ is the number of terms of
   the DNF. We show later in this section that this conjecture does not hold. 

To solve $\enum{DNF}$, a flashlight search can be used. Indeed, $\ext{DNF}$ can be solved in time $|D|$, where $D$
is the input DNF. In fact, the complexity of solving the extension problem can be amortized over a branch of the tree of partial solutions, so that the delay of the enumeration algorithm is $|D|$~\cite{capelli2020enumerating}.
It turns out that the flashlight method can be improved in several directions. First, the input DNF is stored as a trie~\cite{fredkin1960trie}, and when going down the tree of partial solutions, variables are fixed and the size of the trie is decreased. Second, we proved that a DNF with $m$ terms has at least $m^{\log_3(2)}$ models, then we use that fact to prove a good bound on the average delay as in Uno's push-out method~\cite{uno2015constant}.

Let $n$ be the number of variables of $D$. If the trie is shrunk very carefully when fixing variables, a factor $n$ can be shaved from the delay. We can further improve the algorithm and its analysis for the case of monotone DNFs, by further compacting the represented DNF. For $k$-DNF formulas, i.e. DNF formulas with terms of size at most $k$, we use a particular traversal of the tree of partial solutions. It allows to find a $k$-term, whose solutions are disjoint from the ones in the other branches of the tree. These solutions are output with constant delay (or implicitly) and can be used to amortize the time of going down the tree of partial solutions up to the point there are only $2k$ variables left. This yields a \emph{constant delay} algorithm for fixed $k$. The results of~\cite{capelli2020enumerating} are given in Table~\ref{tab:results}.

\begin{table}
  \centering
  \begin{tabular}{|l|l|l|}
    \hline
    {\bf Class} & {\bf Delay} & {\bf Space} \\ \hline
    DNF & $O(\|D\|)$  & Polynomial \\ \hline
    $(\star)$ DNF & $O(nm^{1-\log_3(2)})$ average delay & Polynomial \\ \hline
    $(\star)$ $k$-DNF  & $2^{O(k)}$ & Polynomial \\ \hline
    $(\star)$ Monotone DNF & $O(n^2)$  & Exponential \\ \hline
    $(\star)$ Monotone DNF & $O(\log(nm))$ average delay & Polynomial \\ \hline
  \end{tabular}
  \caption{In this table, $D$ is a DNF, $n$ its number of variables and $m$ its number of terms.}
  \label{tab:results}
\end{table}

We have obtained algorithms with good \emph{average delay}. To falsify the Strong DNF Enumeration Conjecture and 
the DNF Enumeration Conjecture on the class of monotone formulas, we need to replace the average delay by a similar delay.
To do that, we use a generic method~\cite{capelli2021amortization} to transform an average delay into a real delay
in flashlight algorithms, relying on the amortization method of Theorem~\ref{th:amortization}. 

Let the \textbf{path time} of a flashlight search algorithm be the time needed by a flashlight search algorithm to go from the root of the tree of partial solutions to a leaf. For most flashlight algorithms, path time and delay are the same. The main idea is that, at any point of a flashlight algorithm, we are exploring some vertex of the tree of partial solutions. The vertices which have already been explored
are whole subtrees, corresponding to subproblems of the enumeration problem being solved, plus a path from the root, as shown in Figure~\ref{fig:flashlight_amortize}.
This allows to precisely bound the incremental time of flashlight search. In particular, it is linear (up to the path time), if the average delay is polynomial, because the time spent in subproblems
is linear in the number of output solutions. If the path time is bounded by $p(n)$, the enumeration of solutions can be delayed by $p(n)$, using a queue of size at most $p(n)$, which gives an algorithm in incremental linear time. Then we can apply Theorem~\ref{th:amortization} to make the delay polynomial and we obtain the following theorem. 

\begin{theorem}[\cite{capelli2021amortization}]\label{th:self-reduction}
Let $\enum{A}$ be an enumeration problem, solved by a flashlight search algorithm, in polynomial space, polynomial path time and average delay $a(n,k)$, where $n$ is the size of the instance and $k$ a bound on the size of a solution. Let $N$ be the number of produced solutions, then there is an algorithm solving $\enum{A}$, with delay $a(n,k)\log(N)$, average delay $a(n,k)$ and polynomial space.
\end{theorem}

\begin{figure}
\begin{center}
\begingroup%
  \makeatletter%
  \providecommand\color[2][]{%
    \errmessage{(Inkscape) Color is used for the text in Inkscape, but the package 'color.sty' is not loaded}%
    \renewcommand\color[2][]{}%
  }%
  \providecommand\transparent[1]{%
    \errmessage{(Inkscape) Transparency is used (non-zero) for the text in Inkscape, but the package 'transparent.sty' is not loaded}%
    \renewcommand\transparent[1]{}%
  }%
  \providecommand\rotatebox[2]{#2}%
  \newcommand*\fsize{\dimexpr\f@size pt\relax}%
  \newcommand*\lineheight[1]{\fontsize{\fsize}{#1\fsize}\selectfont}%
  \ifx\svgwidth\undefined%
    \setlength{\unitlength}{400bp}%
    \ifx\svgscale\undefined%
      \relax%
    \else%
      \setlength{\unitlength}{\unitlength * \real{\svgscale}}%
    \fi%
  \else%
    \setlength{\unitlength}{\svgwidth}%
  \fi%
  \global\let\svgwidth\undefined%
  \global\let\svgscale\undefined%
  \makeatother%
  \begin{picture}(1,0.78194411)%
    \lineheight{1}%
    \setlength\tabcolsep{0pt}%
    \put(0,0){\includegraphics[width=\unitlength,page=1]{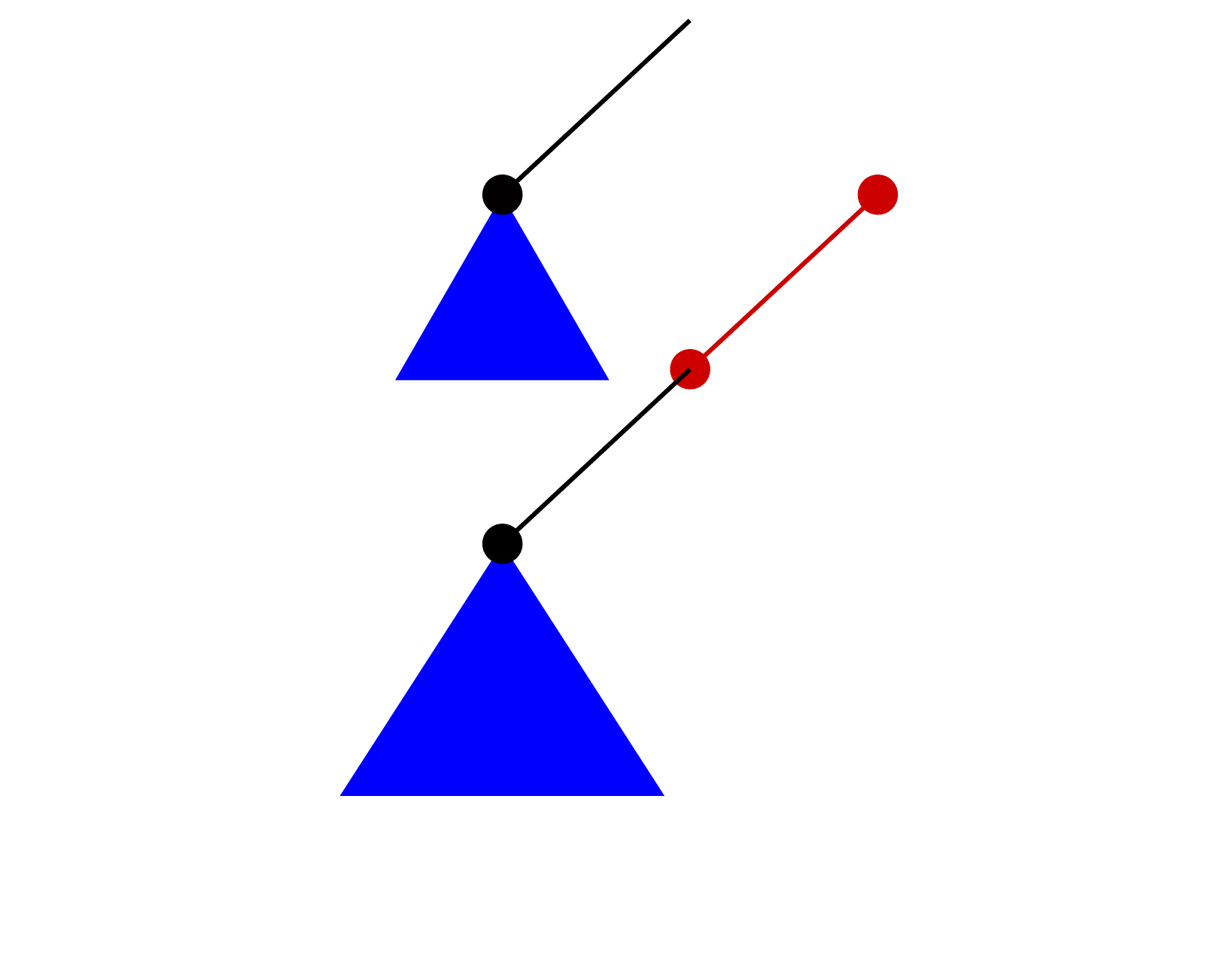}}%
    \put(0.59924179,0.17935316){\color[rgb]{0,0,0}\makebox(0,0)[lt]{\lineheight{1.25}\smash{\begin{tabular}[t]{l}Current subproblem\end{tabular}}}}%
    \put(-0.00735453,0.38929938){\color[rgb]{0,0,0}\makebox(0,0)[lt]{\lineheight{1.25}\smash{\begin{tabular}[t]{l}Enumerated subproblems\end{tabular}}}}%
    \put(0,0){\includegraphics[width=\unitlength,page=2]{flashlight_amortization.pdf}}%
    \put(0.62997965,0.30489999){\color[rgb]{0,0,0}\makebox(0,0)[lt]{\lineheight{1.25}\smash{\begin{tabular}[t]{l}$\dots$\end{tabular}}}}%
    \put(0,0){\includegraphics[width=\unitlength,page=3]{flashlight_amortization.pdf}}%
  \end{picture}%
\endgroup%

\caption{A traversal of the tree of partial solutions by the flashlight search. The subproblems completely solved recursively in blue, the path to the current partial solution in red and the unexplored vertices in the white triangles.\label{fig:flashlight_amortize}}
\end{center}
\end{figure}

\begin{openproblem}
We are left with a single unsettled conjecture, the DNF Enumeration Conjecture. 
Prove it assuming some complexity hypothesis like $\ETH$. As a first step, it should be proven that $\enum{DNF}$ cannot be solved in constant delay or logarithmic delay.  
\end{openproblem}

The problem $\enum{\textsc{Union}}$ of computing the closure by union of a set system has a linear delay algorithm~\cite{DBLP:journals/dmtcs/MaryS19}. Because a solution can be obtained by many different unions, $\enum{\textsc{Union}}$ does not seem to have a strong polynomial delay algorithm. In~\cite{capelli2020enumerating}, we \textbf{wrongly} claim that a method similar to the one presented for $\enum{DNF}$ on monotone formulas solves $\enum{\textsc{Union}}$ with average delay polynomial in the solution size. Indeed, the problem solved when traversing the tree of partial solutions is \emph{more complex} than the original problem,
because some elements are forced to be in the solution. This more complex problem can be solved in linear time, but it is not possible to prove a lower bound on its number of solutions with regards to the input size.

\begin{openproblem}
Give an algorithm for $\enum{\textsc{Union}}$ with average delay in $o(m)$, where $m$ is the number of sets in the input. 
\end{openproblem}

Instead of focusing on specific problems to understand the difference between $\SDelayP$ and $\DelayP$,
we could just prove a separation result using an artificial problem. However, the method to prove a hierarchy inside 
$\IncP$ does not seem useful to prove the same kind of hierarchy inside $\DelayP$. 

\begin{openproblem}
 Prove that $\SDelayP \subsetneq \DelayP$ modulo a complexity hypothesis like $\ETH$.
\end{openproblem}

If we drop the condition that problems must be in $\EnumP$, the previous question becomes easy,
as for the other class separations presented in Chapter~\ref{chap:time}.

\begin{proposition}\label{prop:sep_sdelayp}
There is an enumeration problem with an algorithm in polynomial delay and polynomial preprocessing but no algorithm in strong polynomial delay and polynomial preprocessing.
\end{proposition} 
\begin{proof}
Using the time hierarchy theorem~\cite{hartmanis1965computational}, we obtain a language $L$
which can be solved in time $O(n2^{\log(n)^2})$ but not in time $o(n2^{\log(n)^2}/\log(n)^2)$.
We let $n=|x|$ and we let $A$ be the binary predicate defined as  $A(x) = [2^{\log(n)^2}]$ if $x \notin L$ and  $A(x) = [2^{\log(n)^2}] \cup \{\sharp\}$ otherwise. The set $A(x)$ consists in $2^{\log(n)^2}$ trivial solutions and $\sharp$ when $x$ is in $L$.

 Assume now that we can enumerate the solutions of $\enum{A}$ in strong polynomial delay.
 Since the solutions are of size $\log(n)^2$, the total time to generate all solutions is less than
 $O(2^{\log(n)^2}p(\log(n)))$, with $p$ a polynomial. The polynomial time preprocessing does not appears in the total time since it is dominated by the generation time. Solving $\enum{A}$ allows to decide $L$, which cannot be done in time $O(2^{\log(n)^2}p(\log(n)))$, hence $\enum{A}$ cannot be solved in strong polynomial delay.

 However, $\enum{A}$ can be solved in linear delay.  Enumerate all trivial solutions in $A(x)$ with delay $n$. Since there are $O(2^{\log(n)^2})$ solutions, the total time to generate these solutions is in $O(n2^{\log(n)^2})$, hence we can use this time to decide whether $x \in L$ and to produce the last solution.
\end{proof}

\section{Constant Delay}

The class of problems solvable with a constant delay is very sensitive to the way the computational model and the complexity measures are defined. There are at least four slightly different notions of constant delay used in different contexts. In all cases, we recall that operations on integers less than $n$ are in constant time, with $n$ the size of the input.

 \paragraph{Gray code} 

 In several RAM models, the output registers are fixed and an output solution is always contained in the first output registers. Hence, two consecutive solutions must only differ by a fixed number of registers, such an ordering of solutions is often called a Gray code.  If the solution is represented by $0$ and $1$ in the output registers, it means that the Hamming distance between
 consecutive solutions is constant. In our computation model, the registers may contain integers less than $n$ that can be manipulated in constant time, hence the Hamming distance between two consecutive solutions is bounded by $O(\log(n))$.  Since the algorithm goes from a solution to the next in constant time, it is sometimes called a loopless algorithm in the context of Gray codes. There are many such algorithms for the generation of combinatorial objects~\cite{walsh2003generating,ruskey2003combinatorial,knuth2011art,hartung2019combinatorial}: the integers less than $2^n$, combinations, Dyck words, \dots 
 
\begin{example}
Let $RBC_n$ be the reflected binary code order on $n$ bits words, represented as a list. It is defined recursively 
as the concatenation of the words of $RBC_{n-1}$ prefixed by $0$ followed by the words of $RBC_{n-1}$ in reverse order prefixed by one.
To go from one word to the next in $RBC_n$, the $1$ in the current word should be maintained as a list as well as the parity of the number of ones. If the number of ones is even, the last bit is flipped,
otherwise the bit to the left of the rightmost one is flipped. Since the word is represented as an array of size $n$, all these operations can be done in constant time. 
\end{example}

 \paragraph{Constant delay with dynamic amortization} 

 To generate very different consecutive solutions in constant delay, we must relax the previous definition.
 For instance, it seems natural when several solutions are already computed in memory, to be able to output them in constant delay.
 This is why we have defined RAM machines for enumeration with a generalized OUTPUT instruction, which can output any region of the memory. 
 It allows to output already computed solutions in constant delay or to amortize a slow enumeration process with a fast one. It can be used to generate models of a $k$-DNF in constant delay~\cite{capelli2020enumerating}. 
 This form of constant delay can be obtained even when the solutions do not follow a Gray code order. Hence, it is strictly more general than the previous definition.

 \paragraph{Constant amortized time} 

 An algorithm is in constant amortized time (CAT) if its total time divided by the number of solutions is constant, i.e. its average delay is constant. 

 \begin{example}
 Consider the generation of the integers between $0$ and $2^{n}-1$ in the usual order, represented as an array of bits. Adding one to an integer may change up to $n$ bits, but on average only two bits are changed and the time needed is proportional to the number of changed bits. Hence, on average over the whole enumeration, at most two operations per solution is required: it is a CAT algorithm. If we generate the integers in another order like the Gray code order, it possible to obtain a constant delay algorithm~\cite{knuth2011art}.
 \end{example}
  
  See Ruskey's book~\cite{ruskey2003combinatorial} for many other examples of CAT algorithms. 
  Uno's push-out method helps turn backtracking algorithms into CAT algorithms, see~\cite{uno2015constant} for the method and several of its applications.
 
 Several problems admitting a CAT algorithm, such as the enumeration of unrooted trees~\cite{wright1986constant},
 have later been proved to admit a proper constant delay algorithm~\cite{nakano2003simple}.
 Amortization may be used to turn a CAT algorithm into a constant delay algorithm if the generation of solutions is
 regular enough. Sometimes, the solutions to be generated are organized as nodes of a tree and given 
 by some traversal of the tree. Then, by generating solutions of even depth when going down the tree and solutions of odd depth when going up, one may turn a CAT algorithm into a constant delay one without relying on amortization. This method is sometimes called Feder's trick, as he introduced it to generate the models of a 2CNF~\cite{feder1994network}.

 \begin{openproblem}
 Give a sufficient condition to turn CAT algorithms into constant delay algorithms, possibly using exponential space. The method of Theorem~\ref{th:self-reduction} transforms a flashlight search in constant average delay into a (strong) linear delay algorithm: under which condition can it be made constant?
 \end{openproblem}

 Constant amortized time is stricly more general than constant delay,
  as shown in the next proposition.

 \begin{proposition}
 There is $\enum{A}$ which can be solved in constant amortized time and polynomial preprocessing but not in constant delay and polynomial preprocessing. 
 \end{proposition}
 \begin{proof}
  Let $L$ be a problem, which can be solved in $O(2^{n})$ but not in $o(2^{n}/n)$. We define $A(x)$ as $\{0i \mid i\leq |x|\}$ if $x \in L$ and $\{1i \mid i\leq |x|\}$ otherwise. 
  The problem can be solved in constant amortized time, by computing whether $x \in L$ in time $O(2^{|x|})$ and then generating all solutions using Gray code in total time $O(2^{|x|})$. If $\enum{A}$ can be solved in constant delay and polynomial time preprocessing, then the first solution is generated in polynomial time, which implies that $L$ is solved in polynomial time, a contradiction.
 \end{proof}
 
  \paragraph{Fixed parameter constant delay}

   When doing query evaluation, the query is often considered to be of fixed size, we say that we consider the \emph{data complexity}. Hence, the number of free variables of the query is fixed, let us denote it by $k$. A solution is a $k$-tuple of elements from the structure on which the query is evaluated. Let $n$ be the size of the structure, which is the instance of the problem, then $k$ integers of size $\log(n)$ bits are necessary to encode a solution. Since we have assumed a uniform cost measure, these elements can be dealt with in constant time. 
   Note that there are at most $n^k$ solutions, hence to make the class interesting, it is not enough to ask the preprocessing to be in polynomial time. The choice is often to consider a linear (or almost linear) preprocessing. 
   Since the solutions are of a fixed constant size, we could use for these problems the notion of \emph{strong polynomial delay}, which is more demanding than constant delay: the delay should be 
   polynomial in $k$, while it can be any function in $k$ for (fixed parameter) constant delay.

   Durand and Grandjean have shown that $FO$ queries over bounded degree structures can be enumerated with constant delay and linear preprocessing~\cite{DBLP:journals/tocl/DurandG07}. Following that result, many query languages have been proven to have a constant delay algorithm, see the survey of Segoufin~\cite{DBLP:journals/sigmod/Segoufin15} and Durand~\cite{durand2020fine}. 
 
  The class of enumeration problems solvable with constant delay is stable under cartesian product.
  However, to be stable under union, either the union should be disjoint, or the $\test{A}$ problem must be in constant time.
  These observations can be generalized by considering circuits of disjoint unions and cartesian products, representing a set of solutions,
  which can be enumerated with linear preprocessing and constant delay when the solutions are sparse. This approach has been investigated
  in~\cite{AmarilliBJM17}, and can be applied to show that assignments of $MSO$ formulas, with free variables of the first order only, can be enumerated with constant delay over structures of constant tree-width.

\section{Random Sampling}

As a way to escape the time and space required to generate all solutions, we could only generate one.
We may then ask for a solution at some position in a given order, an approach 
considered in Chapter~\ref{chap:order} on ordered enumeration problems. 
An approach, which does not require to have an a priori on the generated solutions, is to generate a representative solution, which can be achieved by
uniformly sampling the set of solutions. A \textbf{polynomial time uniform generator} for the problem $\enum{A}$ generates an element of $A(x)$, with uniform probability, in time polynomial in $|x|$. The existence of such generator is deeply connected to the possibility of approximately counting the set of solutions, see~\cite{jerrum2003counting}.

The space of solutions of an enumeration problem can often be represented by some connected supergraph.
As explained in Chapter~\ref{chap:time}, a traversal of the graph produces all solutions but often requires 
exponential space. If there is a random walk on the graph of solutions, whose stationary distribution (probability of presence on each vertex in the limit) is almost uniform,
then it can be used as a random generator. When the random walk mixes in polynomial time, i.e. approximates its stationary distribution up to a constant in polynomial time,
then we obtain a polynomial time random generator, which avoids the exponential space use of the supergraph method. 

When the space of solutions to enumerate is infinite and even continuous, sampling is often the only approach.
In~\cite{peyronnet2012approximate}, we try to decide when two non deterministic automata or two Markov decision processes are far enough for some distance.
When they are different, we want to explain to the user why these automata or MDPs are different, hence we want to enumerate words or strategies which prove their separation.
For complexity reason, we use an approximate representation by the polytope of the set of accepted words or the polytope of the stationary distributions of the strategies. There are an \emph{infinite number}
of points in the symmetric difference of two polytopes, hence instead of enumerating them, we sample them uniformly, using a polynomial time random walk on polytopes~\cite{lovasz1993random}.

 A polynomial time uniform generator can be turned into an exhaustive enumeration algorithm
 solving the enumeration problem with \emph{high probability}. The algorithm is simple:
draw repeatedly solutions using the uniform generator. The only difficulty is 
to know when to stop drawing solutions because there is no new solution anymore. The time between the last new solution generated
and the end of the algorithm can be bounded using the Coupon Collector Theorem~\cite{erdos1961classical}, so that the algorithm works with high probability. This methods has been first presented in~\cite{Goldberg91} and improved in~\cite{capelli2019incremental}.

\begin{theorem}
\label{thm:gentopolydelay}
 If $\enum{A} \in \EnumP$ has a polytime uniform generator, then $\enum{A}$ can be solved with high probability by an algorithm in polynomial delay.
\end{theorem}

The algorithm used in Theorem~\ref{thm:gentopolydelay} needs an exponential space for storing the solutions and 
for the amortization of the incremental linear time. In general, we cannot get rid of the space, since Goldberg has given a tight lower bound on the product of space and delay~\cite{Goldberg91}. However, if \emph{repetitions of solutions} are allowed, we can build an exhaustive algorithm using polynomial space and incremental linear time~\cite{capelli2019incremental}. To do that, we must decide when to stop drawing solutions, without being able to maintain explicitly the set of found solutions. There exist data structures which approximate the cardinality of a dynamic set using only a logarithmic number of bits in the size of the set~\cite{flajolet1985probabilistic,kane2010optimal}. The idea is to apply a hash function to each element seen and to remember an aggregated information on the bits of the hashed elements. Such data structures can be used to design a randomized incremental linear algorithm from a uniform generator. 

\begin{theorem}~\label{th:gentoinclin}
 If $\enum{A} \in \EnumP$ has a polytime uniform generator, then $\enum{A}$ can be solved with high probability by an algorithm in incremental linear time \textbf{with repetitions} and \textbf{polynomial space}.
\end{theorem}

Note that Theorem~\ref{th:amortization} does not allow to amortize incremental linear time into a polynomial delay in Theorem~\ref{th:gentoinclin}. Indeed, the algorithm is randomized and the different copies of the enumeration algorithms used in the amortization would produce different solutions! 

\begin{openproblem}
There are several methods to obtain a uniform generator: random walk on a reconfiguration graph~\cite{nishimura2018introduction} or Boltzman samplers~\cite{duchon2004boltzmann}. Can we get an efficient enumeration algorithm directly from these methods without building the random generator?
\end{openproblem}

\chapter{Ordered Enumeration}
\label{chap:order}
\minitoc

In this chapter we briefly explore ordered enumeration, that is producing a set of solutions in some order.

\section{Specified Order}

 \paragraph{K-best solutions} One way to deal with a large solution space is to reduce it arbitrarily.
 We can fix a parameter $k$ and ask to enumerate only $k$ solutions. To make this approach more interesting, we can fix some relevant order
 on the solutions and ask for the $k$ best ones. When $k$ is not fixed but given as input the problem
 is very similar to the (ordered) enumeration of solutions with respect to incremental time. The $k$-shortest path and the $k$-minimum spanning tree problems are typical $k$-best problems with many applications as explained in a survey by Eppstein~\cite{eppstein2015k}. In this survey, methods for generating $k$-best solutions are presented, which are variations on methods also used for enumeration: binary partition or the $I\cup\{j\}$ method~\cite{LawlerLK80}.

\paragraph{Fixed Order} 
 Enumeration problems can be further constrained by requiring the solutions to be produced in a given order. The order can either be given in the input as a black box or with 
 a succinct definition, e.g. when the solutions are sequences of elements, giving an order on the elements yields a lexicographic order on solutions. Alternatively,
 the order (or family of orders) can be fixed by the problem, which is the point of view adopted here. We also require that the order can be decided in polynomial time. Since partial orders are allowed, it is a strict generalization of enumeration problems as defined in this thesis. The complexity measures and complexity classes are defined for ordered enumeration problems as they have been defined 
 for regular enumeration problems. From a practical point of view, ordered enumeration is a way to deal with a large number of solutions by requesting the enumeration of the best ones first and it relates to the problem of finding the $k$-best solutions, presented in the previous paragraph.

 From a structural complexity perspective, adding an order constraint has no effect on the class $\OutputP$ since the solutions can be sorted after having been output in time polynomial in the solution set, as remarked in~\cite{johnson1988generating}. However, this method stores all solutions before sorting them, which may require exponentially more space than the algorithm which does not respect order.
 This space can be traded for time, by doing a selection sort. The algorithm solving the unordered enumeration problem is run to produce all solutions and we maintain $s$ the smallest
 produced solutions during its run. At the end of the run, $s$ is output, and the algorithm is run again to find the smallest solution larger than $s$ and so on. 
 This method multiplies the total time by the number of solutions, hence if the original algorithm is in $\OutputP$, the algorithm for the ordered problem is also in $\OutputP$ without using more space.

 It is easy to enumerate solutions in lexicographic order, when there is a flashlight search for the problem, by doing the traversal of the tree of partial solutions
 following the order.  However, when there is no flashlight search, as for the generation of maximal independent set, a lexicographic order may dramatically change the complexity.
  Maximal independent sets can be enumerated in lexicographic order with polynomial delay but not in reverse lexicographic order~\cite{johnson1988generating} if $\P \neq \NP$! 
  Likewise, maximal acyclic subhypergraphs of a hypergraph can be enumerated in polynomial delay~\cite{daigo2009generating} but finding the first lexicographic maximal acyclic subhypergraph is $\coNP$-hard,  therefore this enumeration problem in lexicographic order is not in $\IncP$. 
  It turns out that many ordered enumeration problems are not in $\IncP$ because finding the first solution is $\NP$-hard. In particular when the order is the solution size, the enumeration of many kinds of subgraphs is hard: 
  independent sets, dominating sets, cycles \dots The Hamming weight of a solution also induces an interesting order in the context of assignements of formulas, see~\cite{DBLP:conf/sat/CreignouOS11,DBLP:conf/lata/Schmidt17}

  Hardness results with a fixed order are often directly derived from the hardness of finding some largest solution and are not specific to enumeration. Indeed, even trivial enumeration problems become hard when paired with the appropriate order (see Section $2.4$ of~\cite{phdstrozecki}). We design a total order $<$, such that $a<b$ can be decided in polynomial time, but the enumeration of $2^n$ trivial solutions of an instance
  of size $n$ is not in incremental polynomial time in this order, while the unordered problem can be solved in constant delay. 
  Hence, separations of classes of ordered enumeration problems are simpler and cleaner, as stated in the following proposition.

\begin{proposition}[\cite{phdstrozecki}]
For ordered enumeration problems, $\P \neq \NP$ is equivalent to any of these separations
 $$\DelayP \subsetneq \IncP \subsetneq \OutputP \subsetneq \EnumP.$$ 
\end{proposition}

\section{Ranked Solution}

 To avoid to deal with the time and space needed to generate the whole solution set, the best is to generate solutions on demand. 
 Let $A$ be a binary predicate, let $<$ be some total order on $A(x)$, then $\jsol{(A,<)}$ is the problem of computing 
 $y$ from the instance $(x,j)$, such that $y$ is the $j$th element of $A(x)$ according to $<$ or a special symbol if $|A(x)| \leq j$.
 Problem $\jsol{(A,<)}$ is also called the $j$th solution problem~\cite{Bagan09}. From that problem, we define the complexity class $\QueryP$ (see~\cite{phdstrozecki}), which is the set of enumeration problems $\enum{A}$, for which there is an order $<$ such that $\jsol{(A,<)} \in \P$. Note that $\QueryP$ is a restriction of the class $\NextP$: instead of accessing the solutions as in a list by a next pointer, we access them directly as if they were stored in an array.

\begin{example}
We define $\textsc{LinearSystem}((A,b),x)$ the predicate which is true if and only if $Ax = b$, with $A$ a $n*m$ matrix over some finite 
field $\mathbb{F}$ and $x$ and $b$ vectors. We show that $\enum{\textsc{LinearSystem}} \in \QueryP$.
 To solve $\enum{\textsc{LinearSystem}}$, one should output all solutions of a linear system, that is a vector space. 
A basis $(y_{1},\dots,y_{k})$ of the vector space can be computed in polynomial time, for instance using Gaussian elimination.
The set of solutions is the set of vectors equal to $\lambda_{1}y_{1} + \dots + \lambda_{k} y_{k}$
for all $(\lambda_{1},\dots,\lambda_{k}) \in \mathbb{F} ^{k}$. We fix an arbitrary order on $\mathbb{F}$, and consider the associated lexicographic order $<$
on the solutions seen as sequences $(\lambda_{1},\dots,\lambda_{k})$. The problem $\jsol{(\textsc{LinearSystem},<)}$ can be solved in polynomial time
since the $j$th element of $<$ is easy to compute in polynomial time and from that, the solution is obtained by a linear number of scalar multiplications and vector additions.
\end{example}

  If $\jsol{(A,<)} \in \P$, the problem of counting the solutions is also in polynomial time. Reciprocally, if an enumeration problem is solved by a flashlight search, and if counting its solutions is in polynomial time, then the problem is in $\QueryP$, using the lexicographic order yielded by the traversal of the tree of partial solutions. 
  Only very simple problems are in $\QueryP$, for instance listing the solutions of a first-order query over a bounded degree structure~\cite{bagan2008computing}, or monadic second-order query over trees or graphs of bounded branch-width~\cite{DBLP:journals/dam/Courcelle09}. The order used for the case of second-order queries is not natural and is a byproduct of the algorithm.
  For first-order queries and lexicographic ordering, the class of problems for which $\jsol{(A,<)}$ can be solved in logarithmic time is characterized using standard complexity assumptions~\cite{carmeli2021tractable}.

  Counting the perfect matchings of a graph is a classical $\sharp\P$-complete problem~\cite{valiant1979complexity}, but their enumeration can be done in polynomial delay~\cite{uno1997algorithms}, hence if $\FP \neq \sharp\P$ then $\QueryP \neq \DelayP$. The same reasonning holds for $\enum{DNF}$, which is in $\NextP$ because of the flashlight search solving it in lexicographic order,
  but counting the solutions of a $DNF$ is $\sharp \P$-complete.
 
  In~\cite{carmeli2020answering} the problem of generating solutions of conjunctive queries is studied in details under three different constraints:  without order, in a random order and with random access, that is $\QueryP$ (unspecified order). By turning a uniform generator into an exhaustive enumeration algorithm by repeated sampling as in Theorem~\ref{thm:gentopolydelay}, it is easy to generate solutions in random order. However, designing a specific algorithm to generate solutions in random order seems more efficient in practice than randomly sampling solutions~\cite{carmeli2020answering}.
  Since any problem in $\QueryP$ has a random uniform generator in polynomial time, generating in random order is easier than random access.
  The sets of free-connex acyclic conjunctive queries, which can be enumerated/sampled in polylogarithmic delay/time after
  linear preprocessing form a strict hierarchy, showing that these notions are different even for a very restricted problem.

\chapter{Restricted Formalisms and Classification}
\label{chap:limits}
\minitoc

We want to precisely characterize the complexity of enumeration problems. We report the very few known lower bounds and hardness results. We then study restricted formalisms of several kinds to be able to prove stronger results. We either restrict the family of enumeration problems we consider, or we weaken the computation model, or we even fix the algorithmic method used to solve an enumeration problem. This allows to prove lower bounds, to classify problems according to their complexity or to give simpler characterizations of their complexity. We also present several interesting algorithms, often derived from flashlight search, to capture some tractable cases in the several presented formalisms.

\section{Hardness and Lower Bound Results}

All problems derived from $\NP$-complete problems are in $\EnumP$ but not in $\OutputP$ because deciding whether there is a solution is hard. In this section, we are interested in problems hard for lower complexity classes like $\IncP$, for which deciding whether there is a solution is easy.

\paragraph{Hard problems}

Recall that, by Proposition~\ref{prop:anothersolincp}, a problem $\enum{A} \in \IncP$ if and only if 
$\mathsf{AnotherSol_A} \in \FP$. Hence, if we assume $\P \neq \NP$ and that $\mathsf{AnotherSol_A}$ is $\NP$-hard, then $\enum{A} \notin \IncP$. This has been used to prove that several problems are $\IncP$-hard: listing maximal models of a Horn formula~\cite{kavvadias2000generating}, listing vertices of a polyhedron~\cite{khachiyan2006generating} and dualization in lattices given by implicational bases~\cite{defrain2019dualization,defrain2021dualization}.

\paragraph{Class from a problem}

To prove hardness results, we may assume complexity hypotheses on decision problems as in the previous 
paragraph. An alternative is to base hardness of an enumeration problem on an hypothesis on the hardness of another \emph{enumeration problem}. The best example of this approach are the many problems which have been proven to be as hard or equivalent to enumerating the minimal transversals of a hypergraph, a problem not known to be in $\OutputP$ after $40$ years of research. Among these problems the dualization of a monotone boolean formula, the generation of vertices of an integral polytope~\cite{boros2009generating}, the enumeration of minimal keys, the enumeration of minimal dominating sets~\cite{kante2014enumeration}, see~\cite{eiter1995identifying,hagen2008algorithmic} for many examples.  In fact, a recent article~\cite{defrain2019dualization} classifies dualization problems as either in $\IncP$ or as hard as the enumeration of minimal transversals.  

We can also use the conjectures made on the complexity of {\enumDNF} in Chapter~\ref{chap:low} as a source of hardness. The class $\SDelayP$ is stable under the reduction given in Definition~\ref{def:reduction}.
Hence, if {\enumDNF} can be reduced to some problem $\enum{A}$, it is a clue that $\enum{A}$ might not be in $\SDelayP$. We have done such kind of reduction from {\enumDNF} over monotone formulas in~\cite{DBLP:journals/dmtcs/MaryS19}, but it does not give much information, because the problem has now been proved to be in $\SDelayP^{poly}$~\cite{capelli2020enumerating,capelli2021amortization}.

\begin{openproblem}
Find a good problem on which to base hardness for $\DelayP$. 
It should be in $\IncP_k$, with $k>1$. Plausible candidates are enumeration of minimal hitting sets in $3$-uniform hypergraphs or enumeration of the circuits of a binary matroid.
\end{openproblem}

\begin{openproblem}
 Can we transfer hardness between classes?  Prove that some problem, for instance {\enumDNF}, is not in $\SDelayP$ if and only if some other problem, for instance generating minimal transversals,
 is not in $\IncP$.
\end{openproblem}

 \paragraph{Strong complexity hypothesis} 

 Instead of assuming $\P \neq \NP$, we may need to make stronger assumptions.
 If we assume that $\textsc{3SUM}$, the problem of deciding whether three numbers in a list sum to zero, cannot be solved faster than $O(n^2)$, there are several lower bounds on the complexity of generating all triangles of a graph~\cite{kopelowitz2016higher}. Some lower bounds involving the number of triangles can be rephrased as stating that the delay is at least $O(m^{1/3})$, with $m$ the number of edges.
 Another example is the characterization of the complexity of an acyclic conjunctive query: either it is free-connex and its assignments can be enumerated in constant delay and linear preprocessing or it is not and its assignments cannot be enumerated with constant delay and linear preprocessing, unless the product of boolean matrices can be computed in $O(n^2)$~\cite{bagan2007acyclic}.
The Strong Exponential Time Hypothesis ($\SETH$) is a typical hypothesis in fine graind complexity, wich states that there is no algorithm solving $\textsf{SAT}$ in time $2^{\alpha n}$ with $\alpha<1$.
An algorithm that generates maximal solutions of any strongly accessible set system has a worst case time complexity of $\Omega(t2^{q/2})$, unless $\SETH$ is false~\cite{conte2019listing}, where $t$ denotes the number of solutions and $q$ denotes their maximal size. 


\begin{openproblem}
Can we prove the same lower bound for maximal solutions of strongly accessible set systems, but with $\ETH$ as an hypothesis, using the amplification method of Theorem~\ref{th:hierarchy}?
\end{openproblem}

\begin{openproblem}
Assuming $\SETH$, prove a weak lower bound: there is no \emph{constant delay} or 
\emph{strong polynomial delay} algorithm for generating the minimal transversals of a hypergraph, the circuits of a binary matroid or the maximal cliques of a graph.
\end{openproblem}

\section{Hardness for Specific Algorithms or Computation model}

\paragraph{Extension problem}

The most common method to obtain a polynomial delay algorithm is to prove that 
$\ext{A}$ is polynomial, which implies that flashlight search has polynomial delay, as explained in Chapter~\ref{chap:time}. Hence, showing that $\ext{A}$ is $\NP$-complete rules out the use of backtrack search. Extension problems have been studied for their own sake and there are many $\NP$-hardness results, see~\cite{DBLP:conf/ciac/CaselFGMS19,DBLP:conf/fct/CaselFGMS19} and the references therein. In particular, the extension problem of minimal transversals is $\NP$-hard~\cite{boros1998dual}. 

The only weakness of this approach, is that for some problems, the input of the extension problem solved 
when going through the tree of partial solutions may not be arbitrary, since we can decide in which order the extension is done. Typically, if a solution is a set of elements, the order in which elements are added (or forbidden) depends on the designer of the flashlight search and can even be different in every branch of the tree. Most of the time, this degree of liberty is not exploited but there are a few examples 
where it is relevant. For instance, to enumerate models of a matched CNF formula (see~\cite{savicky2016generating} for a definition), a specific variable must be fixed at each step of the flashlight search, so that the subformula is still a matched formula~\cite{savicky2016generating}.

\paragraph{Maximal cliques}

Even when we have an algorithm with good theoretical guarantees, it can be interesting to lower bound the complexity of the algorithms used in practice to solve it. For instance, linear programs can be solved in polynomial time~\cite{karmarkar1984new} but the simplex mehod which is often used in practice has been proven to be at least in subexponential time for many pivot choices~\cite{friedmann2011subexponential}.

The problem of generating maximal cliques is an important problem, used to compute subgraph isomorphism.
While there are many polynomial delay algorithms (see the references in~\cite{DBLP:conf/icalp/ConteGMV16}), 
algorithms with no guarantee on their delay such as Bron and Kerbosch algorithm~\cite{bron1973algorithm} and CLIQUES~\cite{tomita2006worst} are often used to find maximal cliques and are considered efficient in practice.
In a recent paper~\cite{conte2021overall}, Conte and Tomita revisit these two algorithms to prove that they do not have polynomial delay. The Bron and Kerbosch algorithm (and one of its variants) have a delay at least $O(3^{n/3})$,
using the same reasonning as for the total time. For the CLIQUE algorithm, they are able to prove that if 
it has polynomial delay then $\P = \NP$, by using it to solve the extension problem for maximal clique, which is $\NP$-hard. It is similar to the notion of $\NP$-mighty algorithm~\cite{disser2018simplex}, i.e. some algorithms such as the simplex, can solve an $\NP$-hard problem inside their computation even if they are used to solve a polynomial time problem.

\paragraph{Equivalence to a simpler enumeration problem}

A method to generate maximal elements of independent set systems (closed by inclusion), in particular the maximal cliques of a graph, has been proposed by Lawler et al.~\cite{LawlerLK80}. The method relies on solving the $I \cup \{j\}$ problem: given $I$ an independent set maximal among the independent sets included in $\{1,\dots,j-1\}$, compute all independent sets maximal
within $I \cup \{j\}$.  If this problem can be solved in output polynomial time, then an output polynomial 
algorithm exists for the general problem. Hence, the complexity of this method can be characterized through
the complexity of a simpler (but similar) problem, which could help to prove that it cannot be applied.
Note that this method yields a polynomial delay algorithm to generate maximal cliques, while the extension problem is $\NP$-hard~\cite{LawlerLK80}, meaning that flashlight search cannot solve it in polynomial delay.

 The  enumeration of maximal induced subgraphs satisfying some hereditary property $\mathcal{P}$ is a generalization of enumeration of maximal sets in independent set systems. In~\cite{cohen2008generating}, the \emph{input restricted problem} is introduced: list maximal induced subgraphs, with the additional constraint that there is a vertex in the input graph $G$ such that $G \setminus v \in \mathcal{P}$. The general problem is in $\DelayP$ if and only if the input restricted version is in polynomial delay (similar criteria can be given for $\IncP$ and $\OutputP$). This has been further generalized to $t$-restricted problems, see~\cite{cao2020enumerating}. While these characterizations could help prove lower bounds, they seem more suited to help design efficient enumeration algorithms (many of which are given in ~\cite{cao2020enumerating}). 

Strongly accessible set systems further generalize the previous frameworks, and their maximal elements can be enumerated using polynomial space, see~\cite{conte2019listing} and its introduction for a nice description of the relationship between most algorithmic methods mentioned in this section.

\paragraph{Uniform generator}

To turn a uniform generator into an enumeration algorithm using Theorem~\ref{thm:gentopolydelay}, we need to store all enumerated solutions, hence space may be exponential. Indeed, it seems necessary to encode the set of already generated solutions and these sets are in doubly exponential number and thus cannot be encoded in polynomial space. Therefore the enumeration algorithm needs time to rule out a large number of possible sets of generated solutions. 

This idea has been made precise by Goldberg (Theorem $3$, p.$33$~\cite{Goldberg91}): the product of the delay and the space of an enumeration algorithm obtained from a uniform generator is lower bounded by the number of solutions to output up to a polynomial factor. On the other hand it is easy to build an enumeration algorithm with such space and delay, by generating solutions 
by blocks in lexicographic order (Theorem $5$, p.$42$~\cite{Goldberg91}). The proof of the lower bound relies on the fact that the enumeration algorithm can only output solutions which are given by calls to the generator: we are in a restricted black box model. A set of possible initial sequences of output elements in the enumeration is built so that its cardinality is bounded by an exponential in the space used and that the enumeration produces one of these sequences with high probability. Then, if the delay is too small, with high probability the calls to the generator have not produced any of those special sequences which proves the theorem.

However, if we allow \emph{unbounded repetitions} of solutions, thus relaxing the constraint on the enumeration algorithm, we can devise an incremental linear algorithm with polynomial space,
see Theorem~\ref{th:gentoinclin}.

\paragraph{Jumping automaton on graphs}

Several enumeration algorithms for solving FO queries over sparse models (see ~\cite{vigny2018query}) 
can be executed by a jumping automaton on graphs or JAG~\cite{cook1980space}. In this model, the input, usually a graph $G=(V,E)$, 
is represented by its incidence graph $I$, whose vertices are $V\cup E$ and with an edge $(v,e)$ if the vertex $v$ is incident to the edge $e$ in $G$.
The graph $I$ is equipped with a unary function, which maps a vertex in $V$ to the next vertex in $V$ to represent the list of all vertices. 
For each vertex, we have access to the list of its incident edges by a function giving the first edge and one giving the next one.
Finally, a constant number of unary relations and functions can be computed from $G$ and added to the structure $I$. A unary relation should be interpreted as a color of elements in $V\cup E$ and a unary function as a colored edge on the vertex set $V\cup E$.

The JAG is an automaton, it has a fixed number of states of pebbles and a list of transitions. A transition of a JAG depends on the elements on which its pebbles are placed and its state.
In particular, it can test whether two pebbles are on the same element or if they are of some given color. The transition moves the pebbles along edges of the structure $I$.  A JAG solving an enumeration problem has output states and a given subset of its pebbles represents the output solution when an output state is reached in the computation. A JAG with a fixed number of pebbles models an almost constant memory. To define constant delay enumeration, we require that the delay between two outputs of a JAG is constant.
We hope to prove lower bounds for constant delay enumeration using JAGs, in fact they have been introduced to prove space lower bounds for problems like connectivity~\cite{wigderson1992complexity}.

\begin{proposition}
There is no JAG, which receives as input a structure $I$ representing a graph $G$ and $u$, $v$ two vertices of $G$,
and decides whether $(u,v)$ is an edge in $G$ in constant time.
\end{proposition}
\begin{proof}
Assume that there is a JAG $J$, such that from any graph $G$ a structure $I$ can be built, on which $J$ decides in constant time
whether two elements are an edge in $G$. It means that the computation begins with two pebbles on $u$ and $v$ and $J$ accepts after a
constant number of transitions if and only if $(u,v) \in E$. Assume that all elements of $I$ are labeled by a distinct label of size $O(\log(n))$, where $n = |V \cup E|$. In constant time, a JAG explores only a constant number of elements around $u$ and $v$, hence its decision to accept 
$(u,v)$ as an edge depends on an information which can be encoded into $C\log(n)$ bits, for some constant $C$. The labeling of the vertices is necessary to differentiate vertices, since a JAG
can test whether two pebbles are on the same vertex.
We have seen that the set of edges of $G$ can be represented by the neighborhood in $I$ of each vertex $v \in V$, hence it can be encoded by $Cn\log(n)$ bits.
As a consequence, there are at most $2^{Cn\log(n)}$ different graphs represented by $J$ and any incidence structure $I$. Since there are at least $2^{n^2}/n!$ unlabeled graphs,
which is asymptotically larger than $2^{Cn\log(n)}$, then there is no $J$, which decides the edge relation in constant time.
\end{proof}

We want to prove a lower bound for an enumeration problem, using an approach similar to the one used for the edge relation. 
However, the enumeration of edges by a JAG is trivial, by enriching $I$ with a function representing the list of edges.
Hence, we need to consider a more complex FO query, such as one defining paths of length $2$. 

 \begin{openproblem}[Suggested by Segoufin, personal communication]\label{open:jag}
  Unconditionally prove that there is some simple FO query, say $\exists z(E(x,z) \wedge E(z,y))$, which cannot be solved using a JAG with a constant number of pebbles and constant delay. 
  The preprocessing to compute the structure $I$ representing the graph may be further restricted.
 \end{openproblem}

 Unconditional lower bounds are extremely hard to obtain: after fifty years of research, it is not even known whether $\mathsf{SAT}$ requires superlinear time to be solved. To obtain unconditional lower bounds in enumeration, the easiest approach is to show that a problem is not in some small complexity class, like constant delay with linear preprocessing, and even to restrict the computational model.  That is why we have proposed Open Problem~\ref{open:jag} on the JAG computational model.
Indeed, unconditional separations from classical complexity involve classes of problems solvable by a family of boolean circuits: strict hierarchy inside $\AC^0$~\cite{hastad1986almost}, polynomial size monotone circuits cannot decide the existence of a $k$-clique~\cite{razborov1985lower} and $ \ACC \neq \NEXP $~\cite{williams2014nonuniform}. 

\begin{openproblem}
 Can we define meaningful circuit classes for enumeration and prove unconditional separations of $\DelayP$ from one of these classes?
 Can we relate circuit classes for enumeration to parallel computation or descriptive complexity?
\end{openproblem}

\section{Restricted Classes of Problems}

We now present several families of enumeration problems, which are easier to classify
but nonetheless contain  many interesting problems. Saturation under set operators, polynomial interpolation and queries with second order variables are presented in more details because I have worked on them.

\subsection{Constraint Satisfaction Problems}

Constraint satisfaction problems are a broad generalization of the $3$-\textsc{SAT} problem: an instance of a constraint satisfaction problem (CSP) is given as a conjunction of constraints over variables.  
The constraints are relations over a domain, usually finite and come from a set of constraints called the constraint language, e.g. the disjunctions over three literals for $3$-\textsc{SAT}. 
A solution of a CSP is an assignment of the variables which satisfies all constraints. A CSP can also be interpreted as finding a homomorphism between two structures representing respectively the variables and the domain.

The fundamental question is to classify the constraint languages for which the problem of finding 
a solution is easy or hard. It turns out that in the boolean case, all constraint languages either yield a CSP in $\P$ or in $\NP$~\cite{schaefer1978complexity} and this result has been recently extended to any finite domain~\cite{bulatov2017dichotomy,zhuk2020proof}, which is surprising since in general, there are problems of intermediate complexity between $\P$ and $\NP$ by Ladner's theorem~\cite{ladner1975structure}.

A CSP problem can easily be turned into an \emph{enumeration CSP} by asking to generate all of its satisfying assignments rather than deciding whether there is one.
Over a boolean domain, there is a dichotomy: for Horn relations, anti-Horn relations, affine relations and bijunctive relations, the enumeration is in $\DelayP$, while all other problems are $\IncP$-hard (and even $\EnumP$-complete)~\cite{creignou1997generating}. This result has been generalized to domains of larger size~\cite{schnoor2007enumerating}, with a new efficient algorithm and some hardness results, but no complete classification. In particular, it is shown that some enumeration CSPs in $\DelayP$ cannot be solved by a flashlight search, contrarily to the boolean case.

\begin{openproblem}
Is it possible to prove a dichotomy for enumeration CSPs and finite domain, using results of ~\cite{bulatov2017dichotomy,zhuk2020proof}?
\end{openproblem}

The enumeration CSPs have also been considered with an order constraint on the output of solutions, as in Chapter~\ref{chap:order}.
When the solutions must be given in increasing or decreasing Hamming weight order, there is also a dichotomy proved in~\cite{DBLP:conf/sat/CreignouOS11,DBLP:conf/lata/Schmidt17}.

Variations on enumeration CSP, where the constraint language is unrestricted but the instances are restricted, are studied in~\cite{bulatov2012enumerating}.
Typically, the problems are in $\DelayP$ when the input has a small tree-width and it seems to be different from the same results for decision problems, see~\cite{greco2013structural} for an overview of results on variants of enumeration CSPs.

\subsection{Saturation by Set Operators}

As explained in Chapter~\ref{chap:time}, all algorithms in incremental polynomial time can be interpreted as 
saturation algorithms. In particular, saturation problems, which are defined as computing the closure of a set of elements by a $k$-ary operation computable in polynomial time, are in $\IncP_{k+1}$.
A natural question is to understand for which operation, the defined saturation problem is in fact in $\IncP_1 = \DelayP$.
We tried to settle this question with Arnaud Mary in~\cite{mary2016efficient,wepa2016,DBLP:journals/dmtcs/MaryS19}, by restricting the set of saturation operators.

 In the proposed framework, a solution is a vector over some finite set, the saturation operations are of fixed arity and they act \emph{component-wise} and in the same way on each component of the vector. We fix a finite domain $D$. Given a $t$-ary operation $f$ (a function from $D^t$ to $D$), $f$ can be naturally extended to a $t$-ary operation over vectors of the same size. Let $(v^{1},\dots v^{t})$ be a $t$-uple of vectors of size $n$, $f$ acts component-wise on it, that is for all $i\leq n$, $f(v^{1},\dots, v^{t})_i = f(v^{1}_i,\dots, v^{t}_i)$. For instance, if $f$ is the conjunction on the boolean domain ($D=\{0,1\}$), then $f$ acts on two vectors, which can be interpreted as subsets of $[n]$, and computes their intersection.

\begin{definition}
Let $\cF$ be a finite set of operations over $D$.
Let $\ccS$ be a set of vectors over $D$.
Let $\cF^i(\ccS) = \{f(v_1,\dots, v_t) \mid v_1,\dots, v_t \in \cF^{i-1}(S) \text{ and } f\in \cF \} \cup \cF^{i-1}(\ccS)$ and $\cF^0(\ccS) = \ccS$.
The closure of $\ccS$ by $\cF$ is $\Cl_\cF(\ccS) = \displaystyle{\cup_i \cF^i(\ccS)}$.
\end{definition}

Notice that $\Cl_{\cF}(\ccS)$ is the smallest set which contains $\ccS$ and which is closed by the operations of $\cF$.
As an illustration, assume that $D = \{0,1\}$ and that  $\cF = \{ \vee \}$.
Then the elements of $\ccS$ can be seen as subsets of $[n]$ (each vector of size $n$ is the characteristic vector of a subset of $[n]$) and $\Cl_{\{\vee\}}(\ccS)$ is the closure by union of all
sets in $S$. Let $\ccS = \{ \{1,2,4\},  \{2,3\},  \{1,3\} \}$ then $$\Cl_{\{\vee\}}(\ccS)= \{ \{1,2,4\},  \{1,2,3,4\}, \{2,3\},  \{1,3\},  \{1,2,3\}\}.$$ 

For each set of operations $\cF$ over $D$, we define the problem $\EnumClo_\cF$: given a set of vectors $\ccS$, list the elements of $\Cl_\cF(\ccS)$.
Note that $\EnumClo_{\{\vee\}}$ is the problem $\enum{\textsc{Union}}$ introduced in Chapter~\ref{chap:time}.
A naive saturation algorithm solves $\EnumClo_\cF$: the elements of $\cF$ are applied to
all tuples of $\ccS$ and any new element is added to $\ccS$, the algorithm stops when no new element can be produced.
If the largest arity of an operation in $\cF$ is $k$ then this algorithm adds a new element to a set of $t$ elements (or stops)
in time $O(t^kP(n))$ where $n$ is the size of a solution and $P(n)$ the time to evaluate an operation of $\cF$. Hence, $\EnumClo_\cF \in \IncP_{k+1}$.

One way of solving $\EnumClo_\cF$ in polynomial delay is to solve the associated extension problem in polynomial time.
The extension problem is equivalent to the membership problem, denoted by $\Mem_\cF$:
given $\ccS$ a set of vectors of size $n$, and a vector $v$ of size $n$ does $v$ belong to $\Cl_\cF(\ccS)$?

\paragraph{The Post lattice}

Let $\cF$ be a finite set of operations over $D$. The \emph{functional clone} of $\cF$,  denoted by $<\cF>$, is the set of operations obtained by any composition of the operations of $\cF$ and projections.
The sets $\Cl_\cF(\ccS)$ and $\Cl_{<\cF>}(\ccS)$ are the same, hence it is enough to study functional clones to understand the complexity of $\EnumClo_{\cF}$.
When the domain is boolean, i.e. $D = \{0,1\}$, the functional clones form a lattice called Post's lattice (see~\cite{post1941two,reith2003optimal,bohler2005bases}). We prove by reduction that, for  many functional clones of this lattice, the problems $\EnumClo_\cF$ are equivalent. We are left with $4$ families of functional clones:

\begin{itemize}
	\item The clone $E_2 = <\wedge>$ or its dual $<\vee>$. The problem $\Mem_{E_2}$ can be solved in linear time, and we can design 
	a flashlight search in linear delay in the input, see Chapter~\ref{chap:time}.
	\item The clone $L_0 = <+>$ where $+$ is the boolean addition. Again $\Mem_{<+>}$ is in polynomial time, and the vector space it generates can 
	be enumerated in linear delay in the size of a solution. The problem is similar to $\enum{\textsc{LinearSystem}}$ presented in Chapter~\ref{chap:order} and is in $\SDelayP \cap \QueryP$.
	The results are similar for the clone $L_2$ generated by the sum modulo two of three elements.
	\item The clone $BF = <\vee,\neg>$, which defines a boolean algebra and $M_2 = <\wedge,\vee>$. In both cases, we can compute atoms whose combinations by $\vee$ produce all solutions,
	and it can be done with linear delay in a solution. Again the problem is in $\SDelayP \cap \QueryP$.
	\item The clones $S_{10}^k = <Th_k^{k+1}> $, where $Th_k^{k+1}$ is the threshold function over $k+1$ variables, which is equal to one if and only if at least $k$ of its arguments are ones.
	Because of the Baker-Pixley Theorem~\cite{baker1975polynomial}, $\Cl_{S_{10}^k}(\ccS)$ can be characterized by all the projections of $\ccS$ over $k$ components. Hence, $\Mem_{S_{10}} \in \P$
	and the flashlight search is in strong polynomial delay. In the particular case of the majority operator $<Th_2^{3}>$, a reduction to the problem of enumerating models 
	of a $2CNF$ formula gives a linear delay algoritm, using the algorithm of~\cite{feder1994network}. 
\end{itemize}

As a conclusion, we obtain the following theorem.

\begin{theorem}[\cite{mary2016efficient}]
For any finite set of boolean operators $\cF$, $\EnumClo_\cF \in \DelayP$.
\end{theorem}

We would like to prove that the complexity of the problem is in $\SDelayP$ rather than $\DelayP$.
The only open case is the clone $E_2 = <\wedge>$ for which no strong polynomial delay algorithm is known. 
It is equivalent to the question $\enum{\textsc{Union}} \in \SDelayP$, which is mentioned in Chapter~\ref{chap:low}.

We can also consider the complexity of the \emph{uniform problem}, where the clone is given as input instead of being fixed. 
Let \textsc{ClosureTreshold} be the following problem: given a set $\ccS$ of vectors and an integer $k$ decide whether the vector $\mathbf{1}\in \Cl_{S_{10}^k}(\ccS)$,
we have proven that \textsc{ClosureTreshold} is $\coNP$-complete. Hence, we cannot use flashlight search to solve the uniform problem unless $\P = \NP$.

\begin{openproblem}
Prove that the uniform problem is hard for $\IncP$. The proof must rely on infinite families of clones such as $\Cl_{S_{10}^k}$.
\end{openproblem}

\paragraph{Larger domain}

The first generalization of the previous result is to consider domains larger than $2$. 
However, the lattice of functional clones gets immensly more complicated. For $|D| \geq 3$,
we can generalize the polynomial cases of the boolean domain: group operations or sets of operations containing a 
near unanimity term (similar to the threshold function) have a polynomial time membership problem. 
The problem $\Mem_{\cF}$ is known in the universal algebra community under the name \emph{subpower membership problem} or SMP. The method for groups also works for extensions of groups by multilinear operations such as ring or modules. For semigroups, some cases are known to be $\PSPACE$-complete and other are in $\P$~\cite{bulatov2016subpower,steindl2017subpower}.

The last case we would like to extend is the clone generated by the conjunction.
A natural generalization is to fix an order on $D$ and to study the complexity
of $\Mem_{f}$ with $f$ a binary monotone operation. Let $f$ be the function over $D=\{0,1,2\}$ defined by  $f(x,y) = \min(x+y,2)$. This function is monotone in each of its arguments. The problem \textsc{Exact3Cover} of finding an exact cover of a set by subsets of size $3$ reduces to $\Mem_{f}$~\cite{mary2016efficient}, which proves that $\Mem_{f}$ is $\NP$-complete.
However, a direct application of the supergraph method yields the following proposition.

\begin{proposition}[\cite{mary2016efficient}]
 If $f$ is an associative function of arity $2$, then $\EnumClo_{f} \in \DelayP$.
\end{proposition}

\begin{openproblem}
Better characterize the clones for which $\EnumClo_{f}$ can be solved by a supergraph method, in particular for arity larger than $2$.
Is it possible to do a reverse search to avoid exponential space?
\end{openproblem}

\paragraph{Non uniform operations}

The second natural generalization is to allow operations to act differently on each component of the vectors.
First, we only allow unary operations acting on a single component by turning it to zero or to one.
It roughly corresponds to listing downward closure (or upward closure) of $\Cl_{\cF}(\ccS)$ and for 
the boolean domain, we are still able to prove that $\Mem_{\cF} \in \P$. 

If we consider unary operators acting on three componenents, then there is one operator
for which the membership problem is equivalent to \textsc{Exact3Cover}, which leads to the following question.

\begin{openproblem}
 Is the membership problem $\NP$-hard for some unary operator acting on $2$ components?
\end{openproblem}

\paragraph{Maximal solution}

The last generalization is to enumerate minimal or maximal elements of $\Cl_{\cF}(\ccS)$, to reduce the number of solutions output while still 
producing the most interesting ones. There is no natural saturation algorithm to compute the maximal elements of $\Cl_{\cF}(\ccS)$, contrarily to all problems studied in this section.
Hence, there is no guarantee that the problem can be solved by an incremental polynomial time algorithm.

We consider the boolean domain, and the order is the inclusion of solutions. We have proved that the problems of 
generating maximal or minimal elements of $\Cl_{\cF}(\ccS)$ is either trivial, or equivalent to one of the following well-known enumeration problems.

\begin{itemize}
\item The enumeration of the maximal models of a $2-CNF$ formula, which can be solved in linear delay~\cite{kavvadias2000generating}.
\item The enumeration of the circuits of a binary matroid (eventually with a fixed element in the circuits), which is in $\IncP_2$~\cite{khachiyan2005complexity}.
\item The enumeration of maximal independent sets of a hypergraph of dimension $k$ (the size of the largest hyperedge), which is in $\IncP_k$~\cite{eiter1995identifying}.
\end{itemize}

As a conclusion, generating maximal or minimal elements of $\Cl_{\cF}(\ccS)$ over the boolean domain is in $\IncP$ for all $\cF$.

\subsection{Query Problem}

Let $\mathcal{L}$ be some logic, then a query $\phi \in \mathcal{L}$ is a formula of $\mathcal{L}$ with 
free variables. Given a model, solving this query means listing all assignments of elements of the model to the free variables which satisfy the query. 
The decision version of this problem is the model checking problem: decide whether a formula withtout free variables is satisfied over some model.

\paragraph{Small solutions}

Typically, the query is of small size, while the model is huge, hence we consider the data complexity
of the problem (query size is fixed), see Section~\ref{chap:low}. Since the size of a solution to a query is bounded by the number of free variables,
the solutions considered are small. 

The most commonly studied logic is the first order logic (FO).
There are constant delay algorithms when the model on which the query is evaluated is restricted to be of bounded or low degree~\cite{DBLP:journals/tocl/DurandG07,DBLP:journals/corr/abs-2010-08382}, of bounded expansion~\cite{DBLP:journals/lmcs/KazanaS19} or nowhere dense~\cite{DBLP:conf/pods/SchweikardtSV18}. The last result gives in fact a dichotomy on the complexity of the enumeration of FO queries modulo a complexity hypothesis.
\begin{theorem}[\cite{vigny2018query}]
Let $\mathcal{C}$ be a class closed under subgraph and assume $\FPT \neq \AW[*]$, then the problem of enumerating the models of a FO query
is not in $\DelayP$ if and only if $\mathcal{C}$ is somewhere dense.
\end{theorem} 

See~\cite{vigny2018query} for more details on this line of work. 
Note that there are lower bounds for the complexity of solving an FO query with regard to the query size (coming from the model checking problem). For instance, if $\FPT \neq \AW[*]$,
there is no doubly exponential algorithm to solve an FO query over a bounded degree structure~\cite{frick2004complexity}.

\paragraph{Large solutions}

If we consider the combined complexity of solving a query, that is the formula is part of the input and the complexity depends on its size,
then we must consider simpler queries (but with larger solutions).
The most studied queries are conjunctive queries (CQ): conjunction of relational atoms and equality atoms, existentially quantified.
Enumeration algorithms with linear preprocessing and constant delay are known~\cite{bagan2007acyclic,Bagan09,brault2013pertinence} for the
class of self-join \emph{free-connex} acyclic CQs. Modulo a strong complexity assumption, this class characterizes exactly the self-join free acyclic CQs with a constant delay algorithm. For a nice exposition of the algorithm for conjunctive queries and of the lower bound see~\cite{DBLP:journals/siglog/BerkholzGS20}

Large solutions also appear when the query is of fixed size but the free variables are second order.
Hence, we list sets of elements of the model, which are of the same size as the domain of the model.
Generalizing Courcelle's theorem, solving a monadic second order query over a tree, a graph or a matroid of bounded branch width can be done 
in $\SDelayP$~\cite{DBLP:conf/csl/Bagan06,DBLP:journals/dam/Courcelle09,strozecki2011monadic}. This result cannot be generalized further, since modulo $\ETH$ and some technical conditions, the tree-width of a class of structures must be polylogarithmic to obtain an \FPT model checking algorithm~\cite{kreutzer2009parameterised}.

With Arnaud Durand~\cite{durand2011enumeration}, we have studied  first order queries with free second order variables in an attempt to 
find classes of tractable queries with large solutions on general models. Note that the second order variables are not monadic, they can be of any fixed arity. 
Since in full generality such formulas may be very expressive, we consider fragments defined by the quantifier alternation of 
formulas inspired by results on counting problems~\cite{saluja1995descriptive}. 

\begin{example}
The formula $IS(T) = \forall x \forall y \ T(x) \wedge T(y) \Rightarrow \neg E(x,y)$ holds if and only 
if $T$ is an independent set.

A hypergraph $H$ is represented by the incidence structure $\langle D,\{V,E,R\}\rangle $ where $V(x)$ means that $x$ is a vertex,
$E(y)$ that $y$ is an hyperedge and $R(x,y)$ that $x$ is a vertex of the hyperedge $y$.
The formula $HS(T) = \forall x (T(x) \Rightarrow V(x)) \wedge \forall y \exists x  E(y) \Rightarrow (T(x) \wedge R(x,y))$ 
holds if and only if $T$ is a hitting set of $H$.
\end{example}

For any quantifier free formula, there is an algorithm with polynomial time preprocessing and constant delay. 
The proof uses a reduction to a constant size \textsc{CNF} formula, whose models can be enumerated in constant time and are then extended using a Gray code.
Note that we cannot obtain an algorithm with a polynomial preprocessing independent from the formula if $\FPT \neq \W[1]$, since we can express the class of $k$-cliques of a graph with a quantifier free formula.

For formulas with only existential quantifiers, the solutions are a non disjoint union of solutions of quantifier free formulas. 
Hence, we obtain a polynomial delay algorithm using the method of Theorem~\ref{th:stability}. Moreover, the problem is equivalent to the enumeration of models of a $k$-\textsc{DNF} formula, a problem we have recently shown to be in constant delay~\cite{capelli2020enumerating}. Hence, the complexity of solving quantifer free queries and existential queries is the same.

There is a formula with only universal quantifiers such that the enumeration problem defined by this formula is equivalent to $\enum{SAT}$, which is $\EnumP$-complete. Hence, we have obtained a complete characterization of the enumeration complexity of first order formulas with regard to their quantifier alternation.
There are interesting tractable fragments of first order formulas above existential formulas. They correspond to formulas which can be reduced to the enumeration of 
models of a $2$-\textsc{CNF} formula or a Horn formula. They contain the problem of generating (maximal) independent sets or dominating sets of a graph.

\subsection{Polynomial Interpolation}

In this section, we consider the interpolation of a multivariate polynomial, a classical problem from computer algebra, as an enumeration problem. The input polynomial $P$ is given as a black box: the number of its variables is $n$, and we can query the black box on $n$ values $v_1,\dots, v_n$,  which returns $P(v_1,\dots,v_n)$. The polynomial may be defined over $\mathbb{Z}$, $\mathbb{Q}$ or a finite field.
Interpolating $P$ means computing explicitly $P$ as a set of monomials, by using queries to compute the value of the polynomial on several points. When analyzing interpolation algorithms from the enumeration point of view, we are interested in their delay and incremental time. 

While any set of objects can be mapped to the monomials of a polynomial, to obtain good enumeration algorithms we are interested in polynomials with low degree and which can be evaluated in polynomial time.
Many interesting problems can be embedded in the monomials of easy to evaluate polynomials:

\begin{itemize}
	\item the cycle-covers of a graph are the monomials of the determinant of its adjacency matrix,  
	\item the spanning trees of a graph are the monomials of the determinant of its Kirchoff matrix by the Matrix-Tree theorem (see~\cite{jerrum2003counting})
	\item the spanning hypertrees of a $3$-uniform hypergraph are the monomials of the Pfaffian of its Laplacian matrix, as proved in~\cite{masbaum2002new}
	\item the probability of acceptation of $n$ letter words by a probabilistic automaton can be encoded into the monomials of a multilinear polynomial obtained as a product of matrices~\cite{kiefer2011language,peyronnet2012approximate}
\end{itemize}

To obtain a complexity depending on the number of monomials, interpolation algorithms for multivariate polynomials need to be randomized. When the complexity depends on the number of monomials, the interpolation is for \emph{sparse polynomials} by opposition to interpolation algorithms for \emph{dense polynomials}, with a complexity depending on the degree and the number of variables. This corresponds to the notions of input sensitive and output sensitive algorithms in enumeration. The first algorithms designed for sparse polynomials are in output polynomial time with a small probability to fail~\cite{ben1988deterministic,zippel1990interpolating}, while we describe here more recent incremental polynomial time or polynomial delay algorithms~\cite{klivans2001randomness,strozecki2010,strozecki2013enumerating}.  

The decision problem associated to interpolation is called \emph{polynomial identity testing} ($\PIT$): given a black box polynomial decide whether it is the zero polynomial.
To solve $\PIT$, we use a probabilistic algorithm, which has a good complexity and an error bound which can be made exponentially small in polynomial time. 

\begin{lemma}[Schwarz-Zippel \cite{zippel1979probabilistic}]
\label{proba}
Let $P$ be a non zero polynomial with $n$ variables of total degree $D$, let $\epsilon >0$, if  $x_{1}, \dots, x_{n}$ are randomly chosen values in a set of integers $S$ of size $\frac{D}{\epsilon}$ then the probability that $P(x_{1}, \dots, x_{n}) = 0$ is bounded by $\epsilon$.
\end{lemma}

The algorithms presented in the following are based on solving $\PIT$ in randomized polynomial time, transformations of the considered polynomials by substitution of variables and interpolation of univariate polynomials in polynomial time.

\paragraph{Polynomial Delay Algorithm}

Let $e$ be a vector of dimension $n$ with positive integer coefficients. We denote by $\textbf{X}^{e}$
the term $\prod_{i=1}^{n} X_i^{e_i}$. A multivariate polynomial is the sum of monomials, that is disctinct terms multiplied by scalar coefficients.
The degree of a monomial $\lambda \textbf{X}^{e}$ is $\max_{1 \leq i \leq n} e_i$, the maximum of the degrees of its variables. The total degree of a monomial
is $\sum_{i=1}^{n} e_i$, the sum of the degrees of its variables. Let $d$ (respectively $D$) denote the degree (respectively the total degree) of the polynomial we consider,
that is the maximum of the degrees (respectively total degrees) of its monomials.
We assume for the first algorithm that the polynomial is multilinear, i.e. $d =1$ which implies that $D$ is bounded by $n$.

We consider the decision problem \textsc{Monomial-Factor}: an input is a term $\textbf{X}^{e}$,
that is a product of variables, and a polynomial $P$ given as a black box; the input is accepted if $\textbf{X}^{e}$ divides a monomial of $P$.  
This problem corresponds to the extension problem defined in Chapter~\ref{chap:time}.
Problem \textsc{Monomial-Factor} can be solved in randomized polynomial time for multilinear polynomials (see~\cite{strozecki2013enumerating}):
Substitute a single variable $Y$ to all variables in $\textbf{X}^{e}$, then choose a random assignment
of the other variables. The problem has a positive answer (with high probability) if and only if the degree of the obtained univariate polynomial in $Y$ is equal to the total degree of $\textbf{X}^{e}$.  

From the resolution of \textsc{Monomial-Factor}, we obtain a randomized polynomial delay interpolation using a flashlight search, with a few technicalities due to the randomization of the algorithm and the problem of computing the coefficient of each monomial.

\begin{theorem}[\cite{strozecki2013enumerating}]
Let $P$ be a multilinear black box polynomial with $n$ variables, and total degree $D$ and let $\epsilon >0$.
There is an algorithm which computes the set of monomials of $P$ with probability $1- \epsilon$, with a delay in time $O(nD^{2}(n + \log(\epsilon^{-1})))$ and $O(nD(n + \log(\epsilon^{-1})))$ black box queries on points of size  $O(\log(D))$.
\end{theorem}

The support of a monomial is the set of variables which have degree one or more 
in the monomial. If the set of supports of monomials of a polynomial is such that,
no support is included into another, then \textsc{Mon-Factor} can be solved in polynomial time 
by substituting zero to all variables not in $\textbf{X}^{e}$ and solving $\PIT$ on the obtained polynomial. We can then use the same flashlight search as for multilinear polynomial.  

\begin{openproblem}
Is it possible to solve \textsc{Mon-Factor} for other classes of polynomials, for instance 
polynomials given by read twice arithmetic formulas or depth three arithmetic circuits? Is it possible to design polynomial delay algorithms based on reverse search rather than on flashlight search?
\end{openproblem}

\paragraph{Polynomial Incremental Time Algorithm}

In the context of polynomial interpolation, the problem of finding another solution 
is equivalent to the problem of finding a single monomial. Indeed, given a polynomial $P$ by a black box and an explicit subset $S$ of its monomials, it is easy to simulate queries to $\displaystyle{P - \sum_{M \in S} M}$. 
Then, it is easy to grow $S$, until $\displaystyle{P - \sum_{M \in S} M} = 0$, which can be decided by solving $\PIT$.

We can relax the condition of supports of monomials not included into another one presented in the previous section. Let us consider polynomials such that all the supports of the monomials are different, we say that
these polynomials have \emph{distinct supports}. In this case, a simple greedy algorithm produces a monomial in polynomial time: from the set of all variables, remove a variable while the polynomial restricted to this set of variables is non zero.  

For polynomials with disjoint supports, all monomials can be easily isolated, hence finding a monomial is easy.
For general polynomials, there is a clever transformation of a multivariate polynomial
into a univariate polynomial, such that the monomial of lowest degree in the univariate polynomial comes from
a single monomial of the multivariate polynomial (see~\cite{klivans2001randomness}). The proof relies on a generalized isolation lemma, where all variables are mapped to a different random power of a single variable.

\begin{theorem}[Adapted from Theorem $12$ of~\cite{klivans2001randomness}]
  There is a randomized polynomial time algorithm which given a black box access to a non zero polynomial $P$ with $n$ variables and total degree $D$ and $\epsilon > 0$ returns a monomial of $P$ with probability $1-\epsilon$ in a time polynomial in $n$, $D$ and with $O(n^{6}D^{4})$ calls to the black box on integers of size $\log(\frac{nD}{\epsilon})$.
\end{theorem} 

There is another method for the case of polynomials of small degree. Using a generalization of the algorithm
solving \textsc{Mon-Factor} over multilinear polynomials, we first find a maximal power of a linear factor of a monomial. Then, we would like to do the division by this factor to recursively find a monomial in the quotient. It turns out that this quotient can be evaluated on values using substitution and univariate interpolation, which is enough for our algorithm. We obtain an algorihm to find a single monomial, which is better than the general algorithm when $d<10$.

\begin{theorem}[\cite{strozecki2013enumerating}]
Let $P$ be a polynomial with $n$ variables of degree $d$ and of total degree $D$ and let $\epsilon > 0$.
 There is an algorithm which computes a monomial of $P$, with probability $1-\epsilon$.
 It performs $O(nD^{d})$ calls to the oracle on points of size $\log(\frac{n^{2}D}{\epsilon})$.
\end{theorem}

\begin{openproblem}
All incremental polynomial time algorithms rely on the evaluation of all the monomials found to produce a new one.
Is there a class of polynomials for which this set can be stored concisely and evaluated efficiently? Is there an alternative algorithm using 
only polynomial space? 
\end{openproblem}

\paragraph{Lower bounds}

For any $t$ evaluation points, one may build a polynomial with $t$ monomials which vanishes at every point (see \cite{zippel1990interpolating}). 
Therefore, if we do not have any a priori bound on $t$, we must evaluate the polynomial on at least $(d+1)^{n}$ $n$-tuples of integers to decide $\PIT$ with a deterministic algorithm. Hence, it is not possible to obtain a deterministic interpolation algorithm in output polynomial time for a polynomial given as a black box. 
Even for an explicit representation of a polynomial, such as arithmetic circuits, finding an algorithm in $\P$ for $\PIT$
would be a major achievement with many complexity implications, see e.g.~\cite{kabanets2004derandomizing}. 

For multilinear polynomials, it is possible to prove that $\PIT$ is equivalent to \textsc{Mon-Factor}~\cite{strozecki2013enumerating}, so the fact which yields a randomized polynomial delay algorithm is also the obstacle to an efficient \emph{deterministic enumeration algorithm}.

To rule out the possibility of doing a flashlight search on a class of polynomials, it is enough to prove
a hardness result for \textsc{Mon-Factor}. By encoding a $2$-\textsc{CNF} formula into an arithmetic circuit, which is then homogenized, we prove that \textsc{Mon-Factor} can simulate the problem of finding an assignment of a $2$-\textsc{CNF} formula of minimal Hamming weight. This proves the $\NP$-hardness of \textsc{Mon-Factor} over a very restricted class of polynomials~\cite{strozecki2013enumerating}: degree two polynomials represented by a depth four arithmetic circuit.

\chapter{Solving a Practical Enumeration Problem}
\label{chap:practical}
\minitoc

Most concrete enumeration problems are too complex to be solved by an algorithm with a theoretical guarantee on its delay.
In that case, input sensitive algorithms may be efficient in practice, more than a polynomial delay algorithm, see the example of maximal cliques~\cite{cazals2008note}. Instead of trying to obtain the best asymptotic worst case delay and fail, it becomes worthwile to optimize for the typical size and typical input, while relaxing the constraints on exhaustivity and redundancy of solutions.
In this chapter we summarize the results we obtained in~\cite{barth2015efficient}:  we present an enumeration problem from cheminformatic, and an algorithm designed to solve it on small instances rather than obtaining a good complexity. We show that methods from enumeration can be combined to build an algorithm which performs well in practice, and we try to underline how the design of a practical algorithm differs from the design of an algorithm with the best possible delay.

\section{Model}

We study the problem of \textbf{generating all colored planar maps up to isomorphism} representing possible
 molecules obtained from a set of elementary starting motifs. A \emph{map} is a graph equipped with, for each vertex, a cyclic order on its neighbors.
 We consider maps, to retain some geometric properties of the molecule we model. We want to design molecular cages which can be embbeded on a sphere and are thus planar.
When the map is planar, the orders of each neighborood are given by a planar embedding of the graph. The two main difficulties of the problem are the sheer number of solutions and the generation of duplicates because we are interested by maps up to isomorphism.

The basic chemical elements are modeled by maps called \emph{motifs} (see Figure~\ref{fig::motif}), with a single \emph{center} vertex which represents the element
and several neighbors which represent its chemical connections, which are labeled by an element of $\alphabet$. The colors are partitionned into
positive and negative colors: each positive color $a$ in $\alphabet$ has a unique complementary negative color denoted by $\overline{a}$. Each color represents a different kind of reacting center,
which can only be paired with a reacting center of complementary color.
Then the motifs are assembled to form a \emph{map of motifs}, as shown in Figure~\ref{fig::mapofmotifs}, by adding an edge between two vertices from different motifs with complementary colors. 
A vertex which is connected only to its center vertex is said to be \emph{free}. We are interested in \emph{saturated} maps of motifs, that is maps with no free vertex, because they may represent real molecules. 
Note that, in this model, all geometric informations on the molecule represented by a motif are lost, but the order of neighbors of a vertex. Hence a map of motifs does not always represent an existing molecule.

\begin{figure}
\centering  
  \begin{tikzpicture}[scale=0.74]
    \node[centre]   (C1) at (-1, 0) {$\mot{Y}$};
    \node[label] (A1) at ( 0,1) {$\mathbf{\overline{a}}$};
    \node[label] (A2) at ( -1, -1) {$\mathbf{\overline{a}}$};
    \node[label] (A3) at ( -2, 1) {$\mathbf{\overline{a}}$};
    \draw (C1) -- (A1);
    \draw (C1) -- (A2);
    \draw (C1) -- (A3);
    \draw[->] (-0.3,0.5) to[bend left=45] (-0.7,-0.5);
     \node at (0.2,0) {\Next};
    \node[centre]   (C2) at (2,0) {$\mot{I}$};
    \node[label] (B1) at (2,1) {$\mathbf{a}$};
    \node[label] (B2) at (2,-1) {$\mathbf{a}$};
    \draw (C2) -- (B1);
    \draw (C2) -- (B2);
    \node[centre]   (C3) at (5,0) {$\mot{X}$};
    \node[label] (D1) at (4,1) {$\mathbf{a}$};
    \node[label] (D2) at (6,1) {$\mathbf{a}$};
    \node[label] (D3) at (4,-1) {$\mathbf{a}$};
    \node[label] (D4) at (6,-1) {$\mathbf{a}$};
    \draw (C3) -- (D1);
    \draw (C3) -- (D2);
    \draw (C3) -- (D3);
    \draw (C3) -- (D4);
   \node[centre]   (C4) at (9, 0) {$\mot{V}$};
    \node[label] (E1) at ( 8,1) {$\mathbf{b}$};
    \node[label] (E2) at ( 10, 1) {$\mathbf{b}$};
    \node[label] (E3) at ( 9, -1) {$\mathbf{a}$};
    \draw (C4) -- (E1);
    \draw (C4) -- (E2);
    \draw (C4) -- (E3);

    \node[centre]   (C5) at (12,0) {$\mot{J}$};
    \node[label] (F1) at (12,1) {$\mathbf{a}$};
    \node[label] (F2) at (12,-1) {$\mathbf{\overline{b}}$};
    \draw (C5) -- (F1);
    \draw (C5) -- (F2);
  \end{tikzpicture}
\caption{Example of motifs on 
  $\alphabet_M=\{\mot{Y},\mot{I},\mot{X},\mot{V},\mot{J}\}$ and
  $\alphabet=\{a,\overline{a},b,\overline{b}\}$.}
\label{fig::motif}
\end{figure}

\begin{figure}
  \begin{tikzpicture}[scale=0.74]
    \node[centre]   (C1) at (-1, 0) {$\mot{Y}$};
    \node[label] (A1) at ( 0,-1) {$\mathbf{\overline{a}}$};
    \node[label] (A2) at ( 0, 0) {$\mathbf{\overline{a}}$};
    \node[labelsat] (A3) at ( 0, 1) {$\mathbf{\overline{a}}$};
    \draw (C1) -- (A1);
    \draw (C1) -- (A2);
    \draw (C1) -- (A3);

    \node[centre]   (C2) at (2,1) {$\mot{I}$};
    \node[labelsat] (B1) at (1,1) {$\mathbf{a}$};
    \node[labelsat] (B2) at (3,1) {$\mathbf{a}$};
    \draw (C2) -- (B1);
    \draw (C2) -- (B2);



    \node[centre] (C5) at (5, 0) {$\mot{Y}$};
    \node[label] (F1) at (4,-1) {$\mathbf{\overline{a}}$};
    \node[label] (F2) at (4, 0) {$\mathbf{\overline{a}}$};
    \node[labelsat] (F3) at (4, 1) {$\mathbf{\overline{a}}$};
    \draw (C5) -- (F1);
    \draw (C5) -- (F2);
    \draw (C5) -- (F3);

        \draw (A3) -- (B1);
        \draw (F3) -- (B2);
  \end{tikzpicture}\hfill
  \begin{tikzpicture}[scale=0.74]
    \node[centre]   (C1) at (-1, 0) {$\mot{Y}$};
    \node[labelsat] (A1) at ( 0,-1) {$\mathbf{\overline{a}}$};
    \node[labelsat] (A2) at ( 0, 0) {$\mathbf{\overline{a}}$};
    \node[labelsat] (A3) at ( 0, 1) {$\mathbf{\overline{a}}$};
    \draw (C1) -- (A1);
    \draw (C1) -- (A2);
    \draw (C1) -- (A3);

    \node[centre]   (C2) at (2,1) {$\mot{I}$};
    \node[labelsat] (B1) at (1,1) {$\mathbf{a}$};
    \node[labelsat] (B2) at (3,1) {$\mathbf{a}$};
    \draw (C2) -- (B1);
    \draw (C2) -- (B2);

    \node[centre]   (C3) at (2,0) {$\mot{I}$};
    \node[labelsat] (D1) at (1,0) {$\mathbf{a}$};
    \node[labelsat] (D2) at (3,0) {$\mathbf{a}$};
    \draw (C3) -- (D1);
    \draw (C3) -- (D2);

    \node[centre]   (C4) at (2,-1) {$\mot{I}$};
    \node[labelsat] (E1) at (1,-1) {$\mathbf{a}$};
    \node[labelsat] (E2) at (3,-1) {$\mathbf{a}$};
    \draw (C4) -- (E1);
    \draw (C4) -- (E2);

    \node[centre]   (C5) at (5, 0) {$\mot{Y}$};
    \node[labelsat] (F1) at (4,-1) {$\mathbf{\overline{a}}$};
    \node[labelsat] (F2) at (4, 0) {$\mathbf{\overline{a}}$};
    \node[labelsat] (F3) at (4, 1) {$\mathbf{\overline{a}}$};
    \draw (C5) -- (F1);
    \draw (C5) -- (F2);
    \draw (C5) -- (F3);

        \draw (A3) -- (B1);
        \draw (A2) -- (D1);
        \draw (A1) -- (E1);
        \draw (F3) -- (B2);
        \draw (F2) -- (D2);
        \draw (F1) -- (E2);
  \end{tikzpicture}
\caption{Example of two maps of motifs based on $\setofmotifs
  = \{\mot{Y},\mot{I}\}$, the right one is saturated.}
\label{fig::mapofmotifs}
\end{figure}

\section{Heuristic Algorithm}

To solve such a problem, we decompose it into several simpler steps. Our algorithm works as follows:

\begin{itemize}
\item Generate very simple maps of motifs with efficient algorithms. Since all graphs have a spanning tree,
we generate all colored ranked trees built from a given set of motifs. To do that, classical CAT algorithms to generate all rooted ranked trees
can be adapted so that the produced trees respect all degrees and color constraints. This could be further improved by generating unrooted trees, by adapting CAT algorithms for this problem such as~\cite{li1999advantages}, but it is less clear how to do so.
\item From a tree, with many free vertices, we need to find all possible ways to pair them so that the obtained map is planar and saturated.
We call the pairing operation a \emph{fold}. It is easy to see that the order in which vertices are paired is irrelevant. Using a dynamic program, the possibility for a fold to contain a given pairing is computed in quadratic time in the number of vertices. Then, using this information, a flashlight search gives all solutions with linear delay.
\item Detect whether the generated map is isomorphic to one generated before, by comparing its signature to the signatures of the previous maps stored in a set data structure (hash map or AVL in our implementation). If the map is new, compute several indices useful for the chemist. We consider the minimum sparsity, to represent structural soundness of the molecule, which can be computed in polynomial time on planar graphs~\cite{park1993finding}.
The equivalence classes of the vertices help to understand the symmetries and the roundness of the molecule. We are also interested in the sizes of the faces of the map, which can be used as entrance inside the molecule when used as a cage.
\end{itemize}

The problem with this three steps approach is that many useless intermediate results are generated, such as trees which cannot be folded into a molecular map. 
Moreover, the same map has many different spanning trees and is thus generated by folding each of these different trees obtained at step one. While the three steps have very efficient algorithms, with strong theoretical guarantees on their complexity, their combination solves our problem, but \emph{has no bound on its delay nor its total time}.

One solution to this problem is to avoid to generate trees which cannot be folded.
A necessary condition for a map of motifs to be foldable into a saturated map is to have as many free vertices of color $a \in \alphabet$
as vertices of color $\bar{a}$. Such a map is called an \emph{almost foldable map}.
To a map of motifs $M$, we associate the statistic vector $C_M$, which counts the number of vertices of each type (negative colors are counted negatively).
An almost foldable map has an all zero vector of statistics.
It is possible to compute by dynamic programming a predicate $A(C_T,n)$, which is true if and only if there is a tree $T$
of size $n$ and statistic vector $C_T$. Using this predicate, the algorithm generating all trees can be turned into an algorithm of the same complexity generating all
almost foldable trees. This makes the algorithm run faster by several orders of magnitude, since the cost of computing $A$, while large, is done once in the preprocessing.

The second problem is that we generate graphs with bad indices which may not be interesting for the chemist, especially because the maps generated are in too large number
to be inspected manually. The trouble here is that indices of interest may change (and have changed during our work) depending on the application, the kind of instances or even our comprehension of the problem.
Moreover, there is no clear cut off $\lambda$ such that we can say, we want to generate all maps with a minimum sparsity larger than $\lambda$. 
We were not able to solve this problem, but a good approach, described in Chapter~\ref{chap:horizons}, is to enumerate only a subset of the solutions, which covers as best as possible
the set of all solutions.

Our work leaves many theoretical questions open. 

\begin{openproblem}
Can we prove that our algorithm is efficient when seen as an output sensitive algorithm, i.e. is its total time bounded by a polynomial (or even linear) in the maximal number of generated solutions.

Is the problem in $\OutputP$? Representation of planar maps by trees~\cite{DBLP:journals/tcs/AleardiDS08} may be helpful to obtain a better algorithm.
\end{openproblem}

\section{Implementing a Tool}

The aim of this work was to produce a usable tool for chemists, 
efficient enough to generate molecular maps with ten to twenty motifs.
The source code is available on \url{https://yann-strozecki.github.io/textesmaths.html} and results are available on \url{http://kekule.prism.uvsq.fr/}. Figure~\ref{fig:kekule} shows how these results are presented to the user. These results have been used by our colleague Olivier David, to synthesize several new molecules, as described in~\cite{david2020chimie},
see Figure~\ref{fig:transformation} for the interpretation of a molecular map as a real molecule.

\begin{figure}

\includegraphics[scale=0.4]{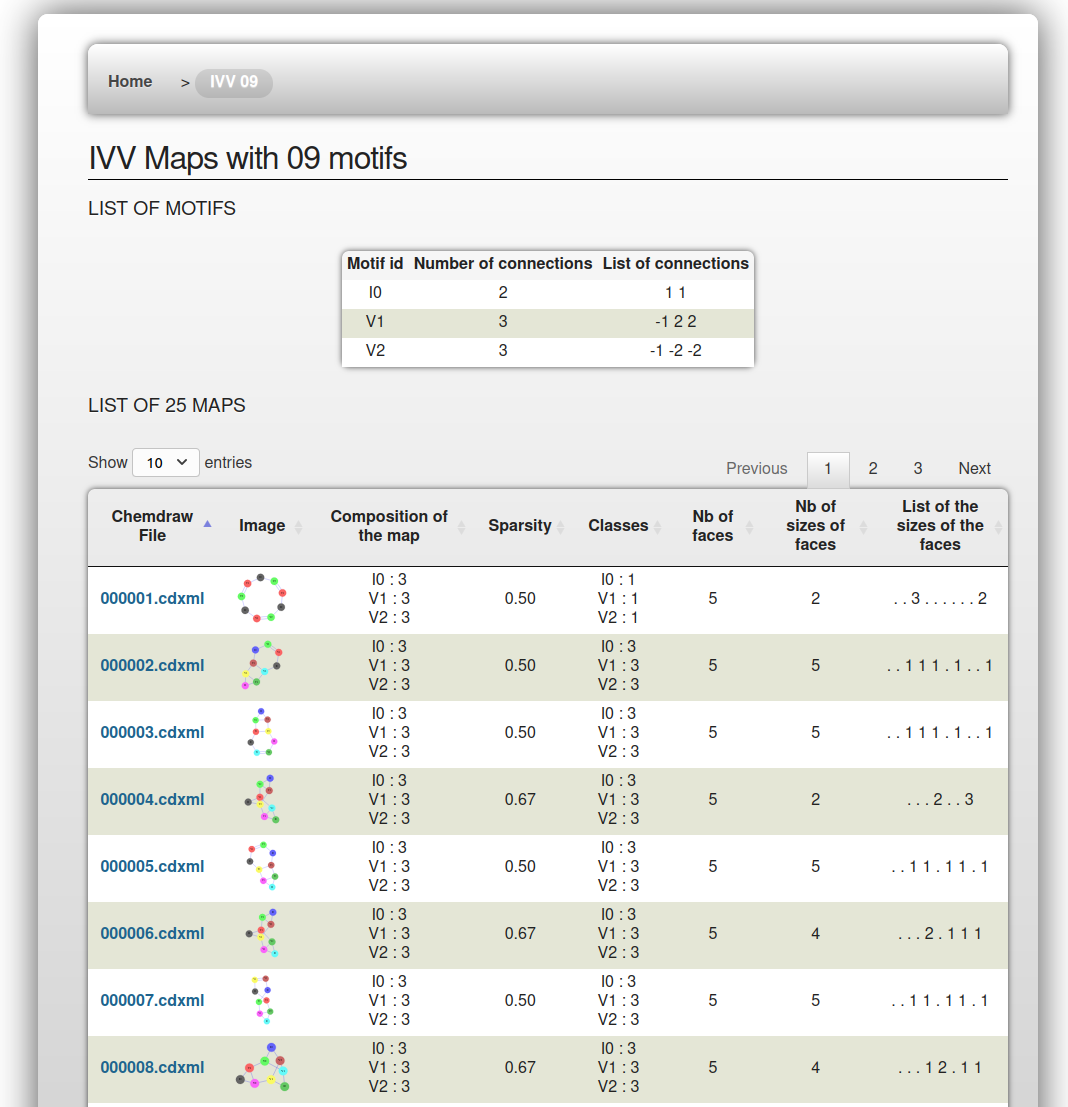}
\caption{Presentation of the results for a set of three motifs and a target size of $9$.}
\label{fig:kekule}
\end{figure}

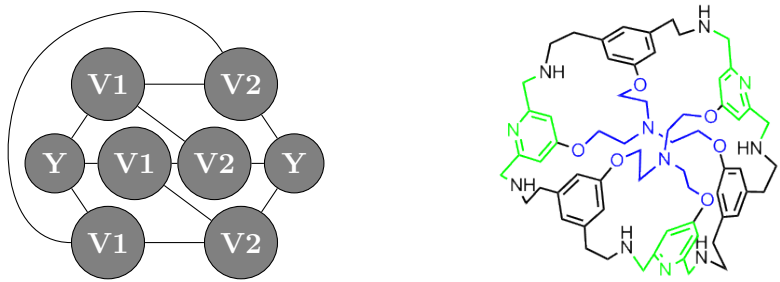
\begin{figure}
 \centering
 \begin{tikzpicture}[scale=0.70]
   \node[centre]   (Y1) at (0, 0) {$\mot{Y}$};
   \node[centre]   (V11) at (1.5,0) {$\mot{V1}$};
   \node[centre]   (V21) at (3,0) {$\mot{V2}$};
   \node[centre]   (Y2) at (4.5,0) {$\mot{Y}$};
   \node[centre]   (V12) at (1, 1.5) {$\mot{V1}$};
   \node[centre]   (V13) at (1, -1.5) {$\mot{V1}$};
   \node[centre]   (V22) at (3.5, 1.5) {$\mot{V2}$};
   \node[centre]   (V23) at (3.5, -1.5) {$\mot{V2}$};
    \draw (Y1) -- (V11);
    \draw (V11) -- (V21);
    \draw (V21) -- (Y2);
    \draw (Y1) -- (V12);
    \draw (Y1) -- (V13);
    \draw (Y2) -- (V22);
    \draw (Y2) -- (V23);
    \draw (V22) -- (V12);
    \draw (V23) -- (V13);
    \draw (V12) -- (V21);
    \draw (V11) -- (V23);
    \draw (V22) to[out= 120 ,in= 60] (-0.5,1.75) to[out= 240 ,in= 180]  (V13);
  \end{tikzpicture}
  \hspace{2cm}
\includegraphics[scale=0.4]{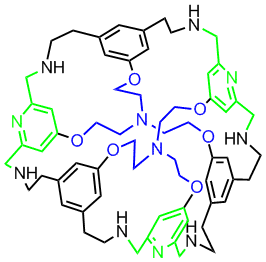}
 \caption{A molecule derived from a map generated by our algorithm.}
 \label{fig:transformation}
\end{figure}

To make our tool useful, it has been quite heavily optimized and tuned. Indeed, generating all molecular maps of size $12$ from three different motifs
of degree three already takes $50$ seconds, because it generates about $50$ million trees and $10$ million saturated maps, to get only $4476$ distinct maps.
We present three important choices made to obtain an efficient algorithm, which are common in practice but may not be natural for theoreticians. 

\begin{enumerate}
	\item \textbf{Forget exhaustiveness.} We restrict the set of enumerated solutions to decrease their number and 
	to get more efficient enumeration algorithms. However, the enumeration is not exhaustive anymore and we may skip interesting solutions. 
	To do that, instead of generating trees in the first phase, we generate either paths or cycles, that are folded like the trees in the second phase. It implies that we get only maps which have either
  a Hamiltonian path or a Hamiltonian cycle. Since small planar graphs have often such subgraphs~\cite{holton1988smallest,aldred1999cycles} it is a mild restriction and it implies that the generated maps  are connected or biconnected, which are relevant properties for a molecule. Generating paths instead of trees in the previous example yields $4463$ distinct maps instead of $4476$, but the generation time drops to less than a second.

	\item \textbf{Reconsider asymptotic complexity.} The main practical bottleneck of our algorithm, taking more than $90$ percent of the time, is the computation of 
  the signature  of each molecular map we obtain, to decide whether it has already been generated. There is a theoretical linear time algorithm to compute the signature of a map~\cite{hopcroft1974linear} and 
  an almost linear time algorithm which is implementable~\cite{hopcroft1973av}.
  However, we deal with very small maps in our work (size at most $40$), hence the complex linear time algorithms are less efficient than a simpler quadratic time algorithm: 
  from each edge of the map, do a traversal following the order of the edges, and keep the smallest traversal (sequence of vertices encountered) as a signature. 

  \item \textbf{Adapt the algorithm to the input.} Worst case complexity is not always relevant since the typical input
  may be quite different from the worst one. In our algorithm, we know that many redundant solutions are generated, hence when we compute a signature, most of the time it is already present.
  Therefore, we want to be able to decide quickly if a solution is redundant, but we have a lot of time budget for each new solution found.
  During the internship of Ruben Staub~\cite{ruben16}, we proposed an alternative signature: the set of all traversals instead of the minimal one. When a map distinct from the previously generated is found, all traversals
  are computed and stored. Each time a map is generated, a \emph{single} traversal is computed, which is enough to find whether the map has already been generated. This requires more space, but 
  speed-ups our algorithm by a factor $10$ in some cases.
\end{enumerate}
\chapter{New Horizons for Enumeration}
\label{chap:horizons}
\minitoc

There are several popular areas of complexity which have been brought to enumeration and that we have not covered. 
For instance, it is possible to define the $\FPT$ versions of the classes presented in this thesis. Enumeration of vertex covers or weighted assignments of satisfiability problems admit enumeration algorithms with an $\FPT$ delay~\cite{creignou2015parameterized,creignou2017paradigms}. Some complexity results presented here have been adapted to $\FPT$ classes~\cite{meier2020incremental}. A more restrictive way of defining kernelization for enumeration has been proposed in~\cite{golovach2021refined}. A fully polynomial enumeration kernel characterizes problems solvable in total time bounded by a polynomial in the input times a function of the parameter, while a polynomial-delay enumeration kernel charaterizes problems solvable in $\FPT$ delay. Using these definitions, lower and upper bounds for several parametrizations of the problem of enumerating maximal (minimal) matching cuts are proved.

In dynamic problems, the input is changed incrementally and some object is maintained, for instance the maximal matching of a graph~\cite{baswana2015fully}.
This idea is a good fit for query answering, where the input is huge and the cost of preprocessing is proportional to the input but usually not the enumeration phase.
Several enumeration algorithms from databases have been adapted to allow for updates of the input. It means that after modifying the input, e.g. adding or removing an edge of a graph, or changing the color of a node, the enumeration can be restarted without doing another preprocessing (or a very short one). For instance, enumerating the satisfying assignments of MSO formulas over trees can be done with a delay linear in the size of each solution and update in logarithmic time~\cite{amarilli2019enumeration}.

Enumeration problems are especially subject to the problem of combinatorial explosion because even a constant delay algorithm cannot escape the cost of going through the whole solution space, which is often too large. Moreover, even if a modern computer equipped with a good algorithm can easily generate thousands of billions of spanning trees or shortest paths, no user will ever be able to make sense of such amount of data. Hence, we should find ways to generate fewer solutions while still capturing the essence of the set of solutions.  Following ideas of several enumeration specialists (in particular several presentations by Takeaki Uno), we present promising directions, which have been explored in the last years. The main idea is, instead of listing the whole set of solutions,
to produce a \emph{compressed} version of this set, of much smaller size but still useful.

\section{Succinct Representation of the Solution Space}

Usually, when confronted to an enumeration problem, one first tries to find a representation 
of the set of solutions which minimizes its size while allowing efficient computation on it
like finding its size, its maximum, some statistic on its elements \dots It should be some middle ground 
between the set of solutions, which is easy to exploit but often too large to compute and store and the instance which represents the set of solution very succinctly but is not directly useful for computation. 

An extreme case consists in problems which admit a polynomial time sampling algorithm or which are in $\QueryP$.
In that case, the input itself (or rather the solution generator) is a good representation of the set of solutions, since they can all be accessed easily and statistical properties of the set of solutions can easily be computed.

Another example of a large set represented in a succinct way is given by algebra: a vector space is represented by a finite basis, a group by a presentation and a matroid by its fundamental cycles\dots In all these cases, it is possible to decompress efficiently the representation, that is to enumerate the elements of these objects from the succinct representation in incremental polynomial time. In Chapter~\ref{chap:limits}, we have presented the problem of computing closure under set operators (see~\cite{DBLP:journals/dmtcs/MaryS19}), which also asks to produce from a basis all elements of a set system closed under several operations.

Sometimes, partial solutions can be trivially extended into many solutions. For instance,
models of existential $FO$ formulas~\cite{durand2011enumeration} or of $k$-\textsc{DNF} formulas~\cite{capelli2020enumerating} are obtained by extending a partial model of constant size (a fixed number of variables are set) with all possible values of the remaining variables using Gray code enumeration. In this case, instead of exhaustively generating models, it is more efficient to generate only the 
partial models, as a concise representation of the set of models. This approach is relevant because the sets of solutions represented by each partial model are disjoint.
A similar but more general representation of set of models is proposed and enumerated for Horn formulas~\cite{wild2012compactly}.

\paragraph{Decomposition}

The two most simple operations on sets are the Cartesian product and the union and we have seen that they let 
enumeration classes stable in Chapter~\ref{chap:time}. Hence, it is interesting to decompose a set of solutions thanks to these two operations.
It may come naturally: the maximal cliques of a graph are the union of the maximal cliques of its connected components and the spanning forests of a graph are the Cartesian product of the spanning forests of its connected components. A constant delay algorithm to list spanning trees using a more involved Cartesian product decomposition is given in~\cite{conte_et_al:LIPIcs:2018:9666} as well as several examples of Cartesian decomposition for enumeration. 

In algebraic complexity, Cartesian product and union correspond to product and sum:
instead of representing a multivariate polynomial explicitly by a list of monomials, it can be represented by an arithmetic circuit with sum and product gates. However, this representation is hard to exploit: deciding whether the represented polynomial is zero is not believed to be in $\P$, deciding whether a monomial has a zero coefficient is $\sharp \P$-hard and counting monomials is even harder, see~\cite{fournier2015monomials}.
A more tractable representation, for which most of the previous problems are simple, is to decompose a polynomial
into product of small polynomials. We have investigated with several coauthors the task of factoring the sparse representation of a polynomial~\cite{chattopadhyay2021computing}. Interestingly, this does not always yield a smaller representation since $X^{k+1} - 1 = (X-1)\sum_{i=0}^{k} X^i$.

The previous decompositions are reminiscent of knowledge compilation~\cite{darwiche2002knowledge}, where a boolean function is represented by a succinct logical circuit (an OBDD or a DNNF). While the cost of computing the representation may be large, after compilation all queries over the original function can be evaluated efficiently over the representation. Amarilli et al.~\cite{AmarilliBJM17} give a strong polynomial delay algorithm to enumerate the models of d-DNNF circuits used in knowledge compilation. They relate it to set circuits, which can be seen as way of generating a set of solutions by using only Cartesian products and \emph{disjoint} unions of ground solutions. Several known enumeration problems such as the enumeration of the models of an MSO formula on a structure of bounded tree-width can then be reduced to the enumeration of the models of a d-DNNF. Set circuit is also essentially the same notion as factorised representation of a database~\cite{olteanu2012factorised}, used to succintly represent databases and query results. The methods of enumerating objects by first reducing the instance to compact boolean circuits is also used to obtain good practical algorithms, for instance to generate interval subgraphs~\cite{kawahara2019colorful}.

\begin{openproblem}
What enumeration problem can we capture using circuits of Cartesian products
and \emph{general} unions? Is it possible to generalize Cartesian product by some form of join without losing tractability? 
\end{openproblem}

\begin{openproblem}
What decomposition of solution sets can be useful in enumeration outside of union and Cartesian product?
\end{openproblem}

\paragraph{Equivalence Relation}

Another way to reduce the solution space is to define an equivalence relation, 
to identifiate solutions which are too similar. Then, only one representative per equivalence class is enumerated. Representatives are often harder to enumerate, a good example is the enumeration of classes of trees or graphs \emph{up to isomorphism}, see Chapter $8$ of~\cite{ruskey2003combinatorial} and the problem presented in Chapter~\ref{chap:practical}.

For problems whose solutions can be represented by AND/OR graphs (representing dynamic programs),
there is a natural notion of equivalence for which equivalence classes can be generated with polynomial delay~\cite{wang2020lazy}, with interesting applications in bioinformatics.

With Maël Guiraud and Dominique Barth, we have designed algorithms to find schedulings of periodic messages with a small latency~\cite{deterministicscheduling,phdmael}. As often with optimization problems, we enumerate exhaustively all solutions to find the best one. Therefore, to make our algorithms interesting both theoretically and practically, the set of considered schedulings must be reduced as much as possible. We use two equivalence relations and a partial order to filter the set of schedulings:

\begin{itemize}
    \item Compact representation: an equivalence class, which takes into account the combinatorial structure
    of a scheduling but abstracts away the precise timings.
    \item Canonical representation: an equivalence class, to abstract out the symmetries of a scheduling.
    \item Minimal solutions: the latency of all messages must be locally minimal, otherwise the scheduling is dominated and is not useful for computing an optimum.
\end{itemize}

\section{Approximate Representation of the Solution Space}

In the previous section  we have explored the idea of compressing the set of solutions exactly.
But we could also compress it with loss, that is producing a representation which 
cannot be used to recover all solutions but only to approximate them.

This idea has been used to replace the Pareto's frontier of a problem by a much smaller number of approximate Pareto optimal points~\cite{papadimitriou2000approximability}. The lossy compression approach is also reminiscent of sketches used to obtain a succinct representation of a very large stream, while maintaining some property with high probability. For instance, using a logarithmic space, the number of distinct elements in a stream can be evaluated~\cite{durand2003loglog}, or with a sublinear space the minimum spanning tree or maximal matching can be approximated~\cite{mcgregor2014graph}.

\paragraph{Over-approximation}

Instead of approximating each solution, we may approximate the set of solutions by a larger set,
outputting elements which are not far from being solutions. For the problem of enumerating minimal sets of size at most $k$ satisfying monotone properties such as minimal vertex covers, minimal dominating sets in bounded degree graphs, minimal feedback vertex set \dots,  efficient algorithms exist but they produce solutions of weight up to $ck$ for a constant $c$~\cite{kobayashi2020efficient}.

Query enumeration over a bounded degree structure can be done in constant delay and linear preprocessing~\cite{DBLP:journals/tocl/DurandG07}. To speed-up the preprocessing, since the database cannot be read entirely anymore, a randomized algorithm must be used, as in property testing.
There is a constant delay algorithm with \emph{polylogarithmic preprocessing} instead of linear if tuples not far from a solution for the edit distance are allowed in the enumeration~\cite{adler2021towards}.
If the proportion of tuples that are answers to the query is sufficiently large, then all answers will be enumerated, by interleaving randomly sampled tuples.
The idea of interleaving solutions, when there is an efficient algorithm to enumerate all solutions and that almost all potential solutions are indeed solutions has already been proposed by Leslie Ann Goldberg in a deterministic setting. She proved that all classes of graphs defined by almost sure first order property can be enumerated with polynomial delay~\cite{Goldberg91}.

Many enumeration problems ask to enumerate minimal or maximal objects, such as maximal cliques or minimal hitting sets,
to avoid less useful solutions. However, these problems are much harder and are sometimes solved by generating all non extremal objects. 
For instance, to generate minimal hitting set, the most efficient algorithms are based on the generation of all hitting sets,
using pruning rules to avoid as most non minimal hitting set as possible~\cite{murakami2014efficient}.

Finally, a more restricted way to do an over-approximation is to allow repetitions of a solution
instead of allowing elements which are not solutions. It is our approach when turning a random generator into an exhaustive enumeration algorithm using only polynomial space, as explained in Chapter~\ref{chap:space}.

\paragraph{Under-approximation}

In the previous section, the approximation is done by outputting more elements than solutions, which goes against the idea of reducing the solution space
and may thus not be useful in practice. To approximate a set of solutions $A(x)$, we can use some subset $S$ such that $|S|$ is within some factor of $|A(x)|$.
I am not aware of any result of this kind, however some algorithms use the property that a large subset of solutions is easy to enumerate to have enough time to enumerate the rest. Such algorithms are used for the enumeration of clause sequences of a $k$-\textsc{CNF} formula~\cite{berczi2021generating} and for the enumeration of the models of a $k$-\textsc{DNF} formula~\cite{capelli2020enumerating}.

\begin{openproblem}
Design a constant factor under-approximation in incremental polynomial time for a non artificial enumeration problem not known to be in $\IncP$.    
\end{openproblem}

When designing an under-approximation, the choice of solutions in the approximation is critical. 
We would like them to be more interesting than the ones not enumerated. When there is a relevant order on solutions, enumerating in this order, as in Chapter~\ref{chap:order}, can be seen as a way to do meaningful under-approximation.

To be more general, we could allow both to enumerate elements which are not solutions and
to miss some solutions. Let $S$ be the set enumerated to approximate the solution set $A(x)$.
To control the quality of the approximation, we say that $S$ is an $\epsilon$-approximation if 
$|S \triangle A(x)| \leq \epsilon |A(x)|$. This definition is similar to the notion of approximation
in knowledge compilation, for which several lower bounds are known~\cite{DBLP:conf/ijcai/ColnetM20}.

\paragraph{Diversity}

The approximation based on the number of solutions may not be relevant: what if we get
a large family of similar solutions but miss all the diversity of the solution set? 

Most of the time, there is no notion of quality over the solutions or this notion is too vague to formalize. 
Think about a system recommanding routes between two locations, each user may have specific desires such 
as only using streets she already knows or constraints too informal to be implemented such as prefering roads with a beautiful scenery.
Hence, such a system should propose solutions as \emph{diverse} as possible, to maximize the chance that one fulfills the unknown constraints of each user.

 To this aim, we could equip the solutions with some distance and try to cover them all
 by as few of them as possible.  Given an instance $x$ of a problem $\enum{A}$ and an integer $d$, we define a $d$-cover of $A(x)$ as $S \subseteq A(x)$ such that for all solutions of $A(x)$, there is a solution in $S$ at distance at most $d$. An enumeration problem is $d$-approximable in polynomial time when $S$ can be generated in total time polynomial in the size of the smallest $d$-cover of $A(x)$.
 Allowing randomization to be able to sample the solutions could be important in this context,
 since we have much less time than solutions. A $d$-approximation algorithm could be extremely useful for a user-interactive system, where the user sets the value $d$, obtains first a coarse representation of the solutions and can then zoom on an interesting solution by restricting the problem around it and asking for an approximation with a smaller $d$.
The notion of $d$-approximability is designed so that any solution is close to some solution in the cover.
Hence, among the cover, there must be many solutions which are quite different, which hopefully makes such a cover useful in practice.
 
\begin{openproblem}
Give a $d$-approximation algorithm for a problem with a structured set of solutions, such as the minimal
spanning trees or the shortest paths, for any relevant distance over solutions.
\end{openproblem}

\begin{openproblem}
Assume we have a supergraph of solutions, and the distance over solutions is related to the distance in the supergraph.
What random walk on the supergraph and which property of the supergraph yield a $d$-approximation? Can we use the same approach with flashlight search?
\end{openproblem}

 A similar approach has been recently proposed through the notion of \emph{diversity}, which is the sum of the Hamming distances of pairs of elements in a set.
  A family of hard problems, such as vertex cover, admit $\FPT$ algorithms for producing a set of solutions of high enough diversity~\cite{DBLP:conf/ijcai/BasteFJMOPR20}. The method converts any tree-decomposition based dynamic programming algorithm into a dynamic programming algorithm for the same problem with an additional diversity constraint. However, the method is exponential in the number of produced solutions in addition to the treewidth of the graph.  To find $r$ maximally diverse spanning trees, there is a polynomial time algorithm in $r$~\cite{hanaka2021finding}, which allows to produce a large set of diverse solutions.

 Diversity seems to well capture the notion of producing different solutions, however it does not garantee to cover well the set of solutions,
 see Figure~\ref{fig:diversity}. We would like to relate large set of diverse solutions to the notion of $d$-approximability. To that aim, it 
 is interesting to replace the sum of the distances between solutions by the minimum of the distances in the definition of diversity.

\begin{openproblem}
Generate $r$ spanning trees, while maximising their minimal Hamming distances, in time polynomial in $r$.
\end{openproblem}

\begin{figure}
\begin{center}
\includegraphics{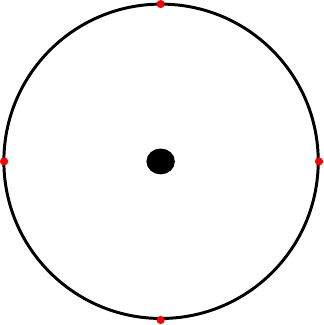}
\end{center}
\caption{Solutions on a circle and around its center. A subset of solutions (here in red) maximising diversity always lies on the circle, missing solutions at the center.}
\label{fig:diversity}
\end{figure}
\appendix

\bibliographystyle{ThesisStyleWithEtAl}
\bibliography{thesis}












\end{document}